\documentclass[aps.pra,onecolumn,showpacs,preprintnumbers,amsmath,amssymb]{revtex4-2}
    \usepackage[utf8]{inputenc}
    \usepackage[english]{babel}
    \newtheorem{theorem}{Theorem}
    \newtheorem{corollary}{Corollary}
    \newtheorem{lemma}{Lemma}
    \newtheorem{definition}{Definition}
    \newtheorem{remark}{Remark}
    \usepackage{setspace}
    \usepackage{booktabs}
    \usepackage{multirow}
    \usepackage{amsmath,amssymb}
    \usepackage{qcircuit}
    \usepackage{caption}
    \usepackage{lipsum}
    \usepackage{tabularx}
    \usepackage{natbib}
    \usepackage{stfloats}
    \usepackage{graphicx}
    \usepackage{rotating}
    \usepackage{diagbox}
    \usepackage{makecell}
    \setcellgapes{2.3pt}
    \usepackage{subcaption} 
    \newenvironment{proof}{\textit{Proof.}}{\hfill$\square$}
    \usepackage{physics}
    \usepackage[figurename=Fig., justification=RaggedRight]{caption}
    \usepackage{amsmath}
    \usepackage{xcolor}
    \usepackage{hyperref}

    \usepackage{adjustbox,expl3,etoolbox}
    \letcs\replicate{prg_replicate:nn}
    
\begin{document} 
    
    \title{Efficient explicit gate construction of block-encoding for Hamiltonians needed for simulating partial differential equations}
    
\author{Nikita Guseynov}
\email{guseynov.nm@gmail.com}
\affiliation{University of Michigan-Shanghai Jiao Tong University Joint Institute, Shanghai 200240, China.}

\author{Xiajie Huang}
\email{xj.huang@sjtu.edu.cn}
\affiliation{School of Mathematical Sciences, Shanghai Jiao Tong University, Shanghai 200240, China}
\affiliation{Shanghai Artificial Intelligence Laboratory, Shanghai, China}

\author{Nana Liu}
\email{nana.liu@quantumlah.org}
\affiliation{Institute of Natural Sciences, School of Mathematical Sciences, Shanghai Jiao Tong University, Shanghai 200240, China}
\affiliation{Ministry of Education Key Laboratory in Scientific and Engineering Computing, Shanghai Jiao Tong University, Shanghai 200240, China}
\affiliation{Shanghai Artificial Intelligence Laboratory, Shanghai, China}
\affiliation{University of Michigan-Shanghai Jiao Tong University Joint Institute, Shanghai 200240, China.}

    \begin{abstract}
    
    One of the most promising applications of quantum computers is solving partial differential equations (PDEs). By using the Schrödingerisation technique — which converts non-conservative PDEs into Schrödinger equations -- the problem can be reduced to Hamiltonian simulations. The particular class of Hamiltonians we consider is shown to be sufficient for simulating almost any linear PDE. In particular, these Hamiltonians consist of discretizations of polynomial products and sums of position and momentum operators. This paper addresses an important gap by efficiently loading these Hamiltonians into the quantum computer through block-encoding. The construction is explicit and efficient in terms of one- and two-qubit operations, forming a fundamental building block for constructing the unitary evolution operator for that class of Hamiltonians. The proposed algorithm demonstrates a squared logarithmic scaling with respect to the spatial partitioning size, offering a polynomial speedup over classical finite-difference methods in the context of spatial partitioning for PDE solving. Furthermore, the algorithm is extended to the multi-dimensional case, achieving an exponential acceleration with respect to the number of dimensions, alleviating the curse of dimensionality problem. This work provides an essential foundation for developing explicit and efficient quantum circuits for PDEs, Hamiltonian simulations, and ground state and thermal state preparation.
    
    \end{abstract}
    
    \maketitle
    
    \section{Introduction}
    
    Quantum computation \cite{FEYNMAN} is an emerging technology that has the potential to tackle certain computational problems intractable on classical devices. The most influential algorithms include Shor's factorization algorithm \cite{math_apply,Shor}, Grover's search algorithm \cite{grover,q_search}, simulation of quantum chemistry models \cite{aspuru2005simulated}, quantum algorithms based on random walks \cite{szegedy2004quantum}, and the quantum linear systems (HHL) algorithm \cite{harrow2009quantum}. Other important applications include quantum simulation for partial differential equations, which form the backbone of mathematical modeling for various physical phenomena and often prove intractable for classical computers as the dimensionality of the problem and the size of discretisation increase \cite{laughlin2000theory}. Quantum algorithms have the potential to improve scaling in both the dimension and the size of the discretisation used. Thus, these algorithms have the potential to improve diverse areas in physics \cite{costa2019quantum,gaitan2020finding}, engineering \cite{jin2023time,linden2022quantum,jin2022quantum, sato2024hamiltonian}, and finance \cite{stamatopoulos2020option,gonzalez2023efficient}. 
    
    Despite significant progress in quantum simulation and advancements in applying these techniques to partial differential equations, few works demonstrate explicit gate construction. This step is crucial for treating quantum computation as a viable future tool and for implementation on fault-tolerant quantum computers. Many algorithms rely on oracles—black boxes that are simply assumed to exist, with no further inquiry. However, measuring efficiency solely by the number of oracle calls does not accurately reflect the total implementation cost, as the construction cost of these black boxes often remains unspecified. If the complexity of constructing these black boxes scales non-polynomially with the number of qubits, any quantum advantage may be negated. This issue highlights a significant gap in the current literature.
    
    In this paper, our key contribution is providing such an efficient explicit gate construction for a block-encoding of a particular class of multimode Hamiltonians with the schematic form $\hat{H}=\sum \alpha \hat{x}^q\hat{p}^m+\alpha^*\hat{p}^m\hat{x}^q$ for positive integers $q,m$, where $\hat{x}$ and $\hat{p}$ are respectively position and momentum operators. The block-encoding of $\hat{H}$ can be used as a fundamental building block for near-optimal quantum simulation and provide a block-encoding of the evolution operator $\exp(-i\hat{H}t)$, with optimal scaling in time $\mathcal{O}(N^mt)$ \cite{gilyen2019quantum}. Using the Schr\"odingerisation technique~\cite{jin2022quantum,jin2023quantum,jin2023analog,cao2023quantum}, 
    we demonstrate that this class of Hamiltonians plays a pivotal role 
    in simulating any linear partial differential equation with polynomial 
    coefficient functions. Compared to the best-known classical methods, 
    our approach achieves a linear acceleration in terms of spatial partitioning, 
    scaling as \(N^m\) instead of \(N^{m+1}\). Furthermore, in the multi-dimensional 
    case, we attain an exponential acceleration with respect to the number 
    of dimensions \(d\), alleviating the curse of dimensionality problem; see Section~\ref{Sec:Generalisation} for details. It is also potentially important for quantum problems like the quantum simulation of $\hat{H}$, as well as ground state and thermal state preparation. 
    
    There are many other approaches for Hamiltonian simulations besides using block-encoding. For example, the quantum random walk combined with quantum signal processing has been recently developed to address sparse Hamiltonian simulations \cite{low2017optimal}, the algorithm demonstrated optimal scaling on all parameters. However, their direct application is hindered by the absence of a method to construct a digital quantum oracle query access to the Hamiltonian, which is a crucial component for constructing the quantum walk isometry operator. Similarly, another approach exploiting fractional and continuous-query models for efficient sparse Hamiltonian simulations \cite{berry2014exponential} encounters the same obstacle as it too requires digital quantum oracle query access to the Hamiltonian under consideration.
    
     Other literature \cite{sato2024hamiltonian,hu2024quantum} attempt to construct optimal circuits via decomposition of $\hat{H}$ into $2\times 2$ ladder operators—an alternative to the well-established Pauli basis \cite{nielsen2002quantum}. These methods are suited to toy models with separate terms for momentum and coordinate quantum operators, not their multiplication. Additionally, the use of the Suzuki-Trotter decomposition \cite{Suzuki,Trotter} results in poor scaling with time $\mathcal{O}(t^2)$, in contrast to the optimal $\mathcal{O}(t)$.
    
    We introduce a general framework for block-encoding construction, which accepts any Hamiltonian composed of superpositions of terms $\hat{x}^q\hat{p}^m$ and $\hat{p}^m\hat{x}^q$. Our proposed algorithm exhibits optimal (polylogarithmic) $\mathcal{O}(n^2)$ scaling in the number of finite-difference partitionings $2^n$ for spatial dimensions and enjoys linear scaling with time $\mathcal{O}(t)$, as we construct the evolution operator for extensive periods. 

        Inspired by the explicitness of Shor’s factoring algorithm, which is the  
    first to theoretically offer robust quantum speedup, we present explicit  
    quantum circuits and estimate their complexity in terms of single-qubit  
    operations and the simplest entangling operations (C-NOT) \cite{nielsen2002quantum}, with all  
    the constants included. Furthermore, we provide a systematic guide,  
    showing step-by-step the necessary quantum circuits. One could think that  
    the foundational role of the T gate simplifies its usage in fault-tolerant  
    quantum computing. However, even a simple one-qubit \( R_z \) gate can  
    require a varying number of T gates, heavily dependent on the gate’s  
    parameters, especially when using the Solovay-Kitaev algorithm \cite{dawson2005solovay}. Thus, we  
    rely on C-NOT and one-qubit gates as fundamental components, given their  
    consistency and simplicity in gate construction.

    The organization of this paper is as follows. In Section~\ref{sec:Motivation, roadmap and preliminary material}, we begin with providing the main results and review how they are useful in Hamiltonian simulations. We also motivate in more detail why we choose to construct the block-encoding of the $\hat{H}$ considered in this paper -- explicit gate construction for the block-encoding for arbitrary $\hat{H}$ simulation is an open problem and too difficult for now. The areas of applications in quantum problems and the simulation of partial differential equations. The overall algorithm contains of many relatively small steps. We provide a road-map of our algorithm in Fig.~\ref{fig:road_map} in Section~\ref{subsec:Roadmap}, making it easier to follow. This can be consulted during the reading of the whole paper. The following sections then show definitions and step-by-step implementation of each step in the road-map. In Section~\ref{sec:preliminaries}, we introduce essential mathematical objects needed. The main results are in Section~\ref{main_section}, which details the process on how to implement nearly optimal block-encoding for the single-mode $\hat{H}$. This culminates with the construction an efficient evolution operator for the Hamiltonian (Theorem~\ref{theorem:Nearly-optimal Hamiltonian simulation}), providing the exact number of required gates. We then dedicate Section~\ref{Sec:Generalisation} to an important generalization to the multi-dimensional $\hat{H}$. The Section~\ref{sec:conclusion} concludes the paper.

    To clarify our exposition, we include Tables \ref{table_of_notations}, \ref{table_scaling}, \ref{table_of_notations_generalization_section}, \ref{table_scaling_multi-dimensional case} in the Section~\ref{sec:tables}. The first and the third tables explain all the notation introduced and indicates their usage for $1$- and multi-dimensional cases respectively. The second and the fourth tables summarizes all milestones from Fig.~\ref{fig:road_map} and details the complexity of each step in the implementation for the both cases. We hope that the tables will prove convenient for our readers.

\section{Results, Motivation, roadmap}\label{sec:Motivation, roadmap and preliminary material}

We begin with the following Hamiltonian as a sum of $\eta$ terms
\begin{equation}
        H=\sum_{k=0}^{\eta-1}\left(\alpha_{k}P_{q_k}(x)p^{m_k}+\alpha^*_{k}p^{m_k}P_{q_k}(x)\right);\qquad \alpha_{k}\in\mathbb{C};\qquad q_k,m_k\in\mathbb{N},
        \label{eq:main_Hamiltonian}
    \end{equation}
where $x$ and $p$ denote coordinate and momentum operators, respectively, and the symbol $*$ represents the complex conjugation operation; $P_{q_k}(x)$ is a degree-$q_k$ polynomial satisfying that 
\begin{itemize}
    \item $P_{q_k}(x)$ has parity-($q_k$ mod 2) which means that it has only odd or even degrees of $x$.
    \item $P_{q_k}(x)$ is a real polynomial which means that all the coefficients are real.
    \item for all $x\in[-1,1]$: $\abs{P_{q_k}}<1$.
\end{itemize}
We note that this Hamiltonian is just a more sophisticated view of already mentioned model
\begin{equation}
        H=\sum_{k=0}^{\eta-1}\left(\alpha^\prime_{k}x^{q_k}p^{m_k}+\alpha^{\prime*}_{k}p^{m_k}x^{q_k}\right);\qquad \alpha^\prime_{k}\in\mathbb{C};\qquad q_k,m_k\in\mathbb{N}.
        \label{eq:simple_main_Hamiltonian}
\end{equation}

To simulate this Hamiltonian numerically (on both classical and quantum devices), finite-difference schemes are commonly employed \cite{zienkiewicz2005finite,ozicsik2017finite}.

The main contribution in this paper is to provide explicit circuit constructions for the block-encoding of the discretised form of the Hamiltonian in Eq.~\eqref{eq:main_Hamiltonian}. It is now a $2^n\times2^n$ matrix of the form 
\begin{equation}
        \hat{H}=\sum_{k=0}^{\eta-1}\left(\alpha_{k}P_{q_k}(\hat{x})\hat{p}^{m_k}+\alpha^*_{k}\hat{p}^{m_k}P_{q_k}(\hat{x})\right);\qquad \alpha_{k}\in\mathbb{C};\qquad q_k,m_k\in\mathbb{N},
        \label{eq:numeric_main_Hamiltonian}
    \end{equation}
    with $\hat{x}$ and $\hat{p}$ being finite difference matrices for the coordinate and momentum operators. 
For example, we can introduce the central symmetrical finite-difference scheme for derivative $\frac{\partial U}{\partial x}\rightarrow \frac{U_{h+1}-U_{h-1}}{2\Delta x}$; and we impose periodic spatial boundary conditions over the interval $(a,b)$:
    \begin{eqnarray}
    \hat{x}=\left(\begin{array}{ccccc}
    a & 0  & \dotsm & 0 & 0 \\
    0 & a+\Delta x  &\dotsm& 0 & 0\\
    \rotatebox[origin=c]{270}{\dots}&&\rotatebox[origin=c]{-45}{\dots}&&\rotatebox[origin=c]{270}{\dots}\\ 
    0 & 0 &\dotsm & b-\Delta x &0\\
    0 & 0 &\dotsm & 0 &b\\
    \end{array}
    \right);\hat{p}=-\frac{i}{2\Delta x}\left(\begin{array}{ccccccc}
    0 & 1 & 0& \dotsm &0& 0 & -1 \\
    -1 & 0 & 1&\dotsm& 0&0 & 0\\
    \rotatebox[origin=c]{270}{\dots}&&&\rotatebox[origin=c]{-45}{\dots}&&&\rotatebox[origin=c]{270}{\dots}\\ 
    0 & 0&0 &\dotsm &-1& 0 &1\\
    1 & 0&0 &\dotsm &0& -1 &0\\
    \end{array}
    \right).
    \label{eq:x_p_cyclic_matrices}
    \end{eqnarray}
However, we do not restrict our consideration for this view of operators only; in the Section~\ref{Sec:Generalisation} we consider more general boundaries as well general finite-difference schemes. From the view of the Hamiltonian it is clear that Without loss of generality we can impose that $-1\leq a<b\leq 1$. We use this condition later on in the paper. Moreover, we impose that $\sum_{\kappa=0}^{2^n-1} (a+\kappa \Delta x)^2=1$. Both those assumption are introduced to simplify normalization constants, making important parts more vivid.

Block-encoding of an operator is defined in the following way. 
   \begin{definition}[Block-encoding(modified Definition 43 from \cite{gilyen2019quantum})]
        Suppose that $A$ is an $n$-qubit operator, $\alpha,\epsilon\in\mathbb{R}_+$; $a,s\in\mathbb{Z}_+$, and let $(a+s+n)$-qubit unitary $U$ be so that for any arbitrary quantum state $\ket{\psi}^n$
        \[ U\ket{0}^a\ket{0}^s\ket{\psi}^n=\ket{0}^a\ket{\phi}^{n+s};\qquad (\bra{0}^s\otimes I^{\otimes n})\ket{\phi}^{n+s}=\tilde{A}\ket{\psi}^n,\]
            then we say that $U$ is an $(\alpha,s,\epsilon)$-block-encoding of $A$, if
        \[ ||A-\alpha \tilde{A}||\leq\epsilon.\]
        Here we underline that the unitary $U$ can exploit some auxiliary qubits setting them back to zero-state. Later we call such qubits  `pure ancillas'. The general scheme of block-encoding is presented in Fig.~\ref{fig:block_encoding_general}.
    \end{definition}
    
This paper's central contributions are encapsulated in the following two Theorems, which establish a framework for efficient quantum simulation methods and demonstrate the scalability of the proposed approach. Later in the paper we provide the detailed version of them, see Theorems~\ref{theorem:Hamiltonian Block-encoding},~\ref{theorem:Multi-dimensional block encoding}.

\begin{theorem}[Hamiltonian Block-encoding]\label{theorem:Hamiltonian Block-encoding_intro}
    Let $H$ be a Hamiltonian as in Eq.~\ref{eq:main_Hamiltonian} with the number of terms $\eta=2^\gamma$, $m = \max\limits_{k}m_k$, $2^l=2m+1$, $q=\max\limits_kq_k$, and with matrix representation $\hat{H}$ as in Eq.~\ref{eq:numeric_main_Hamiltonian} with normalization constant $\mathcal{N}_H$ (see Corollary~\ref{corollary:Hamiltonian block-encoding auxiliary}) then we can implement a $(\sqrt{2^{l}}\mathcal{N}_H,l+\gamma+4,0)$-block-encoding of $2^n\times2^n$ matrix $\hat{H}$
    \[U_H\ket{0}^{l+4}\ket{\phi}^n=\frac{1}{\sqrt{2^l}\mathcal{N}_H}\sum\limits_{\substack{i=0,\dots,2^n-1\\ j\in F^H_i}}\left(\hat{H}_{ij}\sigma_j\ket{0}^{l+4}\ket{i}^n
        +J^U_{\hat{H}_{ij}}\ket{\bot_{0}}^{n+l+4}\right);\qquad \ket{\phi}^n:=\sum_{j=0}^{2^n-1}\sigma_j\ket{j}^n,
    \]
    with resources no greater than:
    \begin{enumerate}
        \item $\mathcal{O}\bigl(\eta q n^2 + \eta mn\bigr)$ quantum gates;
        \item $3n+\gamma-l-1$ pure ancillas.
    \end{enumerate}
    \end{theorem}

\begin{remark}
    The normalization constant, 
\(\mathcal{N}_H \sim 2^{mn} \eta x^q_{\max}\), 
should not be considered related to the probability of 
successfully applying the block-encoded matrix. Instead, it serves 
solely as a normalization factor for the resulting wave function 
after applying \(U_H\). 

The total probability of success for this operation is 
\(\mathcal{O}(1/m)\). This implies that with probability 
\(\mathcal{O}(1/m)\), we achieve the action of:
\[
\frac{H \ket{\psi}}{|H\ket{\psi}|}.
\]

However, the normalization constant \(\mathcal{N}_H\) becomes significant in 
Theorem \ref{theorem: optimal block-hamiltonian simultaions}, where it plays 
an essential role. Consequently, the total gate count will necessarily 
include a factor of \(2^{mn}\), limiting the quantum advantage in 
the number of partitions \(N\) to a polynomial scale, as the best known classical algorithm has complexity $2^{(m+1)n}$.
\end{remark}

We also extend our construction from Eq.\eqref{eq:main_Hamiltonian} to multimode Hamiltonians of the form 
 \begin{eqnarray}
    \begin{gathered}
        H^{(d)}=\sum_{k=0}^{\eta_d-1}\alpha_{k}\prod_{y=0}^{d-1}\otimes R(L_{ky},P_{q^{(y)}_k}(x_y),p_y^{m^{(y)}_k});\\
        R(0,P_{q^{(y)}_k}(x_y),p_y^{m^{(y)}_k})=P_{q^{(y)}_k}(x_y)p_y^{m^{(y)}_k};\quad R(1,P_{q^{(y)}_k}(x_y),p_y^{m^{(y)}_k})=p_y^{m^{(y)}_k}P_{q^{(y)}_k}(x_y),
        \label{eq:multi-dimensional_Hamiltonian_first}
        \end{gathered}
    \end{eqnarray}
where $L_{ky}$ is a binary number indicating the ordering of the
operators; we show how to construct the block-encoding of the discretisation of this $d$-mode Hamiltonian (see Section~\ref{Sec:Generalisation}).

\begin{theorem}[Multi-dimensional block-encoding]\label{theorem_intro:Multi-dimensional block encoding}
        Let $H^{(d)}$ be a Hamiltonian as in Eq.~\ref{eq:multi-dimensional_Hamiltonian_first} with the number of terms $\eta_d=2^\gamma$, sparsity $2^l$, $n=\sum_{y}n_y$-qubit matrix representation $\hat{H}^{(d)}$, and normalization constant $\mathcal{N}^{(d)}_H$ (see Corollary~\ref{corollary:multi-dimensional Hamiltonian superposition}) then we can implement a $(\sqrt{2^{l}}\mathcal{N}^{(d)}_H,l+\gamma+3d,0)$-block-encoding of $\hat{H}$ with 
        \begin{enumerate}
            \item $\mathcal{O}\left(\sum_{k=0}^{\eta_d-1}\sum_{y=0}^{d-1}(q_k^{(y)}n_y^2+m^{(y)}_kn_y)\right)$ quantum gates,
            \item $\mathcal{O}(\sum_{y=0}^dn_y+\log(\eta))$ pure ancillas.
        \end{enumerate}
    \end{theorem}

Classically, solving such PDEs with partitioning \( N = 2^n \)  
requires at least \( O(2^n) \) operations, while our quantum approach  
achieves \( O(n^2) \) scaling for the same setup. This remarkable  
difference opens the door to quantum speedup, enabling the handling  
of complex problems beyond the reach of classical computing.

\noindent \textbf{Main motivating questions.} For those not familiar with this area, we need to address two questions to motivate our work: 
\begin{enumerate}
    \item What is block-encoding and why is its explicit circuit construction for a Hamiltonian important?
    \item Why is our particular Hamiltonian $\hat{H}$ in Eq.~\eqref{eq:numeric_main_Hamiltonian} relevant to many problems? 
\end{enumerate}
 \begin{figure}[h]
    \includegraphics[width=0.4\textwidth]{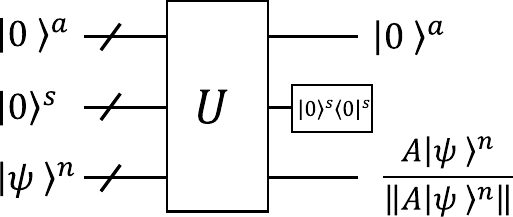}
    \caption{The general scheme of block-encoding for matrix $A$. The latest operation for the second register indicates measurement of zero, which project main $n$-qubit register onto the desired state.}
    \label{fig:block_encoding_general}
    \end{figure} 

    From this definition, it is clear that if the block-encoding of a matrix can be constructed, the matrix itself can be recovered. The main advantage of dilating the matrix to its block-encoding is so that, given access to the block-encoding, sums of matrices and polynomial approximations  of positive and negative power functions of matrices can often be efficiently constructed. Thus, block-encoding can be considered a fundamental building block from which more complex functions of a matrix can be constructed. 

    There is one particular function of $\hat{H}$ that is of particular interest and this is the unitary $\exp(i \hat{H}t)$. This is exactly the unitary evolution operator required for Hamiltonian simulation with respect to a time-independent Hamiltonian $\hat{H}$. It turns out (Theorem~\ref{theorem: optimal block-hamiltonian simultaions}) that by having the access to block-encoding of $\hat{H}$, one can construct a near-optimal implementation of $\exp(i \hat{H}t)$ with respect to time $t$ (linear in $t$) and has an additive logarithmic dependence on inverse error $\epsilon$.   

\begin{theorem}[Optimal block-Hamiltonian simulation(Theorem 58 from \cite{gilyen2019quantum})]\label{theorem: optimal block-hamiltonian simultaions}
    Let $t\in\mathbb{R}\slash \{0\}$, $\epsilon\in (0,1)$ and let $U_H$ be an $(\alpha,a,0)$-block-encoding of the $2^n\times2^n$ Hamiltonian H. Then we can implement an $\epsilon$-precise Hamiltonian simulation unitary $V$ which is a $(2,a+2,\epsilon)$-block-encoding of $e^{itH}$, with resources no greater than
    \begin{enumerate}
        \item $\Omega (t,\epsilon)$ uses of $U_H$ or its inverse;
        \item $1$ uses of controlled-U or its inverse;
        \item $\Omega(\alpha t,\epsilon)(16a+50)+4$ one-qubit gates;
        \item $\Omega(\alpha t,\epsilon)(12a+38)$ C-NOTs;
        \item $\mathcal{O}(n)$ ancilla qubits.
    \end{enumerate}
    $\Omega(\alpha t,\epsilon)$ is implicitly defined through inequality for the truncation parameter $g\geq \Omega(\alpha t,\epsilon)$
    
    \[ \frac{1.07}{\sqrt{g}}\left(\frac{\alpha et}{2g}\right)^g\leq\epsilon. \]
    The inequality determines scaling of the algorithm (Lemma $59$ from \cite{gilyen2019quantum})
    \[ \Omega(\alpha t,\epsilon)=\mathcal{O}\left(\alpha t+\frac{\ln(1/\epsilon)}{\ln(e+\frac{\ln(1/\epsilon)}{\alpha t})}\right) \]
    \end{theorem}
Since all Hamiltonian simulation problems are concerned with simulating the unitary evolution operator for a given Hamiltonian $\hat{H}$, having the block-encoding of $\hat{H}$ is a necessary starting point if we choose to use the block-encoding formalism for Hamiltonian simulation. Often in most applications of block-encoding, this is assumed to be provided, like a black-box, and the total cost for an algorithm only counts the total number of queries to this black-box. However, for a real implementation of the quantum simulation scheme, it is absolutely necessary to explicitly construct this black-box. The total number of queries of this black-box for an algorithm, even if efficient, is not a meaningful measure when considering the total cost for real implementation, unless we know that this black-box itself is also efficient to construct.  

However, explicitly constructing this block-encoding for arbitrary $\hat{H}$ is very difficult and is an open problem. In this paper, we focus on a particular class of $\hat{H}$ -- the discretisations of Eq.~\eqref{eq:main_Hamiltonian} and Eq.~\eqref{eq:multi-dimensional_Hamiltonian_first}. We briefly discuss below some examples of applications including (a) time-independent and time-independent Hamiltonian simulation (b) the quantum simulation of all linear partial differential equations (and some nonlinear problems too) and (c) ground state and thermal state preparation of $\hat{H}$. The detailed implementation of these schemes we leave for future work. \\

\noindent \textbf{Hamiltonian simulation.} 
Schr\"odinger's equation for $\eta$ particles in $d=3$ dimensions with time-independent interacting potential $V(\hat{x}_0, \cdots, \hat{x}_{\eta-1}, \hat{y}_0, \cdots, \hat{y}_{\eta-1}, \hat{z}_0, \cdots, \hat{z}_{\eta-1})$ has the corresponding Hamiltonian 
\begin{align}
    \hat{H}=\sum_{k=0}^{\eta-1} \frac{\hat{p}^2_{x_k}}{2m_k}+\frac{\hat{p}^2_{y_k}}{2m_k}+\frac{\hat{p}^2_{z_k}}{2m_k}+V(\hat{x}_0, \cdots, \hat{x}_{\eta-1}, \hat{y}_0, \cdots, \hat{y}_{\eta-1}, \hat{z}_0, \cdots, \hat{z}_{\eta-1}).
\end{align}
Clearly the kinetic term is of the form in Eq.~\ref{eq:multi-dimensional_Hamiltonian_first}. The potential belongs falls into Eq.~\eqref{eq:multi-dimensional_Hamiltonian_first} only if it is polynomial in $x, y, z$, for instance a sum of harmonic wells. Polynomial potential functions can also be approximate representations in Taylor expansions of more complex potentials.  

We can also easily extend to the case of time-dependent potentials $V$. We can use a simple transformation \cite{cao2023quantum} turning a time-dependent Hamiltonian $H(t)$ into a time-independent Hamiltonian $\hat{H}$ with one extra mode labelled $s$
\begin{align} \label{eq:time-independent}
    H(t) \rightarrow \hat{H}=\hat{p}_s \otimes \mathbf{1}+H(\hat{x}_s)
\end{align}
where $[\hat{x}_s,\hat{p}_s]=i\mathbf{1}_s$. From the method in \cite{cao2023quantum}, if time-independent quantum simulation of $\hat{H}$ is performed, one can easily and efficiently recover the evolution with respect to $H(t)$. Thus, almost for any $H(t)$ we can claim that Eq.~\eqref{eq:time-independent} is also of the form in Eq.~\eqref{eq:multi-dimensional_Hamiltonian_first}.  

\noindent \textbf{Quantum simulation of partial differential equations.}
It is known that for any $d$-dimensional linear partial differential equation for $u(x,t)$, $x=x_1,...,x_d$, can be written in the form \cite{jin2023analog}
\begin{align} \label{eq:uode}
    \frac{d \mathbf{u}(t)}{dt}=-i\mathbf{A}(t,\hat{x},\hat{p}) \mathbf{u}(t), \qquad \mathbf{u}_0=\mathbf{u}(0), \qquad \hat{x}=\hat{x}_1,...,\hat{x}_d, \qquad \hat{p}=\hat{p}_1,...,\hat{p}_d,
\end{align}
where $\mathbf{u}(t)= \int_{-\infty}^{\infty} u(t,x)|x\rangle dx$. Here $\mathbf{A}(t,\hat{x}, \hat{p})$ is a linear operator consisting of factors of $\hat{x}, \hat{p}$ depending on the form of the original partial differential equation. Each factor of $x$ in the original partial differential equation gives rise to $\hat{x}$ and each derivative $\partial^l/\partial x_j^l$ in the partial differential equation gives a contribution of $(i \hat{p}_j)^l$. However, here generally $\mathbf{A}(t) \neq \mathbf{A}^{\dagger}(t)$ is not Hermitian, so cannot be seen directly as a Hamiltonian. 

To turn this simulation problem into a Hamiltonian simulation problem, we want to find a corresponding $\hat{H}$ which is also a function of $\hat{x}$ and $\hat{p}$. This is possible by applying the method of Schr\"odingerisation \cite{schr1, schr2, jin2023analog} which can transform $\mathbf{A}=\mathbf{A}_1-i\mathbf{A}_2$ into a Hamiltonian $\hat{H}$ acting on one extra mode labelled $\xi$
\begin{align}
    & \mathbf{A}(t) \rightarrow \hat{H}=\mathbf{A}_2(t) \otimes \hat{x}_{\xi}+\mathbf{A}_1(t) \otimes \mathbf{1}_{\xi}.
\end{align}
Here $\mathbf{A}_1(t)=(\mathbf{A}(t)+\mathbf{A}^{\dagger}(t))/2=\mathbf{A}_1^{\dagger}$ corresponds to the completely Hermitian part of $\mathbf{A}(t)$ and $\mathbf{A}_2(t)=i(\mathbf{A}(t)-\mathbf{A}^{\dagger}(t))/2=\mathbf{A}_2^{\dagger}(t)$ is associated with the completely anti-Hermitian part of $\mathbf{A}(t)$. The method of Schr\"odingerisation allows the simple recovery of the quantum state proportional to $\mathbf{u}(t)$ by using Hamiltonian simulation in $\hat{H}(t)$. We can easily turn this into a time-independent Hamiltonian using the method previously described in Eq.~\eqref{eq:time-independent}. 

This means if $\mathbf{A}(t)$ is a polynomial sum of $\hat{x}$ and $\hat{p}$, then clearly $\hat{H}(t)$ will also be a polynomial sum of $\hat{x}, \hat{p}$. We underline here that for more complex functions like cosine, sine, etc we can transform it to polynomials using some expansions (Taylor like). Every partial differential equation has only integer-valued derivatives of order $l>0$ (with the exception of fractional partial differential equation), so these will give a contribution to $\hat{H}(t)$ of the polynomial form $\hat{p}^l$. This means, for instance, any constant coefficient partial differential equation are in the form of Eq.~\eqref{eq:multi-dimensional_Hamiltonian_first}. 

in Addition, some nonlinear partial differential equations like the Hamilton-Jacobi and scalar hyperbolic, as well as all nonlinear ordinary differential equations can also be written in the linear form Eq.~\eqref{eq:uode} \cite{jin2023analog, jin2023time, jin2022quantumnonlinear}. 

\noindent \textbf{Ground state and thermal state preparation.}

Given the block-encoding to $\hat{H}$, one can also perform ground state preparation or thermal state preparation with respect to the Hamiltonian $\hat{H}$, where $\hat{H}$ is of the form in Eq.~\eqref{eq:multi-dimensional_Hamiltonian_first}. There are two ways to go about this. A simple method that involves only block-encoding construction of $\exp(-i\hat{H}t)$, i.e. Hamiltonian simulation, is again to apply Schr\"odingerisation, with the choice of $\mathbf{A}=-i\mathbf{A}_2=-i\hat{H}$, see \cite{schr1} for more details. 

An alternative is not to construct block-encodings of $\exp(-i\hat{H}t)$, but to construct block-encodings of functions like $\exp(-\hat{H}t)$ or $\exp(-\beta \hat{H})$, see \cite{gilyen2019quantum} for more details. These operators can be similarly used for ground state and thermal state preparation. 
 
\subsection{Roadmap}\label{subsec:Roadmap}
   
   The steps in constructing the block-encoding for our $\hat{H}$ requires many steps and it can be easy to get lost while compiling many pieces of the puzzle together. To make it easier to use our technique, the roadmap for our algorithm is shown in Fig.~\ref{fig:road_map}.  For convenience, we also include the table of our notation in Table \ref{table_of_notations}. 

    \begin{figure}[h!]
    \includegraphics[width=0.75\textwidth]{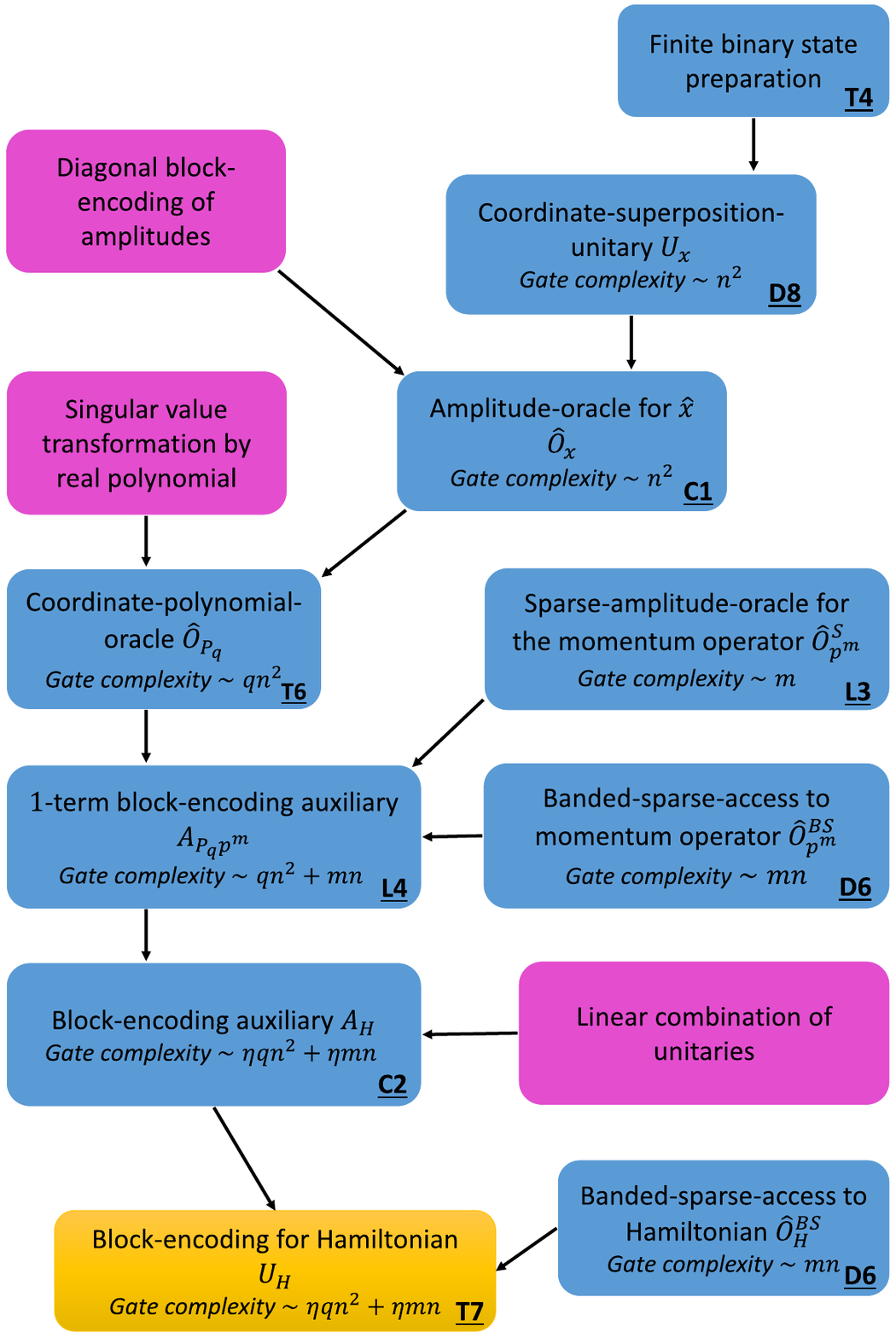}
    \caption{Road-map of implementing an efficient block-encoding for Hamiltonian with form (\ref{eq:main_Hamiltonian}). Each block indicates an important step in the overall algorithm. Moreover, we provide the notation for each structure and the scaling for its construction. The precise number of gates we provide in the text. Some blocks indicate Theorems (\textbf{T}), Corollaries (\textbf{C}), Lemmas (\textbf{L}), Definitions (\textbf{D}) where it can be found. The boxes are painted with blue if the step is an original concept developed in this paper. On the contrary, pink color indicates research from other papers. We also provide tips for gate complexity notation understanding: $n$ -- number of qubits; $m$ -- the maximum degree of momentum operator contributing the Hamiltonian; $q$-- the maximum degree of coordinate operator contributing the Hamiltonian; $\eta$ -- number of terms in the Hamiltonian. }
    \label{fig:road_map}
    \end{figure}

\subsection{Preliminaries}\label{sec:preliminaries}
    
    Now we introduce some mathematical and quantum computing objects that we will use, and the system of indexes. We start with a multi-control operator which is a common structure in quantum computing allowing for the application a certain gate with a given condition.
    
    \begin{definition}[Multi-control operator]
    \label{def:multiconrol operator}
        Let $U$ be an $m$-qubit quantum unitary and $b$ a bit string of length $n$. We define $C_U^b$ as an $n+m$ quantum unitary that applies $U$ to an $m$-qubit quantum register if and only the $n$-qubit quantum register is in state $\ket{b}$.
        \[ C^b_U=\ket{b}^n\bra{b}^n\otimes U+\sum\limits_{\substack{i=0,\dots,2^n-1\\ i\neq b}}\ket{i}^n\bra{i}^n\otimes I^{\otimes m}  \]
    \end{definition}
    
    \begin{figure}[h]
    \subcaptionbox{}{\includegraphics[width=0.283\textwidth]{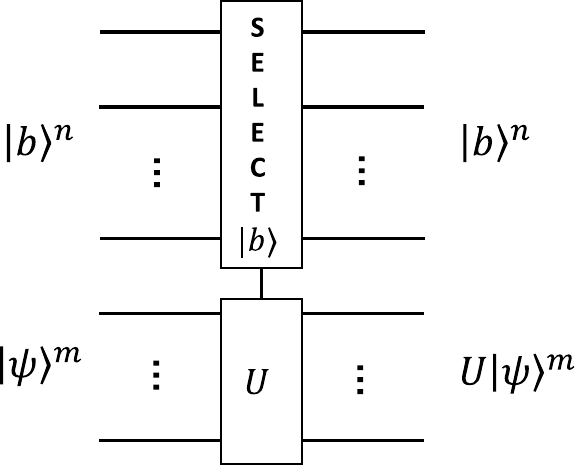}}
    \hspace{0.1\textwidth}
    \subcaptionbox{}{\includegraphics[width=0.48\textwidth]{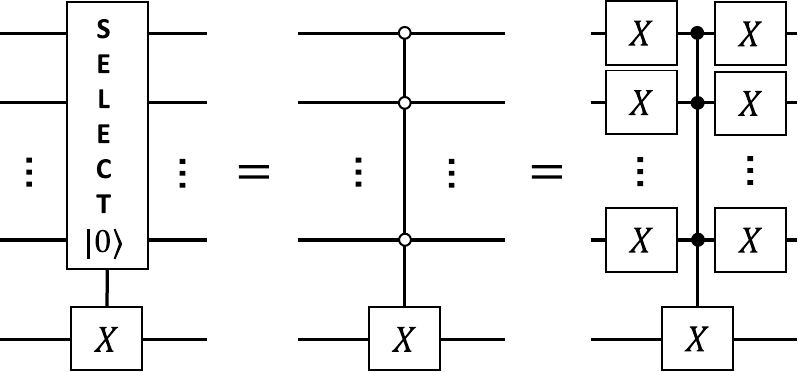}}
    \vspace{2.0em}
    \subcaptionbox{}{\includegraphics[width=0.68\textwidth]{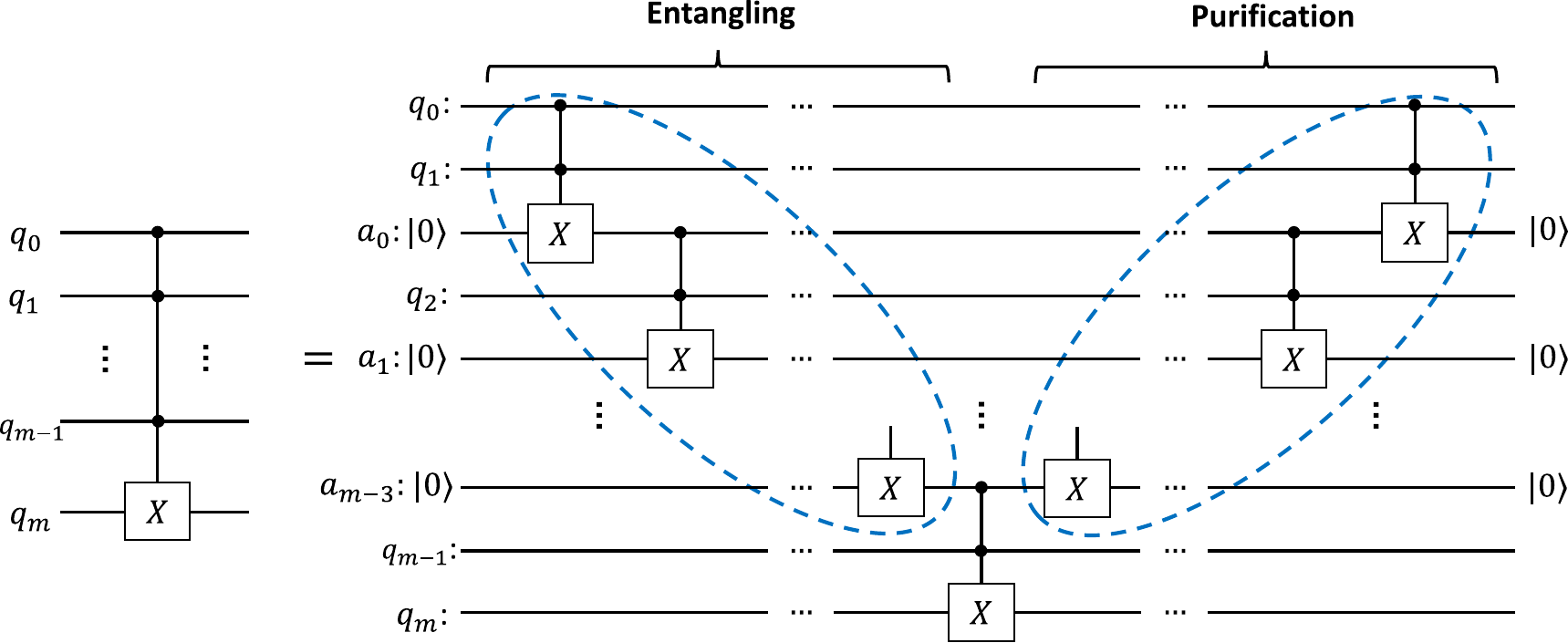}}
    \caption{(a) The general view of multi-control operator $C_U^b$ from the definition \ref{def:multiconrol operator}. Select operation means that the operator $U$ is applied only if the state on the upper register is $\ket{b}^n$ (b) Explicit quantum circuit for implementation $C_X^{00\dots 0}$. (c) Explicit quantum circuit for implementation $C^{11\dots1}_X$ which we use everywhere in this paper. We use $2m-3$ Toffoli application and $m-2$ pure ancillas.}
    \label{fig:Multi_control_qc}
    \end{figure}
    
    The general structure of a multi-control operator is illustrated in Fig.~\ref{fig:Multi_control_qc}. It is known \cite{nielsen2002quantum} that Toffoli gate can be realized using $6$ C-NOTs and $8$ single-qubit operations. Consequently, the resource requirement for a single Multi-control operator includes no more than one $C^1_U$ operator, $16n-16$ single-qubit operations, and $12n-12$ C-NOTs, simplifying the complexity of $C^{b}_X$ to $C^{11\dots 1}_X$ for simplicity. Additionally, we utilize $n-1$ ancillary qubits initialized to the zero state, which we further refer to as pure ancillas.
    
    \begin{definition}[Pure ancilla]
    \label{def:pure ancilla}
    Let $U$ be an $n+m$-qubit unitary operator such that for any arbitrary $\ket{\psi}^n$
    \[U\ket{0}^m\ket{\psi}^n=\ket{0}^m\ket{\phi}^n,\]
    where $\ket{\phi}^n$ is some quantum state then we say that the $m$-qubit quantum register is pure ancillas for operator $U$, if before and after the action it is in the state $\ket{0}^m$. In Fig.~\ref{fig:Multi_control_qc} (c) we use $m-2$ auxiliary qubits for efficient C-NOT construction. Exploitation of this qubits is divided into two epochs: (i) entangling stage that allows us to achieve some sophisticated quantum states, (ii) purification stage that sets ancillas back to zero-state.
    \end{definition}
    
    Now, we introduce a banded matrix which non-zero elements form a band shape. This object is particularly interesting as $\hat{p}^m$ is a banded matrix.
    
    \begin{definition}[Banded matrix]
    \label{def:Banded_matrix}
    A banded matrix is defined as an $N\times N$ matrix in which each $i$-th row is the $i-1$ times right-cyclic permuted first row
    \[ \forall i,j: A_{ij}=A_{0 \, j-i},  \]
    where $j-i$ means subtraction modulo $N$. Thus, the set of indeces is fully determined by its first row. The following $5\times5$ matrix is an example of such banded matrices 
    \[ 
    \left(\begin{array}{ccccc}
    6 & 1  & 0 & 9 & 8 \\
    8 & 6  &1 & 0 & 9\\
    9 & 8  &6 & 1 & 0\\
    0 & 9  &8& 6 & 1\\
    1 & 0  &9& 8 & 6\\
    \end{array}\right) .\]
    
    \end{definition}
    
    We define a $s$ sparse matrix as a matrix which contains no more than $s$ non-zero elements in each row. In our research, we focus on sparse banded matrices. To analyze these matrices in a convenient way, we introduce a sparse index specifically tailored for them.
    
    \begin{definition}[Banded-sparse matrix index]
    \label{def:Sparse_matrix}
        For any given banded sparse matrix $A$, we denote $A^{(s)}$ as the $s$-th non-zero element in the matrix's first row. Here is an example
        \[ A=\left(\begin{array}{cccccccc}
    2&3&0&0&0&0&0&4\\
    4&2&3&0&0&0&0&0\\
    0&4&2&3&0&0&0&0\\
    0&0&4&2&3&0&0&0\\
    0&0&0&4&2&3&0&0\\
    0&0&0&0&4&2&3&0\\
    0&0&0&0&0&4&2&3\\
    3&0&0&0&0&0&4&2
    \end{array}\right)=\left(\begin{array}{cccccccc}
    A^{(0)}&A^{(1)}&0&0&0&0&0&A^{(2)}\\
    A^{(2)}&A^{(0)}&A^{(1)}&0&0&0&0&0\\
    0&A^{(2)}&A^{(0)}&A^{(1)}&0&0&0&0\\
    0&0&A^{(2)}&A^{(0)}&A^{(1)}&0&0&0\\
    0&0&0&A^{(2)}&A^{(0)}&A^{(1)}&0&0\\
    0&0&0&0&A^{(2)}&A^{(0)}&A^{(1)}&0\\
    0&0&0&0&0&A^{(2)}&A^{(0)}&A^{(1)}\\
    A^{(1)}&0&0&0&0&0&A^{(2)}&A^{(0)}
    \end{array}\right). \]
    \end{definition}
    
    \begin{remark}
    We extend the definition of the banded-sparse matrix index to cases where a sparse $N\times N$ matrix $A$ becomes a banded matrix upon substituting all non-zero elements with $1$, reflecting the band structure of non-zero elements. We underline that our main interest here is to reflect the band position of the non-zero elements. Thus, the notation $A^{(s)}_i$ means that we first address an index of $s$-th non-zero elements in the first row then add $i$ (modulo N) to get the column number. Here is an example
    \[ A=\left(\begin{array}{cccccccc}
    2&6&0&0&0&0&0&4\\
    1&2&3&0&0&0&0&0\\
    0&6&7&7&0&0&0&0\\
    0&0&1&6&2&0&0&0\\
    0&0&0&3&8&9&0&0\\
    0&0&0&0&9&4&1&0\\
    0&0&0&0&0&6&6&6\\
    9&0&0&0&0&0&5&2
    \end{array}\right)=\left(\begin{array}{cccccccc}
    A^{(0)}_0&A^{(1)}_0&0&0&0&0&0&A^{(2)}_0\\
    A^{(2)}_1&A^{(0)}_1&A^{(1)}_1&0&0&0&0&0\\
    0&A^{(2)}_2&A^{(0)}_2&A^{(1)}_2&0&0&0&0\\
    0&0&A^{(2)}_3&A^{(0)}_3&A^{(1)}_3&0&0&0\\
    0&0&0&A^{(2)}_4&A^{(0)}_4&A^{(1)}_4&0&0\\
    0&0&0&0&A^{(2)}_5&A^{(0)}_5&A^{(1)}_5&0\\
    0&0&0&0&0&A^{(2)}_6&A^{(0)}_6&A^{(1)}_6\\
    A^{(1)}_7&0&0&0&0&0&A^{(2)}_7&A^{(0)}_7
    \end{array}\right). \]
    \end{remark}
    
    Based on the introduced sparse access we define a quantum operation that transforms that index into the column index of a non-zero element.
    
    \begin{definition}[Banded-sparse-access]
    \label{Banded-sparse-access}
    Given a $2^l$ sparse $2^n\times2^n$ matrix $A$ with non-zero elements arranged to form a banded matrix when replacing all non-zero elements with $1$, we define a Banded-sparse-access unitary oracle to $A$ as
        \[\hat{O}^{BS}_A\ket{0}^{n-l}\ket{s}^l\ket{i}^n:=\ket{r_{si}}^n\ket{i}^n,\]
    where $r_{si}=r_{s0} + i \mod 2^n$ corresponds to the $(s)$-th banded-sparse matrix index. If the transposed matrix $A^T$ shares the same non-zero element positions as $A$, then $\hat{O}^{BS}_A=\hat{O}^{BS}_{A^T}$. 
    \end{definition}
    
    \begin{lemma}\label{Banded-sparse-access lemma}
    The Banded-sparse-access $\hat{O}^{BS}_A$ to $2^l$ sparse $2^n\times2^n$ matrix $A$ with non-zero elements arranged to form a banded matrix when replacing all non-zero elements with $1$ can be implemented with resources no greater than: 
    \begin{enumerate}
        \item $(2^l+1)(32n-48)$ one-qubits operations;
        \item $25*2^ln-36*2^l+32n-48$ C-NOTs;
        \item $n-1$ pure ancillas.
    \end{enumerate}
    \end{lemma}
    \begin{proof}
        The explicit circuit design is shown and explained in the Appendix \ref{Appendix Banded-sparse-access}. It is important to note that when exponentiating the matrix $\hat{p}$ to the power of $m$, denoted $\hat{p}^m$, we obtain a matrix that also qualifies as a sparse banded matrix, with $m+1$ non-zero elements in each row.
    \end{proof}
    
    \section{Efficient block-encoding using Amplitude Hamiltonian query access}\label{main_section}
    
In this section, we present efficient block-encoding techniques utilizing Amplitude Hamiltonian query access. This advanced approach enables us to harness the full potential of quantum computing for simulating complex Hamiltonian dynamics with high precision. Through a detailed exploration of Hamiltonian superposition, we establish an auxiliary unitary framework that significantly enhances our ability to analyze and interact with Hamiltonian operators in a quantum computational setting. This is achieved by synthesizing Oracles to specific operators and employing a combination of sophisticated techniques, including Banded-sparse-access and Coordinate-polynomial-oracle, all supplemented by straightforward unitary operations. This section culminates in the development of a block-encoding for the Hamiltonian (\ref{eq:main_Hamiltonian}), providing us with a new method to access Hamiltonian information directly and efficiently.

\subsection{Coordinate-polynomial-oracle}

In this subsection, we construct an auxiliary unitary framework for the polynomial $P_{q}(\hat{x})$, denoted $\hat{O}_{P_q}$. We further combine this Coordinate-polynomial-oracle with similar structure for $\hat{p}^m$ to build the oracle for $P_{q} (\hat{x})\hat{p}^m$.

\begin{definition}[Coordinate-polynomial-oracle]
\label{def. Coordinate-polynomial-oracle}
Let $P_{q}(x)$ be a degree-$q$ polynomial satisfying that 
\begin{itemize}
    \item $P_{q}(x)$ has parity-($q$ mod 2);
    \item $P_{q}(x)\in\mathbb{R}[x]$;
    \item for all $x\in[-1,1]$: $\abs{P_{q}}<1$.
\end{itemize}
And let $\hat{x}$ be $n$-qubit matrix representations of the coordinate operator, as in Eq.~\ref{eq:x_p_cyclic_matrices}. Then we say that $\hat{O}_{P_q}$ is the Coordinate-polynomial-oracle for $P_q$ if
\[\hat{O}_{P_q}\ket{0}^\lambda\ket{i}^n=P_q(\hat{x}_{ii})\ket{0}^\lambda\ket{i}^n +J^{(i)}_{P_q}\ket{\bot_0}^{\lambda+n},\]
where $\ket{\bot_0}^{\lambda+n}$ is orthogonal to $\ket{0}^\lambda$. In other words, Coordinate-polynomial-oracle is a $(1,\lambda,0)$-block-encoding of the oracle access.
\end{definition}

Next, we design an auxiliary unitary operation which encodes the diagonal of $\hat{x}^q$ operator from (\ref{eq:x_p_cyclic_matrices}) into a quantum state.
    
    \begin{definition}[Coordinate-superposition-unitary]    .
    \label{def:Coordinate-superposition-unitary}
        Let $\hat{x}$ be a $n$-qubit finite-difference representation of quantum coordinate operator as in the Eq.~\ref{eq:x_p_cyclic_matrices}, then we define a $n$-qubit unitary $U_{x^q}$ that prepares a quantum states with amplitudes encoding the diagonal of $\hat{x}^q$
        \[ U_{x^q}\ket{0}^{n}=\frac{1}{\sqrt{\sum_{\kappa=0}^{2^n-1}(a+\kappa \Delta x)^{2q}}}\sum_{i=0}^{2^n-1}(a+i \Delta x)^{q}\ket{i}^n=:\frac{1}{\sqrt{\sum_{\kappa=0}^{2^n-1}\hat{x}_{\kappa\kappa}^{2q}}}\sum_{i=0}^{2^n-1}\hat{x}_{ii}^{q}\ket{i}^n. \]
    
    \end{definition}
    
    \begin{theorem}[Discrete Hadamard-Walsh transform for a polynomial superposition (App. A in \cite{guseynov2023depth})]\label{theorem:Walsh_Hadamard_poly}
        
        Let $\ket{\psi}^n$ be a $n$-qubit quantum state which is an arbitrary superposition of $U_{x^q}$ (see Definition~\ref{def:Coordinate-superposition-unitary}) acting on $\ket{0}^n$
        \[ \ket{\psi}^n=\frac{1}{\mathcal{N}_\psi}\sum^\chi_{q=0}\alpha_qU_{x^q}\ket{0}^n;\qquad \alpha_q\in\mathbb{C}\]
        where $\mathcal{N}_\psi$ is a normalization factor, $\abs{i}_b=\sum_{\kappa=0}^{n-1}i_\kappa$ be a binary norm then the Discrete Hadamard-Walsh transform \cite{walsh1923closed} of the $\ket{\psi}^n$ contains computational basis states only with $\abs{i}_b\leq\chi$
        \[ H_W^{\otimes n}\left(\frac{1}{\mathcal{N}_\psi}\sum^\chi_{q=0}\alpha_qU_{x^q}\ket{0}^n\right)=
        H_W^{\otimes n}\left(\frac{1}{\mathcal{N}_\psi}\sum\limits_{\substack{q=0,\dots,\chi\\i=0,\dots,2^n-1}}
        \alpha_q\hat{x}_{ii}^q\ket{i}^n\right)=\sum_{\abs{i}_b\leq\chi}\beta_i\ket{i}^n;
        \]
        \[\beta_i=\frac{1}{2^{n/2}\mathcal{N}_\psi}\sum\limits_{\substack{q=0,\dots,\chi\\\kappa=0,\dots,2^n-1}}
        \alpha_q\hat{x}_{\kappa\kappa}^q(-1)^{\sum_{m=0}^{n-1}\kappa_mi_m};\qquad H_W=\frac{1}{\sqrt{2}}\left
    (
    \begin{array}{cc}
    1 & 1 \\
    1 &-1 
    \end{array}
    \right),\]
    where $m$ indicates a bit string index.
        
    \end{theorem}

We notice that the equation for $\beta_i$ generally requires exponential resources of classical computer to be numerically find. However, we stress out that in this work we use Coordinate-superposition-unitary only for the first degree ($q=1$). The $\beta_i$ can be computed analytically when $q$ is a small number. The following Theorem shows an explicit quantum circuit for constructing $U_{x^q}$.

\begin{theorem}[Finite binary norm state preparation]
        Let $\ket{\psi}^n=\sum_{\abs{i}_b\leq q}\beta_i\ket{i}^n$ be an arbitrary normalized superposition of $n$-qubit computational basis states with binary norm no greater than $q\in\mathbb{N}$ with $\beta_i\in\mathbb{R}$ then the state preparation unitary 
        \[ A_\psi\ket{0}^n=\ket{\psi}^n\]
        can be implemented using: 
        \begin{enumerate}
            \item $\sum_{i=0}^{n-1}\sum_{s=i}^{\min\{q+i-1,n-1\}}\binom{i}{s} A(i)$ one-qubit operations;
            \item $\sum_{i=0}^{n-1}\sum_{s=i}^{\min\{q+i-1,n-1\}}\binom{i}{s} B(i)$ C-NOTs;
            \item $n-2$ pure ancillas;
        \end{enumerate}
        where
        \[ A(i)=\left\{
      \begin{array}{lll}
        1, & \text{if } i = 0, \\
        2, & \text{if } i=1,\\
        16i-14, & \text{if } i\geq 2
      \end{array}
    \right.\qquad  B(i)=\left\{
      \begin{array}{lll}
        0, & \text{if } i = 0, \\
        2, & \text{if } i=1,\\
        12i-10, & \text{if } i\geq 2
      \end{array}
    \right.\qquad \binom{i}{s}=\frac{s!}{i!(s-i)!}.\]
        \label{theorem Finite binary norm Ansatz}
    \end{theorem}
    \begin{proof}
        An explicit implementation of $A_\psi$ and a rule of choosing parameters of the quantum circuit are presented in the Appendix \ref{appendix Finite binary norm Ansatz}. The expressions for complexity is quite cumbersome, that's why we provide asymptotic with $n\gg q$, $q<50$, the gate complexity becomes: (i) $\frac{16n^{q+1}}{q!}$ one-qubit operations, (ii) $\frac{12n^{q+1}}{q!}$ C-NOTs.
    \end{proof}

The Theorems \ref{theorem:Walsh_Hadamard_poly},\ref{theorem Finite binary norm Ansatz} allow us to build a Coordinate-superposition-unitary $U_{x}$:
    \begin{equation}
        U_{x}\ket{0}^n=\sum_{i=0}^{2^n-1}\hat{x}_{ii}\ket{i}^n,
    \label{eq:ansatz X}
    \end{equation}
we remind that the normalization constant is assumed to be $1$ in this paper.

\subsubsection{From a superposition unitary to Amplitude-oracle}

In this subsection we exploit the technique from \cite{rattew2023non,guo2021nonlinear} to transform $U_x$ to oracle $\hat{O}_x$
\begin{eqnarray}
    \hat{O}_x\ket{0}^1\ket{i}^n=\frac{1}{\sqrt{2}}\hat{x}_{ii}\ket{0}^1\ket{i}^n+J_{x}\ket{\bot_0}^{n+1},
\end{eqnarray}
where $\ket{\bot_0}^{n+1}$ is orthogonal to $\ket{0}^1\ket{i}^n$ for any $i=0,1,2,\dots,2^n-1$.

We begin our algorithm from the discussion of the work \cite{rattew2023non}. Given a unitary $U_S$
    \begin{equation}
        U_S\ket{0}^h=\sum_{i=0}^{2^{h}-1}\psi^S_j\ket{j}=:\ket{\psi^S};\qquad \psi^S_j\in\mathbb{C};\qquad \sum_{j=0}^{2^n-1}\abs{\psi^S_j}^2=1;
    \label{eq:rebentrost unitary}
    \end{equation}
    authors suggest an algorithm to build a $(1,h+2,0)$ block-encoding \cite{dalzell2023quantum} of the $A^{S}_{(p)}:=diag(Re((-i)^p\psi^S_0), Re((-i)^p\psi^S_1),$ $ \dots, Re((-i)^p\psi^S_{2^h-1}))$, see Lemma 6 from \cite{rattew2023non}. We show the direct view of the block-encoding and all the operators that are used in Eq.~\ref{eq:rebentrost block encoding unitary}
    \begin{eqnarray}
        \begin{gathered}
        U^{(S)}_{(p)}:=(XZX\otimes I^{\otimes 2h+1})(H_W\otimes W^\dagger_p)(\ket{0}^1\bra{0}^1\otimes G_p+\ket{1}^1\bra{1}^1\otimes G_p^\dagger)(H_W\otimes W_p);\\
        \tilde{Z}:=I^{\otimes h}\otimes Z\otimes I^{\otimes h};\qquad \tilde{H}_W:=I^{\otimes h}\otimes H_W\otimes I^{\otimes h};\qquad \tilde{S}:=I^{\otimes h}\otimes S\otimes I^{\otimes h};\\
        R:=(I^{\otimes h+1}-2\ket{0}^{h+1}\bra{0}^{h+1})\otimes I^{\otimes h};\\
        U_C^S:=(U_S\otimes\ket{0}^1\bra{0}^1+I^{\otimes h}\otimes\ket{1}^1\bra{1}^1)\otimes I^{\otimes h};\\
        U_{Copy}:=I^{\otimes h}\otimes\ket{0}^1\bra{0}^1\otimes I^{\otimes h}+\sum_{k,j=0}^{2^h}\ket{j\oplus   k}^h\bra{j}^h\otimes\ket{1}^1\bra{1}^1\otimes\ket{k}^h\bra{k}^h;\\
        W_p:=\tilde{H}_W\tilde{S}^pU_{Copy}U_C^S\tilde{H}_W;\qquad G_P:=W_pRW_p^\dagger\tilde{Z},\\
         \end{gathered}
        \label{eq:rebentrost block encoding unitary}
    \end{eqnarray}
    where $S=\left
    (
    \begin{array}{cc}
    1 & 0 \\
    0 & i 
    \end{array}
    \right)$, $p=1/0$ means that $\tilde{S}$ is applied/not applied. A construction of all those operators with one and two qubit operations is well described in the article. Let us now consider the following lemma.
    
    \begin{lemma}[Eigensystem of $-\frac{1}{2}W_p^\dagger(G_p+G_p^\dagger)W_p$ (Lemma $5$ from the \cite{rattew2023non})]
    Let $p=\{0,1\}$. The operator $-\frac{1}{2} W_p^\dagger$ $ (G_p + G_p^\dagger) W_p$ has eigenvectors $\ket{0}^h\ket{0}^1\ket{k}^h$ with associated eigenvalue $Re(\psi^S_k)$ if $p=0$ and $Im(\psi^S_k)$ if $p=1$.  
    \end{lemma}
    
    Since the operation $-\frac{1}{2}W_p^\dagger(G_p+G_p^\dagger)W_p$ is supposed to act on superposition of states like $\ket{0}^h\ket{0}^1\ket{k}^h$ it is clear that at the end of operation the auxiliary registers are set to zero-state. Consequently, these $h+1$ auxiliary qubits should be considered as pure ancillas because we do not have to measure those qubits to get the desired action of $A^S_{(p)}$; in opposite, we can use it further for some computations. The only qubit that is needed to be measured is that one that is used to create the linear combination of unitary $G_p+G_p^\dagger$. Thus, we state that $U^{(S)}_{(p)}$ is a (1,1,0) block-encoding of $A^S_{(p)}$ which uses $h+1$ pure ancillas.
    
    \begin{theorem}[Diagonal block-encoding of amplitudes (Modified theorem 2 from the \cite{rattew2023non}]\label{theorem:Diagonal block encoding of amplitudes}
        Given an $h$-qubit quantum state specified by a state-preparation-unitary $U_S$, such that $\ket{\psi^s}^h=U_S\ket{0}^h=\sum_{j=0}^{2^h-1}\psi_j^S\ket{j}^h$ (with $\psi_J^S\in\mathbb{C}$), we can prepare a $(2,2,0)$ - block-encoding $U^\prime_A$ of the $A^\prime=diag(\psi^S_0,\dots,\psi_{2^h-1}^S)$ using: 
        \begin{enumerate}
            \item $6$ Applications of $U^S_C$;
            \item $528h-463$ one-qubit operations;
            \item $428h-378$ C-NOTs;
            \item $2h-1$ pure ancillas.
        \end{enumerate} 
    \end{theorem}
    \begin{proof}
        Using the linear combination of unitaries $U^{S}_{(0)}+iU^{S}_{(1)}$ combining real and imaginary parts of $\psi^S_j$ we can construct the desired block-encoding of $A^\prime$. The particular number of operations is computed from Eq.~\ref{eq:rebentrost block encoding unitary}. Thus, we count that each control-$U^S_C$ can be implemented using: (i) two Toffoli, (ii) one $U^S_C$, (iii) one pure ancilla, where we assumed that each control-one qubit operation can be implemented using: (i) $2$ one-qubit operations, (ii) $2$ C-NOTs.
    \end{proof}

\begin{corollary}[Amplitude-oracle for $\hat{x}$]
\label{Corollary: x-oracle}
Let $\hat{x}$ be a $n$-qubit matrix as in \ref{eq:x_p_cyclic_matrices} with $\sum_{\kappa=0}^{2^n-1} (a+\kappa \Delta x)^2=1$. Then we can construct $(\sqrt{2},1,0)$-block-encoding $\hat{O}_x$
\[ \hat{O}_x\ket{0}^1\ket{i}^n=\frac{1}{\sqrt{2}}\hat{x}_{ii}\ket{0}^1\ket{i}^n+J_{x}\ket{\bot_0}^{n+1},\]
with resources no greater than
\begin{enumerate}
    \item $2304n^2-1064n-109$ one-qubit operations;
    \item $1864n^2-860n-92$ C-NOTs;
    \item $2n-1$ pure ancillas.
\end{enumerate}

\end{corollary}
\begin{proof}
    We note that to build the controlled version of some unitary having its decomposition to one-qubit gates and C-NOTs we turn all C-NOTs to Toffoli gates, one-qubit rotation to its controlled version (2 C-NOTs, 2 one-qubit rotations). For more details we address the book \cite{nielsen2002quantum}. Hereafter, we always use this fact.

    The Oracle $\hat{O}_x$ can be implemented using the Theorem~\ref{theorem:Diagonal block encoding of amplitudes} applied to $U_x$ from the Eq~\ref{eq:ansatz X} and taking into account that $x$ is real. That's why the number of additional qubits in the block-encoding is $1$ against $2$ in the Theorem~\ref{theorem:Diagonal block encoding of amplitudes}.

\end{proof}

\subsubsection{From $x$ to a polynomial}

In this subsection we use the trick called qubitization \cite{gilyen2019quantum} to transform $\hat{O}_x$ to $\hat{O}_{P_q}$. The main building block in this trick is so called Alternating phase modulation sequence. This unitary operation is a key instrument for building $\hat{O}_{P_q}$.

\begin{definition}[Alternating phase modulation sequence]
    \label{def:Alternating phase modulation sequence}
    Let $U$ be a $(1,m,0)$-block-encoding of hermitian matrix $A$ such that
    \[A=(\ket{0}^m\bra{0}^m\otimes I^{\otimes n})U(I^{\otimes n}\otimes\ket{0}^m\bra{0}^m);\]
    let $\Phi\in\mathbb{R}^q$, then we define the $q$-phased alternating sequence $U_\Phi$ as follows
    \[
U_{\Phi} := 
\begin{cases} 
e^{i \phi_1 (2 \Pi - I)} U \prod_{j=1}^{(q-1)/2} \left( e^{i \phi_{2j} (2 \Pi - I)} U^\dagger e^{i \phi_{2j+1} (2 \Pi - I)} U \right) & \text{if } q \text{ is odd, and} \\
\prod_{j=1}^{q/2} \left( e^{i \phi_{2j-1} (2 \Pi - I)} U^\dagger e^{i \phi_{2j} (2 \Pi - I)} U \right) & \text{if } q \text{ is even,}
\end{cases}
\]
where $2\Pi-I=(2\ket{0}^m\bra{0}^m-I^{\otimes m})\otimes I^{\otimes n}$.
\end{definition}

Now we apply the Alternating phase modulation sequence on $\hat{O}_x$. In this case $m=1$, $e^{i \phi (2 \Pi - I)}$ turns to implementing the one-qubit rotation $R_z(-2\phi)$. The quantum circuit for $U_\Phi$ for this case is depicted in Fig.~\ref{fig:U_phi}.
\begin{equation}
    R_z(\theta) = \begin{pmatrix}
    e^{-i \frac{\theta}{2}} & 0 \\
    0 & e^{i \frac{\theta}{2}}
    \end{pmatrix}.
\label{eq:Rz}
\end{equation}

\begin{figure}[h]
    \includegraphics[width=0.7\textwidth]{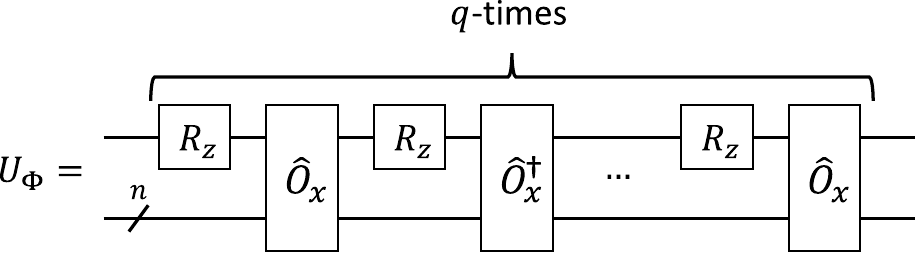}
    \caption{Alternating phase modulation sequence $U_\Phi$ for $\hat{O}_x$ with odd $q$.}
    \label{fig:U_phi}
    \end{figure}

Now we have all the ingredients to build the Coordinate-polynomial-oracle $\hat{O}_{P_q}$ using singular value transformation.

\begin{theorem}[Singular value transformation by real polynomials (Corollary 18 from \cite{gilyen2019quantum})]
\label{theorem: Singular value transformation by real polynomials}
    Let $P_{q}(x)$ be a degree-$q$ polynomial as in Definition~\ref{def. Coordinate-polynomial-oracle} and $U_\Phi$ is the Alternating phase modulation sequence as in Definition~\ref{def:Alternating phase modulation sequence} based on Amplitude-oracle for $n$-qubit matrix $\hat{x}$ as in Corollary~\ref{Corollary: x-oracle}. Then there exist $\Phi\in\mathbb{R}^n$, such that
    \[
    \hat{O}_{P_q} = \left(H_W\otimes I^{\otimes n+1} \right) \left( |0\rangle \langle 0| \otimes U_{\Phi} + |1\rangle \langle 1| \otimes U_{-\Phi} \right)\left(H_W\otimes I^{\otimes n+1} \right),
    \]
    where $\hat{O}_{P_q}$ is Coordinate-polynomial-oracle as in Definition~\ref{def. Coordinate-polynomial-oracle}
    \[\hat{O}_{P_q}\ket{0}^2\ket{i}^n=(P_q(\hat{x}))_{ii}\ket{0}^2\ket{i}^n +J^{(i)}_{P_q}\ket{\bot_0}^{n+2}.\]
    The total complexity of construction $\hat{O}_{P_q}$ does not exceed:
    \begin{enumerate}
        \item $2304qn^2-1064qn-108q$ one-qubit operations;
        \item $1864qn^2-860qn-92q$ C-NOTs;
        \item $2n-1$ pure ancillas.
    \end{enumerate}
    Moreover, given $P_q(x)$ and $\delta \geq 0$ we can find a $P^\prime_q(x)$ and a corresponding $\Phi$, such that $|P^\prime_q(x) - P_q(x)| \leq \delta$ for all $x \in [-1,1]$, using a classical computer in time $\mathcal{O}(\text{poly}(q, \log(1/\delta)))$.
\end{theorem}
\begin{proof}
    The rule of choosing $\Phi$ and the general proof is given in \cite{gilyen2019quantum}. We computed complexities using Amplitude-oracle for $\hat{x}$ (Corollary~\ref{Corollary: x-oracle}) as an input for $U_\Phi$ depicted in Fig.~\ref{fig:U_phi}.
\end{proof}

\begin{remark}
    In fact, we create 
    \[\hat{O}_{P_q}\ket{0}^2\ket{i}^n=(P_q(\hat{x}/\sqrt{2}))_{ii}\ket{0}^2\ket{i}^n +J^{(i)}_{P_q}\ket{\bot_0}^{n+2};\]
    however, we always can redefine polynomials to fit the simple view (\ref{def. Coordinate-polynomial-oracle}).
\end{remark}

\subsection{Sparse-amplitude-oracle for momentum operator}

In this subsection we achieve another important milestone toward a block-encoding of (\ref{eq:main_Hamiltonian}) which is Sparse-amplitude-oracle for the operator $\hat{p}^m$
    \begin{equation}
        \hat{O}^S_{p^m}\ket{0}^1\ket{s}^l:=\frac{(\hat{p}^m)^{(s)}}{\sqrt{\mathcal{N}_{p^m}}}\ket{0}^1\ket{s}^l+\sqrt{1-\frac{\abs{(\hat{p}^m)^{(s)}}^2}{\mathcal{N}_{p^m}}}\ket{1}^1\ket{s}^l,
        \label{momentum oracle}
    \end{equation}
     where $l$ is logarithm of sparsity $l=\lceil \log_2 (m+1)\rceil$ for $\hat{p}^m$ in the form (\ref{eq:x_p_cyclic_matrices}); $\mathcal{N}_{p^m}\geq\norm{\hat{p}^m}^2_{\max}$. We underline here that $(\hat{p}^m)^{(s)}$ doesn't depend on a row index  since matrix $\hat{p}^m$ is sparse and banded.
    
    \begin{lemma}\label{lemma: Amplitude-oracle for momentum operator}
        Let $\hat{p}$ be a $2^n\times 2^n$ matrix as in the Eq. \ref{eq:x_p_cyclic_matrices} then the Sparse-oracle-access to the operator $\hat{p}^m$ where $2^l=m+1$ can be implemented (i) $2^l$ one-qubit operations, (ii) $2^l$ C-NOTs.
    \end{lemma}
    
    \begin{proof}
        We design a quantum circuit for this operation and compute the complexity in the Appendix \ref{appendix_oracle_for_momentum}. We note that $\hat{p}^m_{ij}=(-1)^m\hat{p}^m_{ji}$ which means that the similar construction can be implement for the column sparsity.
    \end{proof}

\subsection{$1$-term block-encoding}

In this section we combine Coordinate-polynomial-oracle from Theorem~\ref{theorem: Singular value transformation by real polynomials} and Sparse-amplitude-oracle for momentum operator from Lemma~\ref{lemma: Amplitude-oracle for momentum operator} into the auxiliary operators
\begin{eqnarray}
    \begin{gathered}
    A_{P_qp^m}\ket{0}^3\ket{0}^n\ket{\psi}^n=\frac{(-1)^m}{\sqrt{2^l\mathcal{N}_{p^m}}}\sum\limits_{\substack{j=0,\dots,2^n-1\\ i\in F^{(m)}_j}}
    \underbrace{P_q(\hat{x}_{ii})\left(\hat{p}^m\right)_{ij}}_{\left(P_q(\hat{x})\hat{p}^m\right)_{ij}}\sigma_j\ket{0}^3\ket{i}^n\ket{j}^n+J_{P_qp^m}\ket{\bot_0}^{2n+3};\\
     A_{p^mP_q}\ket{0}^3\ket{0}^n\ket{\psi}^n=\frac{(-1)^m}{\sqrt{2^l\mathcal{N}_{p^m}}}\sum\limits_{\substack{j=0,\dots,2^n-1\\ i\in F^{(m)}_j}}
    \underbrace{\left(\hat{p}^m\right)_{ij}P_q(\hat{x}_{jj})}_{\left(\hat{p}^mP_q(\hat{x})\right)_{ij}}\sigma_j\ket{0}^3\ket{i}^n\ket{j}^n+J_{p^mP_q}\ket{\bot_0}^{2n+3};\\
    \ket{\psi}^n=\sum_{j=0}^{2^n-1}\sigma_j\ket{j}^n,
    \end{gathered}
    \label{eq:1term block encoding auxiliarry}
\end{eqnarray}
where we associate the index $j$ with columns and the index $i$ with rows; $\ket{\bot_0}^{2n+3}$ is orthogonal to $\ket{0}^3\ket{\kappa}^{2n}$ for any $\kappa=0,1,\dots,2^{2n}-1$; $F_j^{m}$ means set of non-zero elements in $j$-th column of the $\hat{p}^m$. The usefulness of these operators can be seen if we apply Hadamard operator $H_W^{\otimes n}$ on the last register which effectively creates block-encodings of $P_q(\hat{x})\hat{p}^m$ or $\hat{p}^mP_q(\hat{x})$.

\begin{lemma}[$1$-term block-encoding auxiliary]
\label{lemma: 1-term block-encoding auxiliary}

Let $P_q$ be a polynomial as in Definition~\ref{def. Coordinate-polynomial-oracle}, $\hat{x}$ and $\hat{p}$ be the $n$-qubit matrices as in Eq.~\ref{eq:x_p_cyclic_matrices}, and let $m+1=2^l$ then operators $A_{P_qp^m}$ or $A_{p^mP_q}$ can be created using 
\begin{enumerate}
        \item $2304qn^2-1064qn+32\times 2^ln+32n-108q-48\times 2^l+2^l-48$ one-qubit operations;
        \item $1864qn^2-860qn+25\times 2^ln+32n-92q-17\times 2^l-48$ C-NOTs;
        \item $2n-1$ pure ancillas.
    \end{enumerate}
\end{lemma}
\begin{proof}
    We provide the Fig.~\ref{fig:one-term-auxiliary} which explicitly shows how to build the auxiliary unitaries. The number of operations is simply followed from the parts forming those auxiliary operations for the $1$-term: (i) Coordinate-polynomial-oracle $\hat{O}_{P_q} (Theorem~\ref{theorem: Singular value transformation by real polynomials}$), (ii) Sparse-amplitude-oracle for momentum operator $\hat{O}^S_{p^m}$ (Lemma~\ref{lemma: Amplitude-oracle for momentum operator}), (iii) Banded-sparse-access for momentum operator $\hat{O}^{BS}_{p^m}$ (Lemma~\ref{Banded-sparse-access lemma}). The following equation clarify the circuits
    \begin{eqnarray}
        \begin{gathered}
            \ket{\phi_1}^{2n+3}=\frac{1}{\sqrt{2^l}}\sum_{s=0}^{2^l-1}\sum_{j=0}^{2^n-1}\sigma_jP_q(
            \hat{x}_{jj})\ket{0}^1\ket{s}^l\ket{0}^{n-l}\ket{0}^2\ket{j}^n+\dots\\
            \ket{\phi_2}^{2n+3}=\frac{1}{\sqrt{2^l\mathcal{N}_{p^m}}}\sum_{s=0}^{2^l-1}\sum_{j=0}^{2^n-1}\sigma_jP_q(
            \hat{x}_{jj})(\hat{p}^m)^{(s)}\ket{0}^1\ket{s}^l\ket{0}^{n-l}\ket{0}^2\ket{j}^n+\dots\\
            \ket{\phi_3}^{2n+3}=\frac{1}{\sqrt{2^l\mathcal{N}_{p^m}}}\sum\limits_{\substack{j=0,\dots,2^n-1\\ i\in F^{(m)}_j}}\sigma_jP_q(
            \hat{x}_{jj})(\hat{p}^m)_{ji}\ket{0}^3\ket{i}^n\ket{j}^n+\dots\\
            =\frac{(-1)^m}{\sqrt{2^l\mathcal{N}_{p^m}}}\sum\limits_{\substack{j=0,\dots,2^n-1\\ i\in F^{(m)}_j}}\sigma_jP_q(
            \hat{x}_{jj})(\hat{p}^m)_{ij}\ket{0}^3\ket{i}^n\ket{j}^n+\dots.
        \end{gathered}
        \label{eq:Clarifying A_p_x}
    \end{eqnarray}

\end{proof}

\begin{figure}[h!]
    \subcaptionbox{}{\includegraphics[width=0.39\textwidth]{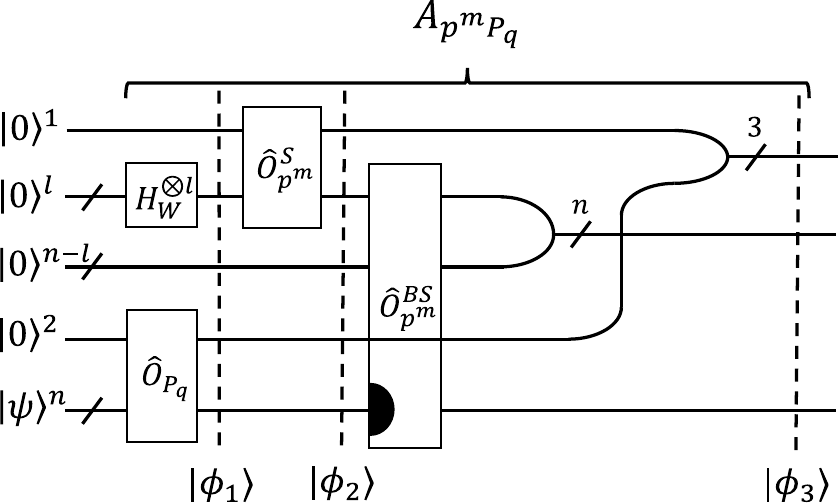}}
    \hspace{0.1\textwidth}
    \subcaptionbox{}{\includegraphics[width=0.39\textwidth]{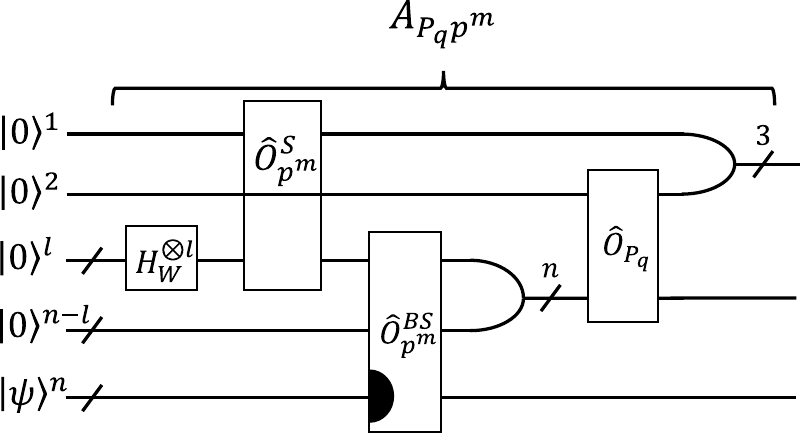}}
    \caption{Circuit design for implementing $A_{p^mP_q}$ (a) and $A_{P_qp^m}$ (b), where the lats register is associated with column number; $l$- and $n-l$-qubit registers are both correspond to the row number. The Banded-sparse-access $\hat{O}_{p^m}^{BS}$ for $p^m$ coincides to the one for $P_q(\hat{x})\hat{p}^m$ since the position of non-zero elements coincides. We underlined the register which does not change under action of $\hat{O}_{p^m}^{BS}$ by the filled semicircle. The Y-shaped frame means simple merging of registers without any operations. In the scheme (a) $\hat{O}^{BS}_{p^m}$ does not act in the fourth register, and on the scheme (b) $\hat{O}^S_{p^m} $does not act on the second one; we indicated this feature by plotting registers' lines above the boxes for operations. The fourth registers in the scheme (a) change his place and merge with the first register at the end of the scheme. We clarify the circuit (a) in Eq.~\ref{eq:Clarifying A_p_x}.}
    \label{fig:one-term-auxiliary}
\end{figure}

\subsection{Block-encoding}
In this section we combine $1$-term block-encodings auxiliaries according the view of the Hamiltonian (\ref{eq:main_Hamiltonian}) using the linear combination of unitary technique \cite{dalzell2023quantum}; then we apply Banded-sparse-access to the Hamiltonian to purify ancillas for row/column number. All those steps allow us to build an efficient block-encoding for the Hamiltonian. This process highlights the interplay between quantum algorithmic constructs and the physical representations of quantum systems, paving the way for advanced simulations and computations in quantum physics.

\begin{corollary}[Block-encoding auxiliary]\label{corollary:Hamiltonian block-encoding auxiliary}
    Let $H$ be a Hamiltonian as in the Eq.~\ref{eq:main_Hamiltonian} with $\eta=2^\gamma$, $q=\max\limits_kq_k$, $2^{l_k}=m_k$, $l_{\max}=\max\limits_kl_k$, and matrix representation as in Eq.~\ref{eq:numeric_main_Hamiltonian} then the unitary operation $A_H$
         \begin{equation*}
         \label{eq:Hamiltonian block-encoding auxiliary}\begin{gathered}
             A_H\sum_{j=0}^{2^n-1}\sigma_j\ket{0}^{\gamma+4}\ket{0}^n\ket{j}^n=\left({U^H_\alpha}^T\otimes I^{\otimes 2n+1}\right)\left(\prod_{k=0}^{\eta-1}C_{A_{p^{m_k}P_{q_k}}}^{k+\eta}\right)\left(\prod_{k=0}^{\eta-1}C_{A_{P_{q_k}p^{m_k}}}^{k}\right)\left[U^H_\alpha\ket{0}^{\gamma+1}\right] \sum_{j=0}^{2^n}\sigma_j\ket{0}^{3}\ket{0}^n\ket{j}^n\\
             =\frac{1}{\mathcal{N}_H}\sum\limits_{\substack{i=0,\dots,2^n-1\\ j\in F^H_i}}\left(\hat{H}_{ij}\sigma_j\ket{0}^{\gamma+4}\ket{i}^n\ket{j}^n+J_{\hat{H}_{ij}}\ket{\bot_{0}}^{\gamma+2n+4}\right);\qquad \sigma_j\in \mathbb{C} \quad \forall j=0,1,\dots,2^n-1 \end{gathered}\end{equation*}
         \begin{equation*}
         \label{eq:Normalization constant for $A_H$}\mathcal{N}_H=2\sum_{\kappa=0}^{2^\gamma -1}\abs{\alpha_k}\sqrt{2^{l_\kappa}\mathcal{N}_{p^{m_k}}},\end{equation*}
    can be implemented with resources no greater than:
    \begin{enumerate}
        \item $2^{\gamma+2}\left[9760n^2q+33\times2^{l_{\max}+2}n+160n+8\gamma-4504nq-476q-115\times 2^{l_{\max}}-239\right]$ one-qubit operations;
        \item $2^{\gamma+1}\left[15792n^2q+107\times 2^{l_{\max}+1}+256n+12\gamma-7288nq-768q-49\times 2^{l_{\max}+1}-382\right]$ C-NOTs;
        \item $2n+\gamma-1$ pure ancillas.
        
    \end{enumerate}
    The state $\ket{\bot_{0}}^{\gamma+2n+4}$ is orthogonal to $\ket{0}^{\gamma+4}\otimes I^{\otimes2n}$; the symbol $T$ in ${U^H_\alpha}^T$ means the matrix transposition.  The procedure exploits the linear combination theorem \cite{gilyen2019quantum,dalzell2023quantum} applied on $A_{P_qp^m}$, $A_{p^mP_q}$ from the Lemma~\ref{lemma: 1-term block-encoding auxiliary}. The quantum circuit for the $A_H$ is depicted in Fig.~\ref{fig:A_H quantum_circuit}. We suppose that a values $J_{\hat{H}_{ij}}$ and explicit view of $\ket{\bot_{0}}^{\gamma+2n+4}$ are unimportant for the reader, so we omit it; however, they can be retrieved from the Fig.~\ref{fig:A_H quantum_circuit} and explicit views of the operators.
    
    \end{corollary}
     \textit{Clarification.} As a subroutine the implementation of $A_H$ uses $U^H_\alpha$
     \[ U_\alpha^H \ket{0}^{\gamma+1} = \frac{1}{\sqrt{\mathcal{N}_H}} \sum_{k=0}^{2^\gamma-1} \left(2^{l_k} \mathcal{N}_p^{m_k}\right)^{1/4} \left[ \sqrt{i^{m_k}\alpha_k} \ket{k}^{\gamma+1} + \sqrt{i^{m_k}\alpha_k^*} \ket{k+2^\gamma}^{\gamma+1} \right],\]    
     where $i^{m_k}$ corresponds to the product of the omitted global phase $(-i)^{m_k}$ discussed in the Appendix \ref{appendix_oracle_for_momentum} and to $(-1)^{m_k}$ from the $A_{P_qp^m},A_{p^mP_q}$; the $\sqrt{\alpha}$ is the principle square root. Since $\alpha_{k}$ can be any complex number in principle, the construction of $U_\alpha^H$ addresses the general problem of construction an arbitrary quantum state. The authors of \cite{plesch2011quantum,bergholm2005quantum} suggest an explicit implementation of $U_\alpha^H$ which consist of less than $\frac{23}{24}2^{\gamma+1}$ for even number of qubits and $2^{\gamma+1}$ for the odd number. Let's consider the number of one-qubit operations here. Naively, one can conclude that each C-NOT can be accompanied by $6$ Euler rotations; however, a rotation around $z$ axis commutes with control operation, and a rotation around $x$ one commutes with Pauli $X$ operation. Thus, the number of one-qubit Euler rotations cannot exceed the number of C-NOTs by more than 4 times \cite{shende2004minimal}. For simplicity in this paper we address a general one-qubit operation, so we merge 2 Euler rotations acting on each two qubits where the C-NOT is applied. Thus, the complexity of the $U_\alpha^H$ implementation is less than: (i) $2^{\gamma+2}$ one-qubit rotations, (ii) $2^{\gamma+1}$ C-NOTs.

  \begin{figure}[h!]
    \includegraphics[width=1.0\textwidth]{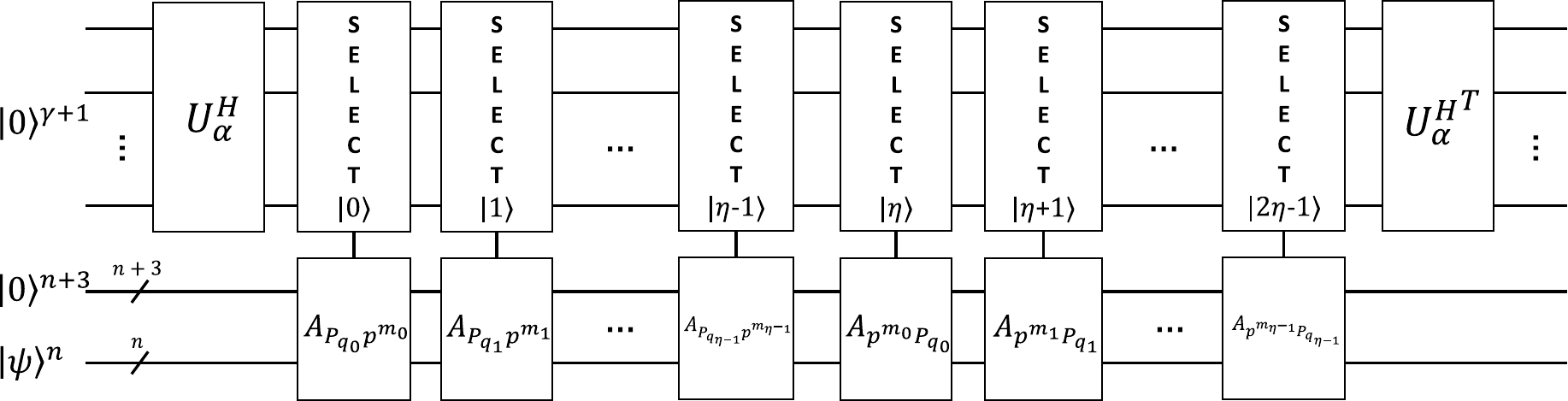}
    \caption{Quantum circuit implementing $A_H$. The overall number of qubits is $2n+\gamma+4$, where $\gamma+1$ qubits are used to create linear superposition of unitaries. The symbol $T$ in ${U^H_\alpha}^T$ means the matrix transposition }
    \label{fig:A_H quantum_circuit}
    \end{figure}

Now we estimate the sparsity of $\hat{H}$ from the formula (\ref{eq:main_Hamiltonian}). Using Binomial theorem \cite{coolidge1949story} applied to $\hat{p}^m$ it can be proved that the sparsity of the Hamiltonian matrix is no greater than $\max\limits_{k}2m_k+1$; we assume that the sparsity value is $2^l$. Thus, using the quantum circuits from the Appendix~\ref{Appendix Banded-sparse-access} we can efficiently implement the quantum circuit for the unitary Hamiltonian banded-sparse access $\hat{O}^{BS}_H$.

\begin{theorem}[Hamiltonian Block-encoding]\label{theorem:Hamiltonian Block-encoding}
    Let $H$ be a Hamiltonian as in Eq.~\ref{eq:main_Hamiltonian} with the number of terms $\eta=2^\gamma$, $2^l=\max\limits_{k}2m_k+1$, $q=\max\limits_kq_k$, and with matrix representation $\hat{H}$ as in Eq.~\ref{eq:numeric_main_Hamiltonian} with normalization constant $\mathcal{N}_H$ (see Corollary~\ref{corollary:Hamiltonian block-encoding auxiliary}) then we can implement a $(\sqrt{2^{l}}\mathcal{N}_H,l+\gamma+4,0)$-block-encoding of $n$-qubit matrix $\hat{H}$
    \[U_H\ket{0}^{l+4}\ket{\phi}^n=\frac{1}{\sqrt{2^l}\mathcal{N}_H}\sum\limits_{\substack{i=0,\dots,2^n-1\\ j\in F^H_i}}\left(\hat{H}_{ij}\sigma_j\ket{0}^{l+4}\ket{i}^n
        +J^U_{\hat{H}_{ij}}\ket{\bot_{0}}^{n+l+4}\right);\qquad \ket{\phi}^n:=\sum_{j=0}^{2^n-1}\sigma_j\ket{j}^n,
    \]
    with resources no greater than:
    \begin{enumerate}
        \item $2^{\gamma+2}\left[9760n^2q+33\times2^{l_{\max}+2}n+160n+8\gamma-4504nq-476q-115\times 2^{l_{\max}}-239\right]+2^{l+5}n-3\times 2^{l+4}+32n+l-48$ one-qubit operations;
        \item $2^{\gamma+1}\left[15792n^2q+107\times 2^{l_{\max}+1}+256n+12\gamma-7288nq-768q-49\times 2^{l_{\max}+1}-383\right]+25\times 2^ln-36\times 2^l+32n-48$ C-NOTs;
        \item $3n+\gamma-l-1$ pure ancillas.
    \end{enumerate}
    \end{theorem}
    
    \begin{figure}[h!]
    \includegraphics[width=0.5\textwidth]{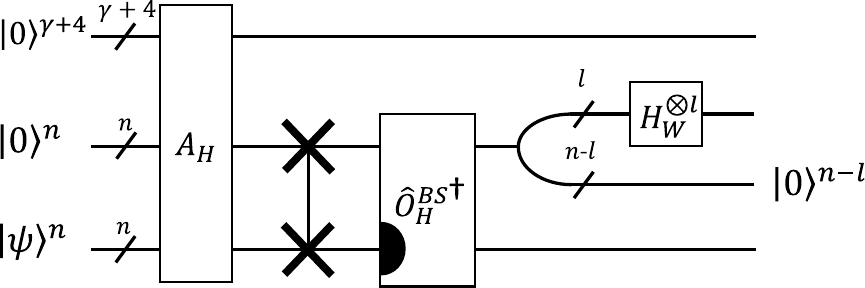}
    \caption{The explicit design of quantum circuit for block-encoding of the Hamiltonian (\ref{eq:numeric_main_Hamiltonian}). The unitary $\hat{O}^{BS\dagger}_H$ is a Hermit conjugated $\hat{O}^{BS}_H$, see Appendix \ref{Appendix Banded-sparse-access}. Again the filled semicircle determines which register doesn't change under $\hat{O}^{BS}_H$ and $\hat{O}^{BS\dagger}_H$. Note that according to the Definition~\ref{Banded-sparse-access} $\hat{O}^{BS}_H=\hat{O}^{BS}_{H^T}$, that's why we omitted the transpose notation for $\hat{O}^{BS\dagger}_H$. We indicated pure ancillas by showing that after the action the register's state becomes $\ket{0}^{n-l}$.}
    \label{fig:Block-encoding_qc}
    \end{figure}
    
    \begin{proof} As a proof we simply present the Fig.~\ref{fig:Block-encoding_qc} which explicity shows how to construct the desired block-encoding using the auxiliary construction from the Corollary~\ref{corollary:Hamiltonian block-encoding auxiliary}. We exploit Banded-sparse-access to the Hamiltonian in Fig.\ref{fig:Block-encoding_qc} which we built using the technique described in the Appendix \ref{Appendix Banded-sparse-access}. As an input for the algorithm we need set of non-zero indexes in the first row of the Hamiltonian. This knowledge is given by degrees of the momentum operator contributing (\ref{eq:main_Hamiltonian}). Using the periodic spatial boundaries (\ref{eq:x_p_cyclic_matrices}) the sparsity of the Hamiltonian is no greater than $\max\limits_{k}2m_k+1$.
    \end{proof}

    \subsection{Evolution operator creation}
    In this subsection we present one of the promising ways how to exploit the efficient block-encoding declared in Theorem~\ref{theorem:Hamiltonian Block-encoding}. Now we apply Theorem~\ref{theorem: optimal block-hamiltonian simultaions} to the designed block-encoding.
    \begin{theorem}[Nearly-optimal Hamiltonian simulation]\label{theorem:Nearly-optimal Hamiltonian simulation}
    Let $t\in\mathbb{R}\slash \{0\}$, $\epsilon\in (0,1)$ and let $U_H$ be an $(\sqrt{2^{l}}\mathcal{N}_H,l+\gamma+4,0)$-block-encoding of the $2^n\times2^n$ Hamiltonian H as in the Theorem~\ref{theorem:Hamiltonian Block-encoding}. Then we can implement an $\epsilon$-precise Hamiltonian simulation unitary $V$ which is an $(2,l+\gamma+6,\epsilon)$-block-encoding of $e^{itH}$, with resources no greater than: 
    \begin{enumerate}
        \item $
305 \times 2^{7 + \gamma} n^2 q \Omega(\alpha t, \epsilon) + 33 \times 2^{4 + l + \gamma} n \Omega(\alpha t, \epsilon) + 5 \times 2^{7 + \gamma} n \Omega(\alpha t, \epsilon) + 2^{5 + l} n \Omega(\alpha t, \epsilon) + 32 n \Omega(\alpha t, \epsilon) + 17 l \Omega(\alpha t, \epsilon) + 66 \Omega(\alpha t, \epsilon) + 2^{5 + \gamma} \gamma \Omega(\alpha t, \epsilon) + 35 \times 2^{7 + l + \gamma} n + 21 \times 2^{8 + \gamma} n + 33 \times 2^{3 + l} n + 320 n + 323 \times 2^{10 + \gamma} n^2 q + 2 l - 563 \times 2^{5 + \gamma} n q \Omega(\alpha t, \epsilon) - 119 \times 2^{4 + \gamma} q \Omega(\alpha t, \epsilon) - 115 \times 2^{2 + l + \gamma} \Omega(\alpha t, \epsilon) - 239 \times 2^{2 + \gamma} \Omega(\alpha t, \epsilon) - 3 \times 2^{4 + l} \Omega(\alpha t, \epsilon) - 2385 \times 2^{6 + \gamma} n q - 503 \times 2^{5 + \gamma} q - 507 \times 2^{3 + l + \gamma} - 1005 \times 2^{3 + \gamma} - 3 \times 2^{5 + l} - 476
$
        \item $
987 \times 2^{5 + \gamma} n^2 q \Omega(\alpha t, \epsilon) + 107 \times 2^{2 + l + \gamma} n \Omega(\alpha t, \epsilon) + 2^{9 + \gamma} n \Omega(\alpha t, \epsilon) + 25 \times 2^l n \Omega(\alpha t, \epsilon) + 32 n \Omega(\alpha t, \epsilon) + 12 l \Omega(\alpha t, \epsilon) + 38 \Omega(\alpha t, \epsilon) + 3 \times 2^{3 + \gamma} \gamma \Omega(\alpha t, \epsilon) + 453 \times 2^{3 + l + \gamma} n + 17 \times 2^{8 + \gamma} n + 107 \times 2^{1 + l} n + 256 n + 4181 \times 2^{6 + \gamma} n^2 q + 2 l - 911 \times 2^{4 + \gamma} n q \Omega(\alpha t, \epsilon) - 3 \times 2^{9 + \gamma} q \Omega(\alpha t, \epsilon) - 49 \times 2^{3 + l + \gamma} \Omega(\alpha t, \epsilon) - 383 \times 2^{1 + \gamma} \Omega(\alpha t, \epsilon) - 3859 \times 2^{5 + \gamma} n q - 407 \times 2^{5 + \gamma} q - 409 \times 2^{3 + l + \gamma} - 1627 \times 2^{2 + \gamma} - 3 \times 2^{5 + l} - 384
        $ C-NOTs;
        \item $3n+\gamma-l-1$ pure ancillas;
    \end{enumerate}
    where we introduced $\alpha=2^{l+1}\mathcal{N}_H$ for the convenience. The coefficient $\Omega(\alpha t,\epsilon)$ is implicitly defined by the condition $g\geq \Omega(\alpha t,\epsilon)$ and can be retrieved from
        \[ \frac{1.07}{\sqrt{g}}\left(\frac{\alpha et}{2g}\right)^g\leq\epsilon. \]
    The parameter has scaling
    \[ \Omega(\alpha t,\epsilon)=\mathcal{O}\left(\alpha t+\frac{\ln(1/\epsilon)}{\ln(e+\frac{\ln(1/\epsilon)}{\alpha t})}\right). \]
    \end{theorem}
    \begin{proof}
        The proof is well described in the \cite{gilyen2019quantum}, where authors provide a decomposition of the desired evolution operator to multi-controlled Pauli rotations $R_z$ and applications of the block-encodingm which is very similar to one described in Theorem~\ref{theorem: Singular value transformation by real polynomials}. We provided the certain number of gates for our case based on the quantum schemes from the paper (Theorem~\ref{theorem:Hamiltonian Block-encoding}).
    \end{proof}
    
    Thus, we prove that the nearly-optimal Hamiltonian simulation is achievable by providing the step-by-step guide how to construct an efficient block-encoding and the evolution operator for the Hamiltonian with form (\ref{eq:numeric_main_Hamiltonian}). Moreover, we claim that the considered case is a basement for more general algorithm. That is why in next Section we generalize our approach to different boundaries and finite-difference schemes. In addition, we show how our approach can be straightforwardly extended to a multi-dimensional case (several qumodes). These techniques all together allow us to build a general toolbox for constructing block-encodings.

 \section{Method generalisation}\label{Sec:Generalisation}
    
    The robustness of a proposed computational method is often gauged by its flexibility and the breadth of its applicability. In this section, we explore how the algorithm described in the Section \ref{main_section} can be straightforwardly extended to the multi-dimensional case. In this subsection we clarify an implementation of each step from the road-map (see Fig.~\ref{fig:road_map}) for the multi-dimensional case.
    
    Suppose we have $d$-qumodes model
     \begin{eqnarray}
    \begin{gathered}
        H^{(d)}=\sum_{k=0}^{\eta_d-1}\alpha_{k}\prod_{y=0}^{d-1}\otimes R(L_{ky},P_{q^{(y)}_k}(x_y),p_y^{m^{(y)}_k});\\
        R(0,P_{q^{(y)}_k}(x_y),p_y^{m^{(y)}_k})=P_{q^{(y)}_k}(x_y)p_y^{m^{(y)}_k};\quad R(1,P_{q^{(y)}_k}(x_y),p_y^{m^{(y)}_k})=p_y^{m^{(y)}_k}P_{q^{(y)}_k}(x_y),
        \label{eq:multi-dimensional_Hamiltonian}
        \end{gathered}
    \end{eqnarray}
    where $x_y$/$p_y$ is the coordinate/momentum operator acting on $y$-qumode; we introduce set $L$ and function $R$ to go through possible permutations of $P_q$ and $p^m$ for each dimension. We assume that the coefficients $\alpha$ and $L$ are such that $H^{(d)}$ is Hermitian. For each spatial dimension we introduce two $n_y$-qubit register encoding row and column numbers. Similarly, we introduce finite-difference representation of coordinate and momentum operators $\hat{x}$, $\hat{p}$ as in the Eq. (\ref{eq:numeric_main_Hamiltonian}),(\ref{eq:x_p_cyclic_matrices}). Moreover, we do not specify spatial boundaries conditions and finite-difference representation of operators for each qumode; we suggest that they can differ for each dimension (see subsequent subsections).
    
    Now let us consider each term of the (\ref{eq:multi-dimensional_Hamiltonian}) independently as in the Section \ref{main_section}. For an individual term we can build a block-encoding auxiliary from the Lemma~\ref{lemma: 1-term block-encoding auxiliary} independently acting on each spatial dimension register
    \begin{eqnarray}
    \begin{gathered}
        A^{(L_k)}_{q_km_k}\sum_{j_0,\dots,j_{n_{d-1}}}\sigma_{j_0,\dots,j_{n_{d-1}}}\left(\ket{0}^3\ket{0}^{n_0}\ket{j_0}^{n_0}\dots\ket{0}^3\ket{0}^{n_{d-1}}\ket{j_{d-1}}^{n_{d-1}}\right)\\=\sum_{j_0,\dots,j_{n_{d-1}}}\sigma_{j_0,\dots,j_{n_{d-1}}}\prod_{y=0}^{d-1}\otimes \left[A_{R(L_{ky},P_{q^{(y)}_k}(x_y),p_y^{m^{(y)}_k})}\ket{0}^3\ket{0}^{n_y}\ket{j_y}^{n_y}\right]\\
        =\frac{1}{\prod_{\sigma=0}^{d-1}\mathcal{N}_{m_k(\sigma)}}\sum\limits_{\substack{i_0=0,\dots,2^{n_0}-1\\ j_0\in F^{0}_{i_0}}}\sum\limits_{\substack{i_1=0,\dots,2^{n_1}-1\\ j_1\in F^{1}_{i_1}}}\dots\sum\limits_{\substack{i_{d-1}=0,\dots,2^{n_{d-1}}-1\\ j_{d-1}\in F^{d-1}_{i_{d-1}}}}
        \left[\prod_{y=0}^{d-1}R(L_{ky},P_{q^{(y)}_k}(\hat{x}_y),\hat{p}_y^{m^{(y)}_k})_{i_yj_y}\right]\\
        \times\sigma_{j_0,\dots,j_{n_{d-1}}}\prod_{y=0}^{d-1}\otimes\ket{0}^3\ket{i_y}^{n_y}\ket{j_y}^{n_y}+\dots,
    \end{gathered}\label{eq:A_xp_multi_dimensional}
    \end{eqnarray}
    where $F^y_i$ indicates a set of non-zero elements in the $i$-th row of $\hat{p}^{m_k^{(y)}}$; the normalization constant $\mathcal{N}_{m_k(\sigma)}=\sqrt{2^{l_{k(\sigma)}}\mathcal{N}_{p^{m_k(\sigma)}}}$ for certain and $m_k(\sigma)$ coincides to one in Eq.~\ref{eq:1term block encoding auxiliarry}.
    In Eq.~\ref{eq:A_xp_multi_dimensional} we omitted the orthogonal vectors for simplicity, so the symbol $"\dots"$ conceals some orthogonal states. Eq.~\ref{eq:A_xp_multi_dimensional} shows that the terminal wave function is factorized. Similarly to one dimensional case, combining the Superposition unitaries (\ref{eq:A_xp_multi_dimensional}) we suggest a way to implement the multi-dimensional Hamiltonian block-encoding auxiliary (similar to Corollary~\ref{corollary:Hamiltonian block-encoding auxiliary}).
    
    \begin{corollary}[Multi-dimensional Hamiltonian block-encoding auxiliary]\label{corollary:multi-dimensional Hamiltonian superposition}
    Let $H^{(d)}$ be a Hamiltonian as in the Eq.~\ref{eq:multi-dimensional_Hamiltonian} with the number of terms $\eta_d=2^\gamma$, matrix representation $\hat{H}^{(d)}$, and the set $L$ then the unitary operation $A^{(d)}_H$
         \[ \begin{gathered}
         A^{(d)}_H\ket{0}^{\gamma}\left(\sum_{j_0,\dots,j_{d-1}}\phi_{j_0,\dots,j_1}\prod_{y=0}^{d-1}\otimes\ket{0}^3\ket{0}^{n_y}\ket{j_y}^{n_y}\right)=\left({U^{H^{(d)}}_\alpha}^T\otimes I^{\otimes \sum_{y=0}^{d-1}(2n_y+1)}\right)\left(\prod_{k=0}^{\eta_d-1}C_{A^{(L_k)}_{q_km_k}}^{k}\right)\left[U^{H^{(d)}}_\alpha\ket{0}^{\gamma}\right]\\\times 
         \left(\sum_{j_0,\dots,j_{d-1}}\phi_{j_0,\dots,j_1}\prod_{y=0}^{d-1}\otimes\ket{0}^3\ket{0}^{n_y}\ket{j_y}^{n_y}\right)\\=
         \frac{1}{\mathcal{N}^{(d)}_H}\sum_{i_0=0}^{2^{n_0}-1}\sum_{i_1=0}^{2^{n_1}-1}\dots\sum_{i_{d-1}=0}^{2^{n_{d-1}}-1}\sum_{j_0,j_1,\dots,j_{d-1}\in F_{i_0,i_1,\dots,i_{d-1}}^H}\hat{H}^{(d)}_{i_0i_1\dots i_{d-1}j_0j_1\dots j_{d-1}}\ket{0}^{\gamma}\phi_{j_0,\dots,j_1}\prod_{y=0}^{d-1}\otimes\ket{0}^3\ket{i_y}^{n_y}\ket{j_y}^{n_y}+\dots\\
         \rightarrow \frac{1}{\mathcal{N}^{(d)}_H}\sum_{i_0=0}^{2^{n_0}-1}\sum_{i_1=0}^{2^{n_1}-1}\dots\sum_{i_{d-1}=0}^{2^{n_{d-1}}-1}\sum_{j_0,j_1,\dots,j_{d-1}\in F_{i_0,i_1,\dots,i_{d-1}}^H}\hat{H}^{(d)}_{i_0i_1\dots i_{d-1}j_0j_1\dots j_{d-1}}\ket{0}^{\gamma+3d}\phi_{j_0,\dots,j_1}\prod_{y=0}^{d-1}\otimes\ket{i_y}^{n_y}\ket{j_y}^{n_y}+\dots\\=
         \frac{1}{\mathcal{N}^{(d)}_H}\sum_{i_0,j_0=0}^{2^{n_0}-1}\sum_{i_1,j_1=0}^{2^{n_1}-1}\dots\sum_{i_{d-1},j_{d-1}=0}^{2^{n_{d-1}}-1}\hat{H}^{(d)}_{i_0i_1\dots i_{d-1}j_0j_1\dots j_{d-1}}\ket{0}^{\gamma+3d}\phi_{j_0,\dots,j_1}\prod_{y=0}^{d-1}\otimes\ket{i_y}^{n_y}\ket{j_y}^{n_y}+\dots;
         \end{gathered}\]
         \[\mathcal{N}^{(d)}_H=\sum_{k=0}^{\eta_d-1}\abs{\alpha_{k}}\prod_{\sigma=0}^{d-1}\mathcal{N}_{m_k(\sigma)}\]
    with $\mathcal{O}\left(\sum_{k=0}^{\eta_d-1}\sum_{y=0}^{d-1}(q_{k}^{(y)}n_y^2+m^{(y)}_kn_y)\right)$ gates and $\mathcal{O}\left(\sum_{y=0}^dn_y\right)$ pure ancillas. We use the notation $F_{i_0,i_1,\dots,i_{d-1}}^H$ for a set of non-zero indexes of non-zero elements in the row with index $i_0,i_1,\dots,i_{d-1}$ of the Hamiltonian $H^{(d)}$. This underlines the entanglement between spatial dimensions registers which arises here. However, we stress out that the other elements of the Hamiltonian are still zeros.
    
    The orthogonal state is omitted for simplicity; however, it can be accurately retrieved if it is needed. Similarly, the procedure exploits the linear combination theorem applied on $A^{(L_k)}_{q_km_k}$. The quantum circuit for the $A^{(d)}_H$ is very similar to one depicted in Fig.~\ref{fig:A_H quantum_circuit}. The auxiliary unitary $U^{H^{(d)}}_\alpha$ is given by
    
         \[ U^{H^{(d)}}_\alpha\ket{0}^{\gamma}=\frac{1}{\sqrt{\mathcal{N}^{(d)}_H}}\sum^{2^{\gamma}-1}_{k=0}\sqrt{i^{\sum_\sigma m_k(\sigma)}\alpha_{k}\left(\prod_{\sigma=0}^{d-1}\mathcal{N}_{m_k(\sigma)}\right)}\ket{k}^{\gamma}.\]  
    
    \end{corollary}
    
    The next step from our road-map (see Fig.~\ref{fig:road_map}) is to create the Block-encoding for the Hamiltonian (\ref{eq:multi-dimensional_Hamiltonian}) which can be fulfilled in analogous way to the Theorem \ref{theorem:Hamiltonian Block-encoding}.
    
    The final step is to build a block-encoding for the multidimensional Hamiltonian $H^{(d)}$ which is a non-trivial step in this generalisation to a multi-dimensional case since we need to consider the multidimensional sparsity encoded by $F^H_{i_0,i_1,\dots,i_{d-1}}$. In particular, to realise the scheme as in the Fig.~\ref{fig:Block-encoding_qc} we should build a multidimensional Banded-sparse-access. The view of (\ref{eq:multi-dimensional_Hamiltonian}) and matrix representation of momentum operator $\hat{p}_y$ allow us to accurately retrieve the set $F^H_{i_0,i_1,\dots,i_{d-1}}$. Thus, all the ingredients for the multidimensional block-encoding are prepared.

    \begin{theorem}[Multi-dimensional block-encoding]\label{theorem:Multi-dimensional block encoding}
        Let $H^{(d)}$ be a Hamiltonian as in Eq.~\ref{eq:multi-dimensional_Hamiltonian_first} with the number of terms $\eta_d=2^\gamma$, sparsity $2^l$, $n=\sum_{y}n_y$-qubit matrix representation $\hat{H}^{(d)}$, and normalization constant $\mathcal{N}^{(d)}_H$ (see Corollary~\ref{corollary:multi-dimensional Hamiltonian superposition}) then we can implement a $(\sqrt{2^{l}}\mathcal{N}^{(d)}_H,l+\gamma+3d,0)$-block-encoding of $\hat{H}$ with $\mathcal{O}\left(\sum_{k=0}^{\eta_d-1}\sum_{y=0}^{d-1}(q_k^{(y)}n_y^2+m^{(y)}_kn_y)\right)$ gates and $\mathcal{O}\left(\sum_{y=0}^dn_y+\log(\eta)\right)$ pure ancillas.
    \end{theorem}
    
    \begin{proof}
        As a proof we address the Fig.~\ref{fig:Block-encoding_qc} which precisely describes the quantum circuit for the multi-dimensional case if $\hat{O}^{BS}_H$, $A_H$ are substituted by the multidimensional analogues $\hat{O}^{BS}_{H^{(d)}}$, $A^{(d)}_H$. Both of them described earlier in the Subsection.
    \end{proof}
    
    As demonstrated, The batch of techniques described in the Section~\ref{main_section} combined together allows us to extend our approach for a wide class of boundary conditions, finite-difference schemes, and the number of dimensions. Now, using the multidimensional block-encoding as as input for Theorem~\ref{theorem: optimal block-hamiltonian simultaions}, we suggest construction a straightforward quantum circuit for the evolution operator.

    \section{Conclusion}\label{sec:conclusion}
     In this paper, we have advanced the methodology for explicit gate construction for block-encoding of Hamiltonians, crucial for quantum simulations of linear partial differential equations. Leveraging the foundational concept of block-encoding, our approach refines its application to enhance both theoretical and practical aspects of quantum simulations by providing the algorithm which is supreme compare to classical state of the art analogues for solving PDEs.

    The core of our contribution is the development of an efficient gate construction strategy for block-encoding of
    \begin{eqnarray}
    \begin{gathered}
        H^{(d)}=\sum_{k=0}^{\eta_d-1}\alpha_{k}\prod_{y=0}^{d-1}\otimes R(L_{ky},P_{q^{(y)}_k}(x_y),p_y^{m^{(y)}_k});\\
        R(0,P_{q^{(y)}_k}(x_y),p_y^{m^{(y)}_k})=P_{q^{(y)}_k}(x_y)p_y^{m^{(y)}_k};\quad R(1,P_{q^{(y)}_k}(x_y),p_y^{m^{(y)}_k})=p_y^{m^{(y)}_k}P_{q^{(y)}_k}(x_y),
        \end{gathered}
    \end{eqnarray}
    which demonstrates the potential for significant computational speedups over traditional classical finite-difference methods in terms of solving a wide class of linear partial differential equation with coefficients which are polynomial functions. By detailing the construction of these encodings with efficient one- and two-qubit operations, we provide a practical blueprint for implementing quantum simulations on currently emerging digital quantum hardware.
    
    This work not only fills an important gap by providing efficient way for for practical implementations but also sets a foundation for further research into optimizing these methods for a wider range of Hamiltonians. As future work we aim to extend our approach to a wider class of numerical schemes and extend the applicability to wider classes of boundary conditions. We will also examine important partial differential equations individually.

\section{Tables}\label{sec:tables}
Here we present our table of notations for both the one-dimensional and multidimensional case. For convenience of the reader, we also present a table summarising the exact one and two-qubit count for each part involved in our protocol.

\begin{table}
\caption{Table of notations.}\label{table_of_notations}
\makegapedcells
\begin{tabular}{
  |p{\dimexpr.095\linewidth-2\tabcolsep-1.3333\arrayrulewidth}
  |p{\dimexpr.4\linewidth-2\tabcolsep-1.3333\arrayrulewidth}
  |p{\dimexpr.4\linewidth-2\tabcolsep-1.3333\arrayrulewidth}|%
  |p{\dimexpr.095\linewidth-2\tabcolsep-1.3333\arrayrulewidth}|
  }
  \hline
  \centering \textbf{Notation}     & \centering \textbf{Name} & \centering \textbf{Formula}     & \centering \arraybackslash 
 \textbf{First app.}    \\ \hline

  \centering $\hat{H}$ & \centering Finite-difference representation of $H$& \centering $\hat{H}=\sum_{k=0}^{\eta-1}\left(\alpha_{k}P_{q_k}(\hat{x})\hat{p}^{m_k}+\alpha^*_{k}\hat{p}^{m_k}P_{q_k}(\hat{x})\right)$ &  \centering\arraybackslash Eq.~\ref{eq:numeric_main_Hamiltonian}\\ \hline
  
  \centering  $\hat{x} $& \centering  Finite-difference representation of coordinate operator $x$ over the interval $(a,b)$ &\centering   $\hat{x}=\left(\begin{array}{ccccc}
    a & 0  & \dotsm & 0 & 0 \\
    0 & a+\Delta x  &\dotsm& 0 & 0\\
    \rotatebox[origin=c]{270}{\dots}&&\rotatebox[origin=c]{-45}{\dots}&&\rotatebox[origin=c]{270}{\dots}\\ 
    0 & 0 &\dotsm & b-\Delta x &0\\
    0 & 0 &\dotsm & 0 &b\\
    \end{array}
    \right)$ &\centering\arraybackslash Eq.~\ref{eq:x_p_cyclic_matrices}\\ \hline
  
  \centering $\hat{p}$ &\centering Finite-difference representation of momentum operator $p$ for periodic spatial boundary conditions and the central symmetrical finite-difference scheme &\centering   $\hat{p}=-\frac{i}{2\Delta x}\left(\begin{array}{ccccccc}
    0 & 1 & 0& \dotsm &0& 0 & -1 \\
    -1 & 0 & 1&\dotsm& 0&0 & 0\\
    \rotatebox[origin=c]{270}{\dots}&&&\rotatebox[origin=c]{-45}{\dots}&&&\rotatebox[origin=c]{270}{\dots}\\ 
    0 & 0&0 &\dotsm &-1& 0 &1\\
    1 & 0&0 &\dotsm &0& -1 &0\\
    \end{array}
    \right)$ &\centering\arraybackslash Eq.~\ref{eq:x_p_cyclic_matrices}\\ \hline

  \centering $C^b_U$ & Multi-control operator that apply $U$ on the target if the control is in $\ket{b}$ &\centering $C^b_U=\ket{b}^n\bra{b}^n\otimes U+\sum\limits_{\substack{i=0,\dots,2^n-1\\ i\neq b}}\ket{i}^n\bra{i}^n\otimes I^{\otimes m}$ &\centering\arraybackslash Def.~\ref{def:multiconrol operator}\\ \hline

  \centering $\hat{O}_A^{BS}$& \centering Banded-sparse-access to matrix $A$ &\centering $\hat{O}^{BS}_A\ket{0}^{n-l}\ket{s}^l\ket{i}^n:=\ket{r_{si}}^n\ket{i}^n$ &\centering\arraybackslash Def.~\ref{Banded-sparse-access}\\ \hline

  \centering $\hat{O}^S_{p^m}$& \centering Amplitude-unitary-oracle-access to the operator $\hat{p}^m$ &\centering $\hat{O}^S_{p^m}\ket{0}^1\ket{s}^l:=\frac{(\hat{p}^m)^{(s)}}{\sqrt{\mathcal{N}_{p^m}}}\ket{0}^1\ket{s}^l+\sqrt{1-\frac{\abs{(\hat{p}^m)^{(s)}}^2}{\mathcal{N}_{p^m}}}\ket{1}^1\ket{s}^l$ &\centering\arraybackslash Eq.~\ref{momentum oracle}\\ \hline

  \centering $\mathcal{N}_{p^m}$& \centering Normalization constant for momentum operator &\centering $\mathcal{N}_{p^m}\geq\norm{\hat{p}^m}^2_{\max}$ &\centering\arraybackslash Eq.~\ref{momentum oracle}\\ \hline
  
  \centering $U_{x^q}$& \centering Coordinate-superposition-unitary &\centering $U_{x^q}\ket{0}^n=\frac{1}{\sqrt{\sum_{\kappa=0}^{2^n-1}\hat{x}_{\kappa\kappa}^{2q}}}\sum_{i=0}^{2^n-1}\hat{x}_{ii}^{q}\ket{i}^n$ &\centering\arraybackslash Def.~\ref{def:Coordinate-superposition-unitary}\\ \hline
  
  \centering $A_{P_qp^m}$& \centering $1$-term block-encoding auxiliary $P_q(\hat{x})\hat{p}^m$&\centering $A_{P_qp^m}\sum_{j=0}^{2^n-1}\sigma_j\ket{0}^3\ket{0}^n\ket{j}^n=\frac{(-1)^m}{\sqrt{2^l\mathcal{N}_{p^m}}} 
   $ $\sum\limits_{\substack{i=0,\dots,2^n-1\\ j\in F^{(m)}_i}}\underbrace{P_q(\hat{x}_{ii})(\hat{p}^m)_{ij}}_{(P_q(\hat{x})\hat{p}^m)_{ij}}\sigma_j\ket{0}^1\ket{i}^n\ket{j}^n +\dots$ &\centering\arraybackslash Eq.~\ref{eq:1term block encoding auxiliarry}\\ \hline

   \centering $A_{p^mP_q}$& \centering $1$-term block-encoding auxiliary $\hat{p}^mP_q(\hat{x})$&\centering $A_{p^mP_q}\sum_{j=0}^{2^n-1}\sigma_j\ket{0}^3\ket{0}^n\ket{j}^n=\frac{(-1)^m}{\sqrt{2^l\mathcal{N}_{p^m}}} 
   $ $\sum\limits_{\substack{i=0,\dots,2^n-1\\ j\in F^{(m)}_i}}\underbrace{(\hat{p}^m)_{ij}P_q(\hat{x}_{jj})}_{(\hat{p}^mP_q(\hat{x}))_{ij}}\sigma_j\ket{0}^1\ket{i}^n\ket{j}^n +\dots$ &\centering\arraybackslash Eq.~\ref{eq:1term block encoding auxiliarry}\\ \hline
  
  
  \centering $F_i^{(m)}$ & \centering Set of indexes of non-zero elements in the $i$-th row of $\hat{p}^m$ &\centering $None$ &\centering\arraybackslash Eq.~\ref{eq:1term block encoding auxiliarry}\\ \hline
  
  \centering $A_H$& \centering Block-encoding auxiliary &\centering $(\bra{0}^{\gamma+4}\otimes I^{\otimes 2n})A_H\sum_{j=0}^{2^n-1}\sigma_j\ket{0}^{\gamma+4}\ket{0}^n\ket{j}^n\sim \sum\limits_{\substack{i=0,\dots,2^n-1\\ j\in F^H_i}}\hat{H}_{ij}\ket{i}^n\ket{j}^n$ &\centering\arraybackslash Cor.~\ref{corollary:Hamiltonian block-encoding auxiliary}\\ \hline
  
  \centering $F_i^H$& \centering Set of indexes of non-zero elements in the $i$-th row of matrix $\hat{H}$&\centering $None$ &\centering\arraybackslash Cor.~\ref{corollary:Hamiltonian block-encoding auxiliary}\\ \hline
  
  \centering $\mathcal{N}_H$& \centering Normalization constant for $A_H$ &\centering $2\sum_{\kappa=0}^{2^\gamma -1}\abs{\alpha_k}\sqrt{2^{l_\kappa}\mathcal{N}_{p^{m_k}}}$ &\centering\arraybackslash Cor.~\ref{corollary:Hamiltonian block-encoding auxiliary}\\ \hline
  
  
  \centering $U_H$& \centering Block-encoding of the Hamiltonian $H$&\centering $U_H\ket{0}^\lambda\ket{\phi}^n=\Phi_H\sum\limits_{\substack{i=0,\dots,2^n-1\\ j\in F^H_i}}\left(\hat{H}_{ij}\sigma_j\ket{0}^\lambda
        +J^U_{\hat{H}_{ij}}\ket{\bot_{0}}^\lambda\right)\ket{i}^n$ &\centering\arraybackslash Thm.~\ref{theorem:Hamiltonian Block-encoding}\\ \hline
\end{tabular}

\end{table}
\begin{sidewaystable}
\centering
\caption{Table of circuit's complexity.(Here, $n$ denotes the number of qubits such that $2^n$ corresponds to the size of the Hamiltonian $\hat{H}$ in Eq.(\ref{eq:numeric_main_Hamiltonian}) and $\gamma$ denotes the number of qubits such that $\eta=2^\gamma$ is the number of terms in Eq.(\ref{eq:numeric_main_Hamiltonian}). And $l$ is also the number of qubits such that $2^l$ indicates the sparsity of the corresponding matrix. Such as, for $\hat{O}^{BS}_A$, the Banded-sparse-access to matrix to matrix $A$, $2^l$ is the sparsity of matrix $A$. Furthermore, we define $q_{\text{max}} = \max_k q_k$, $2^{l_{\text{max}}} = \max_k m_k = m_{\text{max}}$, $2^{l'{\text{max}}} = \max_k (2m_k+1)$, where $q_k$ and $m_k$ are the exponent of finite difference matrices for the coordinate $\hat{x}$ and momentum operator $\hat{p}$ of Hamiltonian $\hat{H}$ in Eq.(\ref{eq:numeric_main_Hamiltonian}) respectively. )}\label{table_scaling}
\makegapedcells
\begin{tabular}{
   |p{\dimexpr.233\linewidth-2\tabcolsep-1.3333\arrayrulewidth}
   |p{\dimexpr.075\linewidth-2\tabcolsep-1.3333\arrayrulewidth}
   |p{\dimexpr.11\linewidth-2\tabcolsep-1.3333\arrayrulewidth}
  |p{\dimexpr.26\linewidth-2\tabcolsep-1.3333\arrayrulewidth}
  |p{\dimexpr.22\linewidth-2\tabcolsep-1.3333\arrayrulewidth}%
  |p{\dimexpr.07\linewidth-2\tabcolsep-1.3333\arrayrulewidth}
  |p{\dimexpr.065\linewidth-2\tabcolsep-1.3333\arrayrulewidth}|
  }
     \hline
     \centering\textbf{Name}&\centering\textbf{Notation}&\centering\textbf{Scaling}& \centering\textbf{C-NOTS}&\centering\textbf{One-qubits operations} & \centering\textbf{Ancillas}& \textbf{Main Thm.} \\
     \hline
     
     Banded-sparse-access to matrix $A$&$\hat{O}_A^{BS}$& $\mathcal{O}(mn)$&$(2^l+1)(32n-48)$  &$25 \times 2^l n - 36 \times 2^l + 32 n - 48$&$n-1$ & Lem. \ref{Banded-sparse-access lemma}\\
     \hline
     
     Coordinate-superposition-unitary&$U_{x^q}$&$\mathcal{O}\left(\frac{n^{q+1}}{q!}\right)$&$\frac{12n^{q+1}}{q!}$&$\frac{16n^{q+1}}{q!}$&$n-1$& Thm. \ref{theorem Finite binary norm Ansatz} \\
     \hline 
     
     Amplitude-oracle for $\hat{x}$&$\hat{O}_x$&$\mathcal{O}\left(n^2\right)$&$1864n^2-860n-92$&$2304n^2-1064n-109$&$2n-1$& Cor. \ref{Corollary: x-oracle}
     \\
     \hline
     
     Coordinate-polynomial-oracle &$\hat{O}_{P_q}$&$\mathcal{O}\left(qn^2\right)$&$1864qn^2-860qn-92q$&$2304qn^2-1064qn-108q$&$2n-1$& Thm. \ref{theorem: Singular value transformation by real polynomials}
     \\
     \hline

     Sparse-amplitude-oracle for the operator $\hat{p}^m$&$\hat{O}_{p^m}^S$&$\mathcal{O}(m)$&$2^l$&$2^l$&$0$& Lem. \ref{eq:x_p_cyclic_matrices} \\\hline

     $1$-term block-encoding auxiliary & $A_{P_qp^m}$ $(A_{p^mP_q})$ & $\mathcal{O}\left(qn^2+mn\right)$ & $1864qn^2-860qn+25\times 2^ln+32n-92q-17\times 2^l-48$ & $2304qn^2-1064qn+32\times 2^ln+32n-108q-48\times 2^l+2^l-48$ & $2n-1$ & Lem.~\ref{lemma: 1-term block-encoding auxiliary} \\\hline

     Block-encoding auxiliary & $A_H$ & $\mathcal{O}\bigl(\eta q_{\max} n^2 + \eta m_{\max} n\bigr)$ &  $2^{\gamma+1}\bigl[15792n^2q+107\times 2^{l_{\max}+1}+256n+12\gamma-7288nq-768q-49\times 2^{l_{\max}+1}-383\bigr]$ & $2^{\gamma+2}\bigl[9760n^2q+33\times2^{l_{\max}+2}n$ $+160n+8\gamma-4504nq-476q-115\times 2^{l_{\max}}-239\bigr]$ & $2n+\gamma-1$ & Cor.~\ref{corollary:Hamiltonian block-encoding auxiliary} \\\hline
     
     Block-encoding of Hamiltonian $H$&$U_H$&$\mathcal{O}\bigl(\eta q_{\max} n^2 + \eta m_{\max} n\bigr)$&$2^{\gamma+1}\bigl[15792n^2q+107\times 2^{l_{\max}+1}+256n+12\gamma-7288nq-768q-49\times 2^{l_{\max}+1}-383\bigr]+25\times 2^ln-36\times 2^l+32n-48$&$2^{\gamma+2}\bigl[9760n^2q+33\times2^{l_{\max}+2}n+160n+8\gamma-4504nq-476q-115\times 2^{l_{\max}}-239\bigr]+2^{l+5}n-3\times 2^{l+4}+32n+l-48$&$3n+\gamma+l-1$& Thm. \ref{theorem:Hamiltonian Block-encoding} \\
     \hline
\end{tabular}
\end{sidewaystable}

\begin{table}[h!]
\caption{Table of notations for a multidimensional case.}\label{table_of_notations_generalization_section}
\makegapedcells
\begin{tabular}{
  |p{\dimexpr.095\linewidth-2\tabcolsep-1.3333\arrayrulewidth}
  |p{\dimexpr.2\linewidth-2\tabcolsep-1.3333\arrayrulewidth}
  |p{\dimexpr.6\linewidth-2\tabcolsep-1.3333\arrayrulewidth}|%
  |p{\dimexpr.095\linewidth-2\tabcolsep-1.3333\arrayrulewidth}|
  }
  \hline
  \centering \textbf{Notation}     & \centering \textbf{Name} & \centering \textbf{Formula}     & \centering \arraybackslash 
 \textbf{First app.}    \\ \hline

  \centering 
  
  \centering $H^{(d)}$ & \centering Multi-dimensional Hamiltonian &\centering $\sum_{k=0}^{\eta_d-1}\alpha_{k}\prod_{y=0}^{d-1}\otimes R(L_{ky},P_{q^{(y)}_k}(x_y),p_y^{m^{(y)}_k})$ &  \centering\arraybackslash Eq.~\ref{eq:multi-dimensional_Hamiltonian} \\ \hline

  \centering $R$ & \centering Auxiliary functions for permutations &\centering $R(0,P_{q^{(y)}_k}(x_y),p_y^{m^{(y)}_k})=P_{q^{(y)}_k}(x_y)p_y^{m^{(y)}_k}$; $R(1,P_{q^{(y)}_k}(x_y),p_y^{m^{(y)}_k})=p_y^{m^{(y)}_k}P_{q^{(y)}_k}(x_y)$ &  \centering\arraybackslash Eq.~\ref{eq:multi-dimensional_Hamiltonian} \\ \hline

    \centering $A^{(L_k)}_{q_km_k}$ & \centering Multi-dimensional $1$-term block-encoding auxiliary &\centering $A^{(L_k)}_{q_km_k}\sum_{j_0,\dots,j_{n_{d-1}}}\alpha_{j_0,\dots,j_{n_{d-1}}}\left(\prod_{y=0}^{d-1}\otimes\ket{0}^3\ket{0}^{n_0}\ket{j_0}^{n_0}\right)=$ $\frac{1}{\prod_{\sigma=0}^{d-1}\mathcal{N}_{m_k(\sigma)}}\sum\limits_{\substack{i_0=0,\dots,2^{n_0}-1\\ j_0\in F^{0}_{i_0}}}\sum\limits_{\substack{i_1=0,\dots,2^{n_1}-1\\ j_1\in F^{1}_{i_1}}}\dots\sum\limits_{\substack{i_{d-1}=0,\dots,2^{n_{d-1}}-1\\ j_{d-1}\in F^{d-1}_{i_{d-1}}}} \alpha_{j_0,\dots,j_{n_{d-1}}}
    $ $ \bigl[\prod_{y=0}^{d-1}R(L_{ky},P_{q^{(y)}_k}(\hat{x}_y),\hat{p}^{m_k^{(y)}})_{i_yj_y}\bigr]\prod_{y=0}^{d-1}\otimes\ket{0}^3\ket{i_y}^{n_y}\ket{j_y}^{n_y}+\dots$ &  \centering\arraybackslash Eq.~\ref{eq:A_xp_multi_dimensional} \\ \hline

  \centering $\mathcal{N}_{m_k(\sigma)}$ & \centering Normalization constant for $A_{P_{q_k^{(\sigma)}}p_\sigma^{m_k^{(\sigma)}}}$ &\centering $\sqrt{2^{l_{k(\sigma)}}\mathcal{N}_{p^{m_k(\sigma)}}}$ &\centering\arraybackslash Eq.~\ref{eq:A_xp_multi_dimensional}\\ \hline

  \centering $F^y_i$ & \centering Set of non-zero elements in the $i$-th row of $\hat{p}^{m_k^{(y)}}$ &\centering $None$ &  \centering\arraybackslash Eq.~\ref{eq:A_xp_multi_dimensional} \\ \hline
  \centering $A^{(d)}_H$ & \centering Multi-dimensional Hamiltonian block-encoding auxiliary &\centering $A^{(d)}_H\ket{0}^{\gamma}\left(\sum_{j_0,\dots,j_{d-1}}\phi_{j_0,\dots,j_1}\prod_{y=0}^{d-1}\otimes\ket{0}^3\ket{0}^{n_y}\ket{j_y}^{n_y}\right)=$ $\frac{1}{\mathcal{N}^{(d)}_H}\sum_{i_0,j_0=0}^{2^{n_0}-1}\sum_{i_1,j_1=0}^{2^{n_1}-1}\dots\sum_{i_{d-1},j_{d-1}=0}^{2^{n_{d-1}}-1}\hat{H}^{(d)}_{i_0i_1\dots i_{d-1}j_0j_1\dots j_{d-1}}\ket{0}^{\gamma+3d}$ $\hat{H}^{(d)}_{i_0i_1\dots i_{d-1}j_0j_1\dots j_{d-1}}\prod_{y=0}^{d-1}\otimes\ket{i_y}^{n_y}\ket{j_y}^{n_y}+\dots$ &  \centering\arraybackslash Cor.~\ref{corollary:multi-dimensional Hamiltonian superposition} \\ \hline
  \centering $\mathcal{N}^{(d)}_H$ & \centering Normalization constant for $A^{(d)}_H$ &\centering $\sum_{k=0}^{\eta_d-1}\abs{\alpha_{k}}\prod_{\sigma=0}^{d-1}\mathcal{N}_{m_k(\sigma)}$ &  \centering\arraybackslash Cor.~\ref{corollary:multi-dimensional Hamiltonian superposition} \\ \hline

\end{tabular}

\end{table}

\begin{table}[h!]
\caption{Table of circuit's complexity in multi-dimension.(Here, $n_y$ denotes the number of qubits register encoding row and column numbers in $y$-th dimension in Eq. (\ref{eq:multi-dimensional_Hamiltonian}). $q_k^{(y)}$ and $m_k^{(y)}$ are the exponent of coordinate operator and momentum operator in $y$-th dimesion $x_y$ and $p_y$ in Eq. \ref{eq:multi-dimensional_Hamiltonian}). )}\label{table_scaling_multi-dimensional case}
\makegapedcells
\begin{tabular}{
   |p{\dimexpr.3\linewidth-2\tabcolsep-1.3333\arrayrulewidth}
   |p{\dimexpr.1\linewidth-2\tabcolsep-1.3333\arrayrulewidth}
   |p{\dimexpr.34\linewidth-2\tabcolsep-1.3333\arrayrulewidth}
  |p{\dimexpr.15\linewidth-2\tabcolsep-1.3333\arrayrulewidth}
  |p{\dimexpr.09\linewidth-2\tabcolsep-1.3333\arrayrulewidth}|
  }
     \hline
     \centering\textbf{Name}&\centering\textbf{Notation}&\centering\textbf{Scaling} & \centering\textbf{Ancillas}& \textbf{Main Thm.} \\
     \hline
     Multi-dimensional Hamiltonian block-encoding auxiliary &$A_H^{(d)}$&$\mathcal{O}\left(\sum_{k=0}^{\eta_d-1}\sum_{y=0}^{d-1}(q_{k}^{(y)}n_y^2+m^{(y)}_kn_y)\right)$&$\mathcal{O}\left(\sum_{y=0}^dn_y\right)$&Cor. \ref{corollary:multi-dimensional Hamiltonian superposition}\\\hline
     
     
     Multi-dimensional block-encoding of Hamiltonian & $H^{(d)}$&$\mathcal{O}\left(\sum_{k=0}^{\eta_d-1}\sum_{y=0}^{d-1}(q_k^{(y)}n_y^2+m^{(y)}_kn_y)\right)$ & $\mathcal{O}\bigl(\sum_{y=0}^dn_y+\log(\eta)\bigr)$ & Thm. \ref{theorem:Multi-dimensional block encoding} \\
     \hline
\end{tabular}
\end{table}

\section*{Acknowledgments}

Nana Liu acknowledges funding from the Science and Technology Program of  
Shanghai, China (21JC1402900). Nana Liu is also supported by NSFC grants  
No. 12471411 and No. 12341104, the Shanghai Jiao Tong University  
2030 Initiative, and the Fundamental Research Funds for the  
Central Universities.

\bibliography{references}

\begin{thebibliography}{46}%
\makeatletter
\providecommand \@ifxundefined [1]{%
 \@ifx{#1\undefined}
}%
\providecommand \@ifnum [1]{%
 \ifnum #1\expandafter \@firstoftwo
 \else \expandafter \@secondoftwo
 \fi
}%
\providecommand \@ifx [1]{%
 \ifx #1\expandafter \@firstoftwo
 \else \expandafter \@secondoftwo
 \fi
}%
\providecommand \natexlab [1]{#1}%
\providecommand \enquote  [1]{``#1''}%
\providecommand \bibnamefont  [1]{#1}%
\providecommand \bibfnamefont [1]{#1}%
\providecommand \citenamefont [1]{#1}%
\providecommand \href@noop [0]{\@secondoftwo}%
\providecommand \href [0]{\begingroup \@sanitize@url \@href}%
\providecommand \@href[1]{\@@startlink{#1}\@@href}%
\providecommand \@@href[1]{\endgroup#1\@@endlink}%
\providecommand \@sanitize@url [0]{\catcode `\\12\catcode `\$12\catcode `\&12\catcode `\#12\catcode `\^12\catcode `\_12\catcode `\%12\relax}%
\providecommand \@@startlink[1]{}%
\providecommand \@@endlink[0]{}%
\providecommand \url  [0]{\begingroup\@sanitize@url \@url }%
\providecommand \@url [1]{\endgroup\@href {#1}{\urlprefix }}%
\providecommand \urlprefix  [0]{URL }%
\providecommand \Eprint [0]{\href }%
\providecommand \doibase [0]{https://doi.org/}%
\providecommand \selectlanguage [0]{\@gobble}%
\providecommand \bibinfo  [0]{\@secondoftwo}%
\providecommand \bibfield  [0]{\@secondoftwo}%
\providecommand \translation [1]{[#1]}%
\providecommand \BibitemOpen [0]{}%
\providecommand \bibitemStop [0]{}%
\providecommand \bibitemNoStop [0]{.\EOS\space}%
\providecommand \EOS [0]{\spacefactor3000\relax}%
\providecommand \BibitemShut  [1]{\csname bibitem#1\endcsname}%
\let\auto@bib@innerbib\@empty
\bibitem [{\citenamefont {Pour-El}\ and\ \citenamefont {Richards}(1982)}]{FEYNMAN}%
  \BibitemOpen
  \bibfield  {author} {\bibinfo {author} {\bibfnamefont {M.}~\bibnamefont {Pour-El}}\ and\ \bibinfo {author} {\bibfnamefont {I.}~\bibnamefont {Richards}},\ }\bibfield  {title} {\bibinfo {title} {Simulating physics with computers},\ }\href@noop {} {\bibfield  {journal} {\bibinfo  {journal} {International Journal of Theoretical Physics}\ }\textbf {\bibinfo {volume} {21}},\ \bibinfo {pages} {553} (\bibinfo {year} {1982})}\BibitemShut {NoStop}%
\bibitem [{\citenamefont {Shor}(1994)}]{math_apply}%
  \BibitemOpen
  \bibfield  {author} {\bibinfo {author} {\bibfnamefont {P.~W.}\ \bibnamefont {Shor}},\ }\bibfield  {title} {\bibinfo {title} {Algorithms for quantum computation: discrete logarithms and factoring},\ }in\ \href@noop {} {\emph {\bibinfo {booktitle} {Proceedings 35th annual symposium on foundations of computer science}}}\ (\bibinfo {organization} {Ieee},\ \bibinfo {year} {1994})\ pp.\ \bibinfo {pages} {124--134}\BibitemShut {NoStop}%
\bibitem [{\citenamefont {Shor}(1999)}]{Shor}%
  \BibitemOpen
  \bibfield  {author} {\bibinfo {author} {\bibfnamefont {P.~W.}\ \bibnamefont {Shor}},\ }\bibfield  {title} {\bibinfo {title} {Polynomial-time algorithms for prime factorization and discrete logarithms on a quantum computer},\ }\href@noop {} {\bibfield  {journal} {\bibinfo  {journal} {SIAM Review}\ }\textbf {\bibinfo {volume} {41}},\ \bibinfo {pages} {303} (\bibinfo {year} {1999})}\BibitemShut {NoStop}%
\bibitem [{\citenamefont {Grover}(1996)}]{grover}%
  \BibitemOpen
  \bibfield  {author} {\bibinfo {author} {\bibfnamefont {L.~K.}\ \bibnamefont {Grover}},\ }\bibfield  {title} {\bibinfo {title} {A fast quantum mechanical algorithm for database search},\ }in\ \href@noop {} {\emph {\bibinfo {booktitle} {Proceedings of the twenty-eighth annual ACM symposium on Theory of computing}}}\ (\bibinfo {year} {1996})\ pp.\ \bibinfo {pages} {212--219}\BibitemShut {NoStop}%
\bibitem [{\citenamefont {Ambainis}(2004)}]{q_search}%
  \BibitemOpen
  \bibfield  {author} {\bibinfo {author} {\bibfnamefont {A.}~\bibnamefont {Ambainis}},\ }\bibfield  {title} {\bibinfo {title} {Quantum search algorithms},\ }\href@noop {} {\bibfield  {journal} {\bibinfo  {journal} {ACM SIGACT News}\ }\textbf {\bibinfo {volume} {35}},\ \bibinfo {pages} {22} (\bibinfo {year} {2004})}\BibitemShut {NoStop}%
\bibitem [{\citenamefont {Aspuru-Guzik}\ \emph {et~al.}(2005)\citenamefont {Aspuru-Guzik}, \citenamefont {Dutoi}, \citenamefont {Love},\ and\ \citenamefont {Head-Gordon}}]{aspuru2005simulated}%
  \BibitemOpen
  \bibfield  {author} {\bibinfo {author} {\bibfnamefont {A.}~\bibnamefont {Aspuru-Guzik}}, \bibinfo {author} {\bibfnamefont {A.~D.}\ \bibnamefont {Dutoi}}, \bibinfo {author} {\bibfnamefont {P.~J.}\ \bibnamefont {Love}},\ and\ \bibinfo {author} {\bibfnamefont {M.}~\bibnamefont {Head-Gordon}},\ }\bibfield  {title} {\bibinfo {title} {Simulated quantum computation of molecular energies},\ }\href@noop {} {\bibfield  {journal} {\bibinfo  {journal} {Science}\ }\textbf {\bibinfo {volume} {309}},\ \bibinfo {pages} {1704} (\bibinfo {year} {2005})}\BibitemShut {NoStop}%
\bibitem [{\citenamefont {Szegedy}(2004)}]{szegedy2004quantum}%
  \BibitemOpen
  \bibfield  {author} {\bibinfo {author} {\bibfnamefont {M.}~\bibnamefont {Szegedy}},\ }\bibfield  {title} {\bibinfo {title} {Quantum speed-up of markov chain based algorithms},\ }in\ \href@noop {} {\emph {\bibinfo {booktitle} {45th Annual IEEE symposium on foundations of computer science}}}\ (\bibinfo {organization} {IEEE},\ \bibinfo {year} {2004})\ pp.\ \bibinfo {pages} {32--41}\BibitemShut {NoStop}%
\bibitem [{\citenamefont {Harrow}\ \emph {et~al.}(2009)\citenamefont {Harrow}, \citenamefont {Hassidim},\ and\ \citenamefont {Lloyd}}]{harrow2009quantum}%
  \BibitemOpen
  \bibfield  {author} {\bibinfo {author} {\bibfnamefont {A.~W.}\ \bibnamefont {Harrow}}, \bibinfo {author} {\bibfnamefont {A.}~\bibnamefont {Hassidim}},\ and\ \bibinfo {author} {\bibfnamefont {S.}~\bibnamefont {Lloyd}},\ }\bibfield  {title} {\bibinfo {title} {Quantum algorithm for linear systems of equations},\ }\href@noop {} {\bibfield  {journal} {\bibinfo  {journal} {Physical review letters}\ }\textbf {\bibinfo {volume} {103}},\ \bibinfo {pages} {150502} (\bibinfo {year} {2009})}\BibitemShut {NoStop}%
\bibitem [{\citenamefont {Laughlin}\ and\ \citenamefont {Pines}(2000)}]{laughlin2000theory}%
  \BibitemOpen
  \bibfield  {author} {\bibinfo {author} {\bibfnamefont {R.~B.}\ \bibnamefont {Laughlin}}\ and\ \bibinfo {author} {\bibfnamefont {D.}~\bibnamefont {Pines}},\ }\bibfield  {title} {\bibinfo {title} {The theory of everything},\ }\href@noop {} {\bibfield  {journal} {\bibinfo  {journal} {Proceedings of the national academy of sciences}\ }\textbf {\bibinfo {volume} {97}},\ \bibinfo {pages} {28} (\bibinfo {year} {2000})}\BibitemShut {NoStop}%
\bibitem [{\citenamefont {Costa}\ \emph {et~al.}(2019)\citenamefont {Costa}, \citenamefont {Jordan},\ and\ \citenamefont {Ostrander}}]{costa2019quantum}%
  \BibitemOpen
  \bibfield  {author} {\bibinfo {author} {\bibfnamefont {P.~C.}\ \bibnamefont {Costa}}, \bibinfo {author} {\bibfnamefont {S.}~\bibnamefont {Jordan}},\ and\ \bibinfo {author} {\bibfnamefont {A.}~\bibnamefont {Ostrander}},\ }\bibfield  {title} {\bibinfo {title} {Quantum algorithm for simulating the wave equation},\ }\href@noop {} {\bibfield  {journal} {\bibinfo  {journal} {Physical Review A}\ }\textbf {\bibinfo {volume} {99}},\ \bibinfo {pages} {012323} (\bibinfo {year} {2019})}\BibitemShut {NoStop}%
\bibitem [{\citenamefont {Gaitan}(2020)}]{gaitan2020finding}%
  \BibitemOpen
  \bibfield  {author} {\bibinfo {author} {\bibfnamefont {F.}~\bibnamefont {Gaitan}},\ }\bibfield  {title} {\bibinfo {title} {Finding flows of a navier--stokes fluid through quantum computing},\ }\href@noop {} {\bibfield  {journal} {\bibinfo  {journal} {npj Quantum Information}\ }\textbf {\bibinfo {volume} {6}},\ \bibinfo {pages} {61} (\bibinfo {year} {2020})}\BibitemShut {NoStop}%
\bibitem [{\citenamefont {Jin}\ \emph {et~al.}(2023{\natexlab{a}})\citenamefont {Jin}, \citenamefont {Liu},\ and\ \citenamefont {Yu}}]{jin2023time}%
  \BibitemOpen
  \bibfield  {author} {\bibinfo {author} {\bibfnamefont {S.}~\bibnamefont {Jin}}, \bibinfo {author} {\bibfnamefont {N.}~\bibnamefont {Liu}},\ and\ \bibinfo {author} {\bibfnamefont {Y.}~\bibnamefont {Yu}},\ }\bibfield  {title} {\bibinfo {title} {Time complexity analysis of quantum algorithms via linear representations for nonlinear ordinary and partial differential equations},\ }\href@noop {} {\bibfield  {journal} {\bibinfo  {journal} {Journal of Computational Physics}\ }\textbf {\bibinfo {volume} {487}},\ \bibinfo {pages} {112149} (\bibinfo {year} {2023}{\natexlab{a}})}\BibitemShut {NoStop}%
\bibitem [{\citenamefont {Linden}\ \emph {et~al.}(2022)\citenamefont {Linden}, \citenamefont {Montanaro},\ and\ \citenamefont {Shao}}]{linden2022quantum}%
  \BibitemOpen
  \bibfield  {author} {\bibinfo {author} {\bibfnamefont {N.}~\bibnamefont {Linden}}, \bibinfo {author} {\bibfnamefont {A.}~\bibnamefont {Montanaro}},\ and\ \bibinfo {author} {\bibfnamefont {C.}~\bibnamefont {Shao}},\ }\bibfield  {title} {\bibinfo {title} {Quantum vs. classical algorithms for solving the heat equation},\ }\href@noop {} {\bibfield  {journal} {\bibinfo  {journal} {Communications in Mathematical Physics}\ }\textbf {\bibinfo {volume} {395}},\ \bibinfo {pages} {601} (\bibinfo {year} {2022})}\BibitemShut {NoStop}%
\bibitem [{\citenamefont {Jin}\ \emph {et~al.}(2022)\citenamefont {Jin}, \citenamefont {Liu},\ and\ \citenamefont {Yu}}]{jin2022quantum}%
  \BibitemOpen
  \bibfield  {author} {\bibinfo {author} {\bibfnamefont {S.}~\bibnamefont {Jin}}, \bibinfo {author} {\bibfnamefont {N.}~\bibnamefont {Liu}},\ and\ \bibinfo {author} {\bibfnamefont {Y.}~\bibnamefont {Yu}},\ }\bibfield  {title} {\bibinfo {title} {Quantum simulation of partial differential equations via schrodingerisation: technical details},\ }\href@noop {} {\bibfield  {journal} {\bibinfo  {journal} {arXiv preprint arXiv:2212.14703}\ } (\bibinfo {year} {2022})}\BibitemShut {NoStop}%
\bibitem [{\citenamefont {Sato}\ \emph {et~al.}(2024)\citenamefont {Sato}, \citenamefont {Kondo}, \citenamefont {Hamamura}, \citenamefont {Onodera},\ and\ \citenamefont {Yamamoto}}]{sato2024hamiltonian}%
  \BibitemOpen
  \bibfield  {author} {\bibinfo {author} {\bibfnamefont {Y.}~\bibnamefont {Sato}}, \bibinfo {author} {\bibfnamefont {R.}~\bibnamefont {Kondo}}, \bibinfo {author} {\bibfnamefont {I.}~\bibnamefont {Hamamura}}, \bibinfo {author} {\bibfnamefont {T.}~\bibnamefont {Onodera}},\ and\ \bibinfo {author} {\bibfnamefont {N.}~\bibnamefont {Yamamoto}},\ }\bibfield  {title} {\bibinfo {title} {Hamiltonian simulation for time-evolving partial differential equation by scalable quantum circuits},\ }\href@noop {} {\bibfield  {journal} {\bibinfo  {journal} {arXiv preprint arXiv:2402.18398}\ } (\bibinfo {year} {2024})}\BibitemShut {NoStop}%
\bibitem [{\citenamefont {Stamatopoulos}\ \emph {et~al.}(2020)\citenamefont {Stamatopoulos}, \citenamefont {Egger}, \citenamefont {Sun}, \citenamefont {Zoufal}, \citenamefont {Iten}, \citenamefont {Shen},\ and\ \citenamefont {Woerner}}]{stamatopoulos2020option}%
  \BibitemOpen
  \bibfield  {author} {\bibinfo {author} {\bibfnamefont {N.}~\bibnamefont {Stamatopoulos}}, \bibinfo {author} {\bibfnamefont {D.~J.}\ \bibnamefont {Egger}}, \bibinfo {author} {\bibfnamefont {Y.}~\bibnamefont {Sun}}, \bibinfo {author} {\bibfnamefont {C.}~\bibnamefont {Zoufal}}, \bibinfo {author} {\bibfnamefont {R.}~\bibnamefont {Iten}}, \bibinfo {author} {\bibfnamefont {N.}~\bibnamefont {Shen}},\ and\ \bibinfo {author} {\bibfnamefont {S.}~\bibnamefont {Woerner}},\ }\bibfield  {title} {\bibinfo {title} {Option pricing using quantum computers},\ }\href@noop {} {\bibfield  {journal} {\bibinfo  {journal} {Quantum}\ }\textbf {\bibinfo {volume} {4}},\ \bibinfo {pages} {291} (\bibinfo {year} {2020})}\BibitemShut {NoStop}%
\bibitem [{\citenamefont {Gonzalez-Conde}\ \emph {et~al.}(2023)\citenamefont {Gonzalez-Conde}, \citenamefont {Rodr{\'\i}guez-Rozas}, \citenamefont {Solano},\ and\ \citenamefont {Sanz}}]{gonzalez2023efficient}%
  \BibitemOpen
  \bibfield  {author} {\bibinfo {author} {\bibfnamefont {J.}~\bibnamefont {Gonzalez-Conde}}, \bibinfo {author} {\bibfnamefont {{\'A}.}~\bibnamefont {Rodr{\'\i}guez-Rozas}}, \bibinfo {author} {\bibfnamefont {E.}~\bibnamefont {Solano}},\ and\ \bibinfo {author} {\bibfnamefont {M.}~\bibnamefont {Sanz}},\ }\bibfield  {title} {\bibinfo {title} {Efficient hamiltonian simulation for solving option price dynamics},\ }\href@noop {} {\bibfield  {journal} {\bibinfo  {journal} {Physical Review Research}\ }\textbf {\bibinfo {volume} {5}},\ \bibinfo {pages} {043220} (\bibinfo {year} {2023})}\BibitemShut {NoStop}%
\bibitem [{\citenamefont {Gily{\'e}n}\ \emph {et~al.}(2019)\citenamefont {Gily{\'e}n}, \citenamefont {Su}, \citenamefont {Low},\ and\ \citenamefont {Wiebe}}]{gilyen2019quantum}%
  \BibitemOpen
  \bibfield  {author} {\bibinfo {author} {\bibfnamefont {A.}~\bibnamefont {Gily{\'e}n}}, \bibinfo {author} {\bibfnamefont {Y.}~\bibnamefont {Su}}, \bibinfo {author} {\bibfnamefont {G.~H.}\ \bibnamefont {Low}},\ and\ \bibinfo {author} {\bibfnamefont {N.}~\bibnamefont {Wiebe}},\ }\bibfield  {title} {\bibinfo {title} {Quantum singular value transformation and beyond: exponential improvements for quantum matrix arithmetics},\ }in\ \href@noop {} {\emph {\bibinfo {booktitle} {Proceedings of the 51st Annual ACM SIGACT Symposium on Theory of Computing}}}\ (\bibinfo {year} {2019})\ pp.\ \bibinfo {pages} {193--204}\BibitemShut {NoStop}%
\bibitem [{\citenamefont {Jin}\ and\ \citenamefont {Liu}(2023{\natexlab{a}})}]{jin2023quantum}%
  \BibitemOpen
  \bibfield  {author} {\bibinfo {author} {\bibfnamefont {S.}~\bibnamefont {Jin}}\ and\ \bibinfo {author} {\bibfnamefont {N.}~\bibnamefont {Liu}},\ }\bibfield  {title} {\bibinfo {title} {Quantum simulation of discrete linear dynamical systems and simple iterative methods in linear algebra via schrodingerisation},\ }\href@noop {} {\bibfield  {journal} {\bibinfo  {journal} {arXiv preprint arXiv:2304.02865}\ } (\bibinfo {year} {2023}{\natexlab{a}})}\BibitemShut {NoStop}%
\bibitem [{\citenamefont {Jin}\ and\ \citenamefont {Liu}(2023{\natexlab{b}})}]{jin2023analog}%
  \BibitemOpen
  \bibfield  {author} {\bibinfo {author} {\bibfnamefont {S.}~\bibnamefont {Jin}}\ and\ \bibinfo {author} {\bibfnamefont {N.}~\bibnamefont {Liu}},\ }\bibfield  {title} {\bibinfo {title} {Analog quantum simulation of partial differential equations},\ }\href@noop {} {\bibfield  {journal} {\bibinfo  {journal} {arXiv preprint arXiv:2308.00646}\ } (\bibinfo {year} {2023}{\natexlab{b}})}\BibitemShut {NoStop}%
\bibitem [{\citenamefont {Cao}\ \emph {et~al.}(2023)\citenamefont {Cao}, \citenamefont {Jin},\ and\ \citenamefont {Liu}}]{cao2023quantum}%
  \BibitemOpen
  \bibfield  {author} {\bibinfo {author} {\bibfnamefont {Y.}~\bibnamefont {Cao}}, \bibinfo {author} {\bibfnamefont {S.}~\bibnamefont {Jin}},\ and\ \bibinfo {author} {\bibfnamefont {N.}~\bibnamefont {Liu}},\ }\bibfield  {title} {\bibinfo {title} {Quantum simulation for time-dependent hamiltonians--with applications to non-autonomous ordinary and partial differential equations},\ }\href@noop {} {\bibfield  {journal} {\bibinfo  {journal} {arXiv preprint arXiv:2312.02817}\ } (\bibinfo {year} {2023})}\BibitemShut {NoStop}%
\bibitem [{\citenamefont {Low}\ and\ \citenamefont {Chuang}(2017)}]{low2017optimal}%
  \BibitemOpen
  \bibfield  {author} {\bibinfo {author} {\bibfnamefont {G.~H.}\ \bibnamefont {Low}}\ and\ \bibinfo {author} {\bibfnamefont {I.~L.}\ \bibnamefont {Chuang}},\ }\bibfield  {title} {\bibinfo {title} {Optimal hamiltonian simulation by quantum signal processing},\ }\href@noop {} {\bibfield  {journal} {\bibinfo  {journal} {Physical review letters}\ }\textbf {\bibinfo {volume} {118}},\ \bibinfo {pages} {010501} (\bibinfo {year} {2017})}\BibitemShut {NoStop}%
\bibitem [{\citenamefont {Berry}\ \emph {et~al.}(2014)\citenamefont {Berry}, \citenamefont {Childs}, \citenamefont {Cleve}, \citenamefont {Kothari},\ and\ \citenamefont {Somma}}]{berry2014exponential}%
  \BibitemOpen
  \bibfield  {author} {\bibinfo {author} {\bibfnamefont {D.~W.}\ \bibnamefont {Berry}}, \bibinfo {author} {\bibfnamefont {A.~M.}\ \bibnamefont {Childs}}, \bibinfo {author} {\bibfnamefont {R.}~\bibnamefont {Cleve}}, \bibinfo {author} {\bibfnamefont {R.}~\bibnamefont {Kothari}},\ and\ \bibinfo {author} {\bibfnamefont {R.~D.}\ \bibnamefont {Somma}},\ }\bibfield  {title} {\bibinfo {title} {Exponential improvement in precision for simulating sparse hamiltonians},\ }in\ \href@noop {} {\emph {\bibinfo {booktitle} {Proceedings of the forty-sixth annual ACM symposium on Theory of computing}}}\ (\bibinfo {year} {2014})\ pp.\ \bibinfo {pages} {283--292}\BibitemShut {NoStop}%
\bibitem [{\citenamefont {Hu}\ \emph {et~al.}(2024)\citenamefont {Hu}, \citenamefont {Jin}, \citenamefont {Liu},\ and\ \citenamefont {Zhang}}]{hu2024quantum}%
  \BibitemOpen
  \bibfield  {author} {\bibinfo {author} {\bibfnamefont {J.}~\bibnamefont {Hu}}, \bibinfo {author} {\bibfnamefont {S.}~\bibnamefont {Jin}}, \bibinfo {author} {\bibfnamefont {N.}~\bibnamefont {Liu}},\ and\ \bibinfo {author} {\bibfnamefont {L.}~\bibnamefont {Zhang}},\ }\bibfield  {title} {\bibinfo {title} {Quantum circuits for partial differential equations via schr$\backslash$" odingerisation},\ }\href@noop {} {\bibfield  {journal} {\bibinfo  {journal} {arXiv preprint arXiv:2403.10032}\ } (\bibinfo {year} {2024})}\BibitemShut {NoStop}%
\bibitem [{\citenamefont {Nielsen}\ and\ \citenamefont {Chuang}(2011)}]{nielsen2002quantum}%
  \BibitemOpen
  \bibfield  {author} {\bibinfo {author} {\bibfnamefont {M.~A.}\ \bibnamefont {Nielsen}}\ and\ \bibinfo {author} {\bibfnamefont {I.~L.}\ \bibnamefont {Chuang}},\ }\href@noop {} {\emph {\bibinfo {title} {Quantum Computation and Quantum Information: 10th Anniversary Edition}}}\ (\bibinfo  {publisher} {Cambridge University Press},\ \bibinfo {year} {2011})\BibitemShut {NoStop}%
\bibitem [{\citenamefont {Suzuki}(1993)}]{Suzuki}%
  \BibitemOpen
  \bibfield  {author} {\bibinfo {author} {\bibfnamefont {M.}~\bibnamefont {Suzuki}},\ }\bibfield  {title} {\bibinfo {title} {Improved trotter-like formula},\ }\href@noop {} {\bibfield  {journal} {\bibinfo  {journal} {Phys. Lett. A}\ }\textbf {\bibinfo {volume} {180}},\ \bibinfo {pages} {232} (\bibinfo {year} {1993})}\BibitemShut {NoStop}%
\bibitem [{\citenamefont {Trotter}(1959)}]{Trotter}%
  \BibitemOpen
  \bibfield  {author} {\bibinfo {author} {\bibfnamefont {H.~F.}\ \bibnamefont {Trotter}},\ }\bibfield  {title} {\bibinfo {title} {On the product of semi-groups of operators},\ }\href@noop {} {\bibfield  {journal} {\bibinfo  {journal} {Proceedings of the American Mathematical Society}\ }\textbf {\bibinfo {volume} {10}},\ \bibinfo {pages} {545} (\bibinfo {year} {1959})}\BibitemShut {NoStop}%
\bibitem [{\citenamefont {Dawson}\ and\ \citenamefont {Nielsen}(2005)}]{dawson2005solovay}%
  \BibitemOpen
  \bibfield  {author} {\bibinfo {author} {\bibfnamefont {C.~M.}\ \bibnamefont {Dawson}}\ and\ \bibinfo {author} {\bibfnamefont {M.~A.}\ \bibnamefont {Nielsen}},\ }\bibfield  {title} {\bibinfo {title} {The solovay-kitaev algorithm},\ }\href@noop {} {\bibfield  {journal} {\bibinfo  {journal} {arXiv preprint quant-ph/0505030}\ } (\bibinfo {year} {2005})}\BibitemShut {NoStop}%
\bibitem [{\citenamefont {Zienkiewicz}\ \emph {et~al.}(2005)\citenamefont {Zienkiewicz}, \citenamefont {Taylor},\ and\ \citenamefont {Zhu}}]{zienkiewicz2005finite}%
  \BibitemOpen
  \bibfield  {author} {\bibinfo {author} {\bibfnamefont {O.~C.}\ \bibnamefont {Zienkiewicz}}, \bibinfo {author} {\bibfnamefont {R.~L.}\ \bibnamefont {Taylor}},\ and\ \bibinfo {author} {\bibfnamefont {J.~Z.}\ \bibnamefont {Zhu}},\ }\href@noop {} {\emph {\bibinfo {title} {The finite element method: its basis and fundamentals}}}\ (\bibinfo  {publisher} {Elsevier},\ \bibinfo {year} {2005})\BibitemShut {NoStop}%
\bibitem [{\citenamefont {{\"O}zi{\c{s}}ik}\ \emph {et~al.}(2017)\citenamefont {{\"O}zi{\c{s}}ik}, \citenamefont {Orlande}, \citenamefont {Cola{\c{c}}o},\ and\ \citenamefont {Cotta}}]{ozicsik2017finite}%
  \BibitemOpen
  \bibfield  {author} {\bibinfo {author} {\bibfnamefont {M.~N.}\ \bibnamefont {{\"O}zi{\c{s}}ik}}, \bibinfo {author} {\bibfnamefont {H.~R.}\ \bibnamefont {Orlande}}, \bibinfo {author} {\bibfnamefont {M.~J.}\ \bibnamefont {Cola{\c{c}}o}},\ and\ \bibinfo {author} {\bibfnamefont {R.~M.}\ \bibnamefont {Cotta}},\ }\href@noop {} {\emph {\bibinfo {title} {Finite difference methods in heat transfer}}}\ (\bibinfo  {publisher} {CRC press},\ \bibinfo {year} {2017})\BibitemShut {NoStop}%
\bibitem [{\citenamefont {Jin}\ \emph {et~al.}()\citenamefont {Jin}, \citenamefont {Liu},\ and\ \citenamefont {Yu}}]{schr1}%
  \BibitemOpen
  \bibfield  {author} {\bibinfo {author} {\bibfnamefont {S.}~\bibnamefont {Jin}}, \bibinfo {author} {\bibfnamefont {N.}~\bibnamefont {Liu}},\ and\ \bibinfo {author} {\bibfnamefont {Y.}~\bibnamefont {Yu}},\ }\bibfield  {title} {\bibinfo {title} {Quantum simulation of partial differential equations via schrodingerisation (2022)},\ }\href@noop {} {\bibinfo  {journal} {arXiv preprint arXiv:2212.13969}\ }\BibitemShut {NoStop}%
\bibitem [{\citenamefont {Jin}\ \emph {et~al.}(2023{\natexlab{b}})\citenamefont {Jin}, \citenamefont {Liu},\ and\ \citenamefont {Yu}}]{schr2}%
  \BibitemOpen
\bibfield  {journal} {  }\bibfield  {author} {\bibinfo {author} {\bibfnamefont {S.}~\bibnamefont {Jin}}, \bibinfo {author} {\bibfnamefont {N.}~\bibnamefont {Liu}},\ and\ \bibinfo {author} {\bibfnamefont {Y.}~\bibnamefont {Yu}},\ }\bibfield  {title} {\bibinfo {title} {Quantum simulation of partial differential equations: Applications and detailed analysis},\ }\href@noop {} {\bibfield  {journal} {\bibinfo  {journal} {Physical Review A}\ }\textbf {\bibinfo {volume} {108}},\ \bibinfo {pages} {032603} (\bibinfo {year} {2023}{\natexlab{b}})}\BibitemShut {NoStop}%
\bibitem [{\citenamefont {Jin}\ and\ \citenamefont {Liu}(2022)}]{jin2022quantumnonlinear}%
  \BibitemOpen
  \bibfield  {author} {\bibinfo {author} {\bibfnamefont {S.}~\bibnamefont {Jin}}\ and\ \bibinfo {author} {\bibfnamefont {N.}~\bibnamefont {Liu}},\ }\bibfield  {title} {\bibinfo {title} {Quantum algorithms for computing observables of nonlinear partial differential equations},\ }\href@noop {} {\bibfield  {journal} {\bibinfo  {journal} {arXiv preprint arXiv:2202.07834}\ } (\bibinfo {year} {2022})}\BibitemShut {NoStop}%
\bibitem [{\citenamefont {Guseynov}\ \emph {et~al.}(2023)\citenamefont {Guseynov}, \citenamefont {Zhukov}, \citenamefont {Pogosov},\ and\ \citenamefont {Lebedev}}]{guseynov2023depth}%
  \BibitemOpen
  \bibfield  {author} {\bibinfo {author} {\bibfnamefont {N.}~\bibnamefont {Guseynov}}, \bibinfo {author} {\bibfnamefont {A.}~\bibnamefont {Zhukov}}, \bibinfo {author} {\bibfnamefont {W.}~\bibnamefont {Pogosov}},\ and\ \bibinfo {author} {\bibfnamefont {A.}~\bibnamefont {Lebedev}},\ }\bibfield  {title} {\bibinfo {title} {Depth analysis of variational quantum algorithms for the heat equation},\ }\href@noop {} {\bibfield  {journal} {\bibinfo  {journal} {Phys. Rev. A}\ }\textbf {\bibinfo {volume} {107}},\ \bibinfo {pages} {052422} (\bibinfo {year} {2023})}\BibitemShut {NoStop}%
\bibitem [{\citenamefont {Walsh}(1923)}]{walsh1923closed}%
  \BibitemOpen
  \bibfield  {author} {\bibinfo {author} {\bibfnamefont {J.~L.}\ \bibnamefont {Walsh}},\ }\bibfield  {title} {\bibinfo {title} {A closed set of normal orthogonal functions},\ }\href@noop {} {\bibfield  {journal} {\bibinfo  {journal} {American Journal of Mathematics}\ }\textbf {\bibinfo {volume} {45}},\ \bibinfo {pages} {5} (\bibinfo {year} {1923})}\BibitemShut {NoStop}%
\bibitem [{\citenamefont {Rattew}\ and\ \citenamefont {Rebentrost}(2023)}]{rattew2023non}%
  \BibitemOpen
  \bibfield  {author} {\bibinfo {author} {\bibfnamefont {A.~G.}\ \bibnamefont {Rattew}}\ and\ \bibinfo {author} {\bibfnamefont {P.}~\bibnamefont {Rebentrost}},\ }\bibfield  {title} {\bibinfo {title} {Non-linear transformations of quantum amplitudes: Exponential improvement, generalization, and applications},\ }\href@noop {} {\bibfield  {journal} {\bibinfo  {journal} {arXiv preprint arXiv:2309.09839}\ } (\bibinfo {year} {2023})}\BibitemShut {NoStop}%
\bibitem [{\citenamefont {Guo}\ \emph {et~al.}(2021)\citenamefont {Guo}, \citenamefont {Mitarai},\ and\ \citenamefont {Fujii}}]{guo2021nonlinear}%
  \BibitemOpen
  \bibfield  {author} {\bibinfo {author} {\bibfnamefont {N.}~\bibnamefont {Guo}}, \bibinfo {author} {\bibfnamefont {K.}~\bibnamefont {Mitarai}},\ and\ \bibinfo {author} {\bibfnamefont {K.}~\bibnamefont {Fujii}},\ }\bibfield  {title} {\bibinfo {title} {Nonlinear transformation of complex amplitudes via quantum singular value transformation},\ }\href@noop {} {\bibfield  {journal} {\bibinfo  {journal} {arXiv preprint arXiv:2107.10764}\ } (\bibinfo {year} {2021})}\BibitemShut {NoStop}%
\bibitem [{\citenamefont {Dalzell}\ \emph {et~al.}(2023)\citenamefont {Dalzell}, \citenamefont {McArdle}, \citenamefont {Berta}, \citenamefont {Bienias}, \citenamefont {Chen}, \citenamefont {Gily{\'e}n}, \citenamefont {Hann}, \citenamefont {Kastoryano}, \citenamefont {Khabiboulline}, \citenamefont {Kubica} \emph {et~al.}}]{dalzell2023quantum}%
  \BibitemOpen
  \bibfield  {author} {\bibinfo {author} {\bibfnamefont {A.~M.}\ \bibnamefont {Dalzell}}, \bibinfo {author} {\bibfnamefont {S.}~\bibnamefont {McArdle}}, \bibinfo {author} {\bibfnamefont {M.}~\bibnamefont {Berta}}, \bibinfo {author} {\bibfnamefont {P.}~\bibnamefont {Bienias}}, \bibinfo {author} {\bibfnamefont {C.-F.}\ \bibnamefont {Chen}}, \bibinfo {author} {\bibfnamefont {A.}~\bibnamefont {Gily{\'e}n}}, \bibinfo {author} {\bibfnamefont {C.~T.}\ \bibnamefont {Hann}}, \bibinfo {author} {\bibfnamefont {M.~J.}\ \bibnamefont {Kastoryano}}, \bibinfo {author} {\bibfnamefont {E.~T.}\ \bibnamefont {Khabiboulline}}, \bibinfo {author} {\bibfnamefont {A.}~\bibnamefont {Kubica}}, \emph {et~al.},\ }\bibfield  {title} {\bibinfo {title} {Quantum algorithms: A survey of applications and end-to-end complexities},\ }\href@noop {} {\bibfield  {journal} {\bibinfo  {journal} {arXiv preprint arXiv:2310.03011}\ } (\bibinfo {year} {2023})}\BibitemShut {NoStop}%
\bibitem [{\citenamefont {Plesch}\ and\ \citenamefont {Brukner}(2011)}]{plesch2011quantum}%
  \BibitemOpen
  \bibfield  {author} {\bibinfo {author} {\bibfnamefont {M.}~\bibnamefont {Plesch}}\ and\ \bibinfo {author} {\bibfnamefont {{\v{C}}.}~\bibnamefont {Brukner}},\ }\bibfield  {title} {\bibinfo {title} {Quantum-state preparation with universal gate decompositions},\ }\href@noop {} {\bibfield  {journal} {\bibinfo  {journal} {Physical Review A}\ }\textbf {\bibinfo {volume} {83}},\ \bibinfo {pages} {032302} (\bibinfo {year} {2011})}\BibitemShut {NoStop}%
\bibitem [{\citenamefont {Bergholm}\ \emph {et~al.}(2005)\citenamefont {Bergholm}, \citenamefont {Vartiainen}, \citenamefont {M{\"o}tt{\"o}nen},\ and\ \citenamefont {Salomaa}}]{bergholm2005quantum}%
  \BibitemOpen
  \bibfield  {author} {\bibinfo {author} {\bibfnamefont {V.}~\bibnamefont {Bergholm}}, \bibinfo {author} {\bibfnamefont {J.~J.}\ \bibnamefont {Vartiainen}}, \bibinfo {author} {\bibfnamefont {M.}~\bibnamefont {M{\"o}tt{\"o}nen}},\ and\ \bibinfo {author} {\bibfnamefont {M.~M.}\ \bibnamefont {Salomaa}},\ }\bibfield  {title} {\bibinfo {title} {Quantum circuits with uniformly controlled one-qubit gates},\ }\href@noop {} {\bibfield  {journal} {\bibinfo  {journal} {Physical Review A}\ }\textbf {\bibinfo {volume} {71}},\ \bibinfo {pages} {052330} (\bibinfo {year} {2005})}\BibitemShut {NoStop}%
\bibitem [{\citenamefont {Shende}\ \emph {et~al.}(2004)\citenamefont {Shende}, \citenamefont {Markov},\ and\ \citenamefont {Bullock}}]{shende2004minimal}%
  \BibitemOpen
  \bibfield  {author} {\bibinfo {author} {\bibfnamefont {V.~V.}\ \bibnamefont {Shende}}, \bibinfo {author} {\bibfnamefont {I.~L.}\ \bibnamefont {Markov}},\ and\ \bibinfo {author} {\bibfnamefont {S.~S.}\ \bibnamefont {Bullock}},\ }\bibfield  {title} {\bibinfo {title} {Minimal universal two-qubit controlled-not-based circuits},\ }\href@noop {} {\bibfield  {journal} {\bibinfo  {journal} {Physical Review A}\ }\textbf {\bibinfo {volume} {69}},\ \bibinfo {pages} {062321} (\bibinfo {year} {2004})}\BibitemShut {NoStop}%
\bibitem [{\citenamefont {Coolidge}(1949)}]{coolidge1949story}%
  \BibitemOpen
  \bibfield  {author} {\bibinfo {author} {\bibfnamefont {J.~L.}\ \bibnamefont {Coolidge}},\ }\bibfield  {title} {\bibinfo {title} {The story of the binomial theorem},\ }\href@noop {} {\bibfield  {journal} {\bibinfo  {journal} {The American Mathematical Monthly}\ }\textbf {\bibinfo {volume} {56}},\ \bibinfo {pages} {147} (\bibinfo {year} {1949})}\BibitemShut {NoStop}%
\bibitem [{\citenamefont {Draper}(2000)}]{draper2000addition}%
  \BibitemOpen
  \bibfield  {author} {\bibinfo {author} {\bibfnamefont {T.~G.}\ \bibnamefont {Draper}},\ }\bibfield  {title} {\bibinfo {title} {Addition on a quantum computer},\ }\href@noop {} {\bibfield  {journal} {\bibinfo  {journal} {arXiv preprint quant-ph/0008033}\ } (\bibinfo {year} {2000})}\BibitemShut {NoStop}%
\bibitem [{\citenamefont {Draper}\ \emph {et~al.}(2004)\citenamefont {Draper}, \citenamefont {Kutin}, \citenamefont {Rains},\ and\ \citenamefont {Svore}}]{draper2004logarithmic}%
  \BibitemOpen
  \bibfield  {author} {\bibinfo {author} {\bibfnamefont {T.~G.}\ \bibnamefont {Draper}}, \bibinfo {author} {\bibfnamefont {S.~A.}\ \bibnamefont {Kutin}}, \bibinfo {author} {\bibfnamefont {E.~M.}\ \bibnamefont {Rains}},\ and\ \bibinfo {author} {\bibfnamefont {K.~M.}\ \bibnamefont {Svore}},\ }\bibfield  {title} {\bibinfo {title} {A logarithmic-depth quantum carry-lookahead adder},\ }\href@noop {} {\bibfield  {journal} {\bibinfo  {journal} {arXiv preprint quant-ph/0406142}\ } (\bibinfo {year} {2004})}\BibitemShut {NoStop}%
\bibitem [{\citenamefont {Camps}\ and\ \citenamefont {Van~Beeumen}(2022)}]{camps2022fable}%
  \BibitemOpen
  \bibfield  {author} {\bibinfo {author} {\bibfnamefont {D.}~\bibnamefont {Camps}}\ and\ \bibinfo {author} {\bibfnamefont {R.}~\bibnamefont {Van~Beeumen}},\ }\bibfield  {title} {\bibinfo {title} {Fable: Fast approximate quantum circuits for block-encodings},\ }in\ \href@noop {} {\emph {\bibinfo {booktitle} {2022 IEEE International Conference on Quantum Computing and Engineering (QCE)}}}\ (\bibinfo {organization} {IEEE},\ \bibinfo {year} {2022})\ pp.\ \bibinfo {pages} {104--113}\BibitemShut {NoStop}%
\bibitem [{\citenamefont {M{\"o}tt{\"o}nen}\ \emph {et~al.}(2004)\citenamefont {M{\"o}tt{\"o}nen}, \citenamefont {Vartiainen}, \citenamefont {Bergholm},\ and\ \citenamefont {Salomaa}}]{mottonen2004quantum}%
  \BibitemOpen
  \bibfield  {author} {\bibinfo {author} {\bibfnamefont {M.}~\bibnamefont {M{\"o}tt{\"o}nen}}, \bibinfo {author} {\bibfnamefont {J.~J.}\ \bibnamefont {Vartiainen}}, \bibinfo {author} {\bibfnamefont {V.}~\bibnamefont {Bergholm}},\ and\ \bibinfo {author} {\bibfnamefont {M.~M.}\ \bibnamefont {Salomaa}},\ }\bibfield  {title} {\bibinfo {title} {Quantum circuits for general multiqubit gates},\ }\href@noop {} {\bibfield  {journal} {\bibinfo  {journal} {Physical review letters}\ }\textbf {\bibinfo {volume} {93}},\ \bibinfo {pages} {130502} (\bibinfo {year} {2004})}\BibitemShut {NoStop}%
\end{thebibliography}%

\appendix
    
    \section{Explicit quantum circuit construction for $\hat{O}^{BS}_A$}\label{Appendix Banded-sparse-access}
    
    Let a $2^l$ sparse matrix $A$ be as in the definition \ref{Banded-sparse-access}; so, the non-zero elements follows
    
    \begin{equation}
        \hat{O}^{BS}_A\ket{0}^{n-l}\ket{s}^l\ket{i}^n=\ket{r_{s0} + i}^n\ket{i}^n,
    \end{equation}
    where sum $r_{s0} + i$ means an addition modulo $2^n$. So we decompose $\hat{O}^{BS}_A$ into a product of two unitary operators $\hat{O}^{BS}_A=U^{SUM}(U_A^{(l)}\otimes \hat{I}^n)$, where $U^{SUM}$ performs addition modulo $2^n$
    \begin{equation}
        U^{SUM}\ket{i}^n\ket{j}^n=\ket{i+j\mod 2^n}^n\ket{j}^n,
    \end{equation}
    and operator $U^{(l)}_A$ transforming a banded sparse index into a column index
    \begin{equation}
       U^{(l)}_A\ket{0}^{n-l}\ket{s}^l=\ket{r_{s0} }^n.
    \end{equation}
    We underline that the information about matrix $A$ contained only in the $U_A^{(l)}$; so we omitted the corresponding index for the $U^{SUM}$. Thus, we combine those two operators to build the Banded-sparse-access
    \begin{equation}
        \hat{O}^{BS}_A\ket{0}^{n-l}\ket{s}^l\ket{i}^n=U^{SUM}(U^{(l)}_A\ket{0}^{n-l}\ket{s}^l \otimes \hat{I}^n\ket{i}^n)=U^{SUM}\ket{r_{s0}}^n\ket{i}^n=\ket{r_{s0} + i}^n\ket{i}^n.
    \end{equation}
    
    Now we demonstrate how $U_A^{(l)}$ can be constructed using one and two qubit basic quantum computer operations. First, we further decompose the operator $U_A^{(l)}$ into a product
    \begin{equation}
        \label{eq:creation first row oracle}
        U_A^{(l)}=\prod_{s=0}^{2^l-1}U_s^{r_{s0}};\qquad U^{r_{s0}}_s:=\ket{r_{s0}}^n(\bra{0}^{n-l}\otimes\bra{s}^l)+\sum_{j\neq s} \ket{j}^n\bra{j}^n.
    \end{equation}
    The quantum circuit depicted in the Fig.~\ref{fig:Sparse_first_row} shows the explicit construction of $U^{r_{s0}}_s$ where each of the multi-controlled-$X$ we implement as in Fig.~\ref{fig:Multi_control_qc}. Thus, the overall resources for $U_A^{(l)}$ implementation no greater than: (i) $n-1$ pure ancillas, (ii) $2^l(32n-48)$ one-qubit operations, (iii) $2^l(25n-36)$ C-NOTs.
    
    \begin{figure}[h]
    \includegraphics[width=0.5\textwidth]{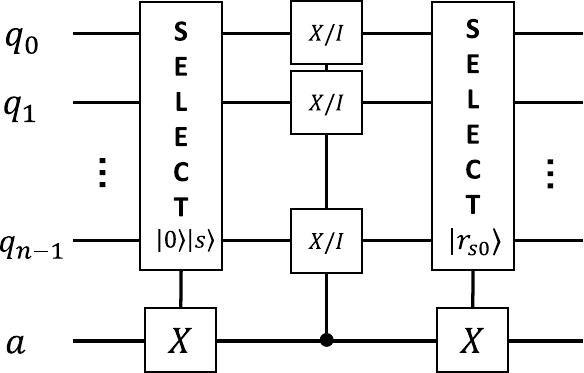}
    \caption{Quantum circuit implementing $U^{r_{s0}}_s$. The first layer set the initialized with zero-state ancilla to $\ket{1}$ for state $\ket{0}^{n-l}\otimes\ket{s}^l$ in the frist register. The second layer operation make state $\ket{r_{s0}}^n$ from $\ket{0}^{n-l}\otimes\ket{s}^l$ by choosing Pauli $X$ or $I$ according to bit-difference between $s$ and $r_{s0}$. The third layer unentangle ancilla making it pure.}
    \label{fig:Sparse_first_row}
    \end{figure}
    
    Now, let us discuss how to implement $U^{SUM}$ which adds the row number $i$ to the $r_{s0}$. In Fig.~\ref{fig:addition_circuit} we explicitly demonstrate how such a circuit can be built. We slightly modify the adder from \cite{draper2000addition} to meet our needs. From the quantum scheme it can be concluded that resources for the addition is no more than: (i) $n-1$ pure ancillas, (ii) $32n-48$ one-qubit operations, (iii) $26n-37$ C-NOTs. Although the final circuit depth is linear, which isn't optimal, there exists a logarithmic-depth $(\mathcal{O}(\log n))$algorithm according to \cite{draper2004logarithmic}, albeit requiring more gates (up to $\approx124n$ C-NOTs). The choice of algorithm depends on hardware specifics. However, in this paper, we opt for the circuit in Fig.\ref{fig:addition_circuit} due to its minimal gate and ancilla requirements.
    
    \begin{figure}[h]
    \includegraphics[width=0.8\textwidth]{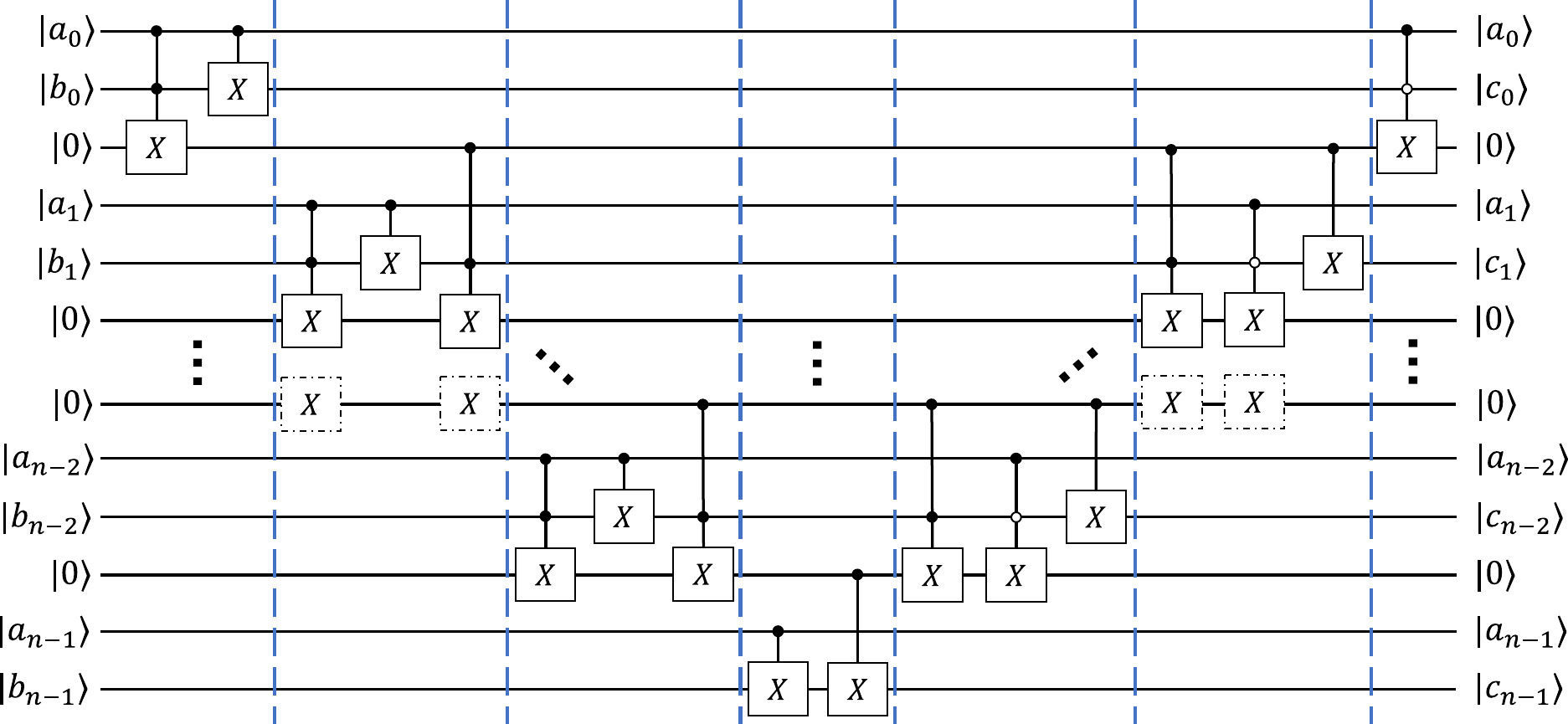}
    \caption{Quantum circuit implementing $U^{SUM}$ which adds two binary strings: $a+b=c\mod 2^n$. }
    \label{fig:addition_circuit}
    \end{figure}
    
    \section{Quantum circuit design for $\hat{O}_{p^m}^S$}\label{appendix_oracle_for_momentum}
    
    A first naive quantum circuit can be implemented in a straightforward way, see Fig.~\ref{fig:naive_momentum_oracle}
    \begin{equation}
        \hat{O}_{p^m}^S=e^{-i\pi m/2}\prod_{s=0}^{2^l-1}C_{Ry(\theta_s)}^s;\qquad \theta_s=2\arccos(e^{i\pi m/2}\frac{(\hat{p}^m)^{(s)}}{\sqrt{\mathcal{N}_{p^m}}}),
    \end{equation}
    where we hereafter omit the global phase $e^{-i\pi m/2}$; however, we take it into account in cases when it does matter.
    \begin{figure}[h]
    \includegraphics[width=0.5\textwidth]{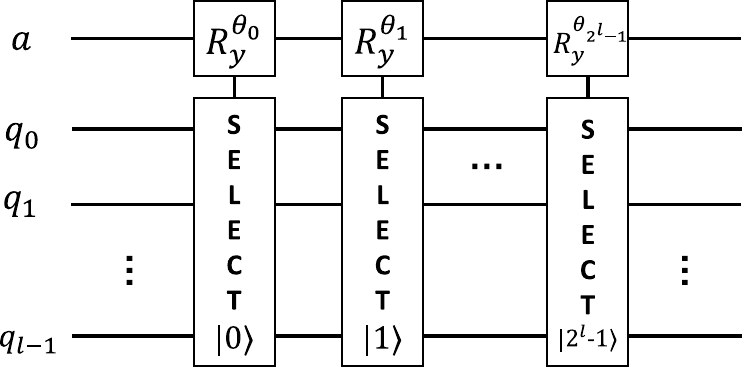}
    \caption{A naive design of the quantum circuit that implements $\hat{O}_{p^m}^S$. Each layer corresponds to the different non-zero elements of the first row of the matrix $\hat{p}^m$. The qubit $a$ correspond to the first register in the definition of $\hat{O}_{p^m}^S$.}
    \label{fig:naive_momentum_oracle}
    \end{figure}
    A more efficient approach can be constructed using Fable implementation of block-encoding \cite{camps2022fable,mottonen2004quantum}. A $2$-qubit toy example of the method which explains the general concept is depicted in Fig.~\ref{fig:toy_example_O_p_s}. Basically, the authors use Gray code ordering to reduce complexity of the $\hat{O}_{p^m}^S$ implementation. The algorithm is fully described in \cite{camps2022fable}; here we just present the complexity: (i) $2^l$ one-qubit operations, (ii) $2^l$ C-NOTs. 
    
    \begin{figure}[h]
    \includegraphics[width=0.5\textwidth]{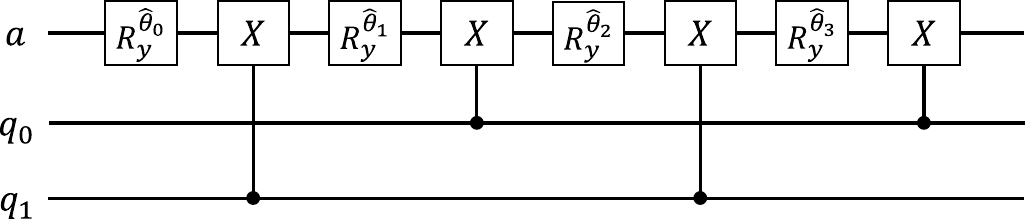}
    \caption{Fable implementation of the query oracle $\hat{O}_{p}^S$. The particular one-qubit rotations angles are according to the \cite{camps2022fable}.}
    \label{fig:toy_example_O_p_s}
    \end{figure}

    \section{Explicit state preparation circuit design for finite binary norm quantum state}\label{appendix Finite binary norm Ansatz}
    The circuit performs the following task. Assuming that the initial state of the circuit's $n$ qubits is $\ket{0}^n$, we aim to transform it into
    \begin{equation*}
    A_\psi\ket{0}^n=\ket{\psi}^n=\sum_{|i|_b\leq q} \beta_i \ket{i}^n
    \end{equation*}
    where the coefficients $\beta_i\neq 0(|i|_b\leq q)$ are all known. Let $m=\min \{q,n\}$,  
     We first design a circuit, as shown in \ref{fig:theta_circuit_app}. Now we will derive the rotation angles for each $R_y$ gate in the circuit.
    \begin{figure}[h!]
    \includegraphics[width=0.8\textwidth]{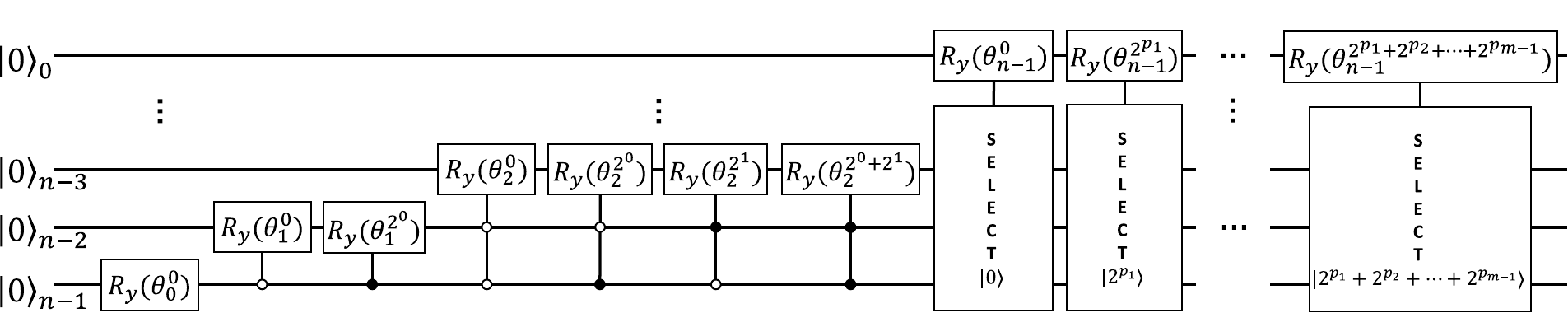}
    \caption{The initial design of the quantum circuit that implements $A_\psi$. Angles of $R_y$ gates are given by Eq.(\ref{eq:theta(n-1)(0)})-(\ref{thetapm}).}
    \label{fig:theta_circuit_app}
    \end{figure}

    
     By derivation, we know that the coefficients of the quantum state $\ket{\psi}^n$ satisfy the following relationship with the angles of the $R_y$ gates in the circuit.

    \begin{equation}
        \begin{aligned}
        \beta_0=&\prod_{j_0=0}^{n-1}\text{cos}\frac{\theta_{j_0}^0}{2},\\
       \beta_{2^{p_1}}=&\left(\prod\limits_{j_1=p_1+1}^{n-1}\text{cos}\frac{\theta_{j_1}^{2^{p_1}}}{2} \right)\text{sin}\frac{\theta_{p_1}^0}{2}\left(\prod_{j_0=0}^{p_1-1}\text{cos}\frac{\theta_{j_0}^0}{2}\right) \qquad (0\leq p_1 < n-1),\\
       \beta_{2^{n-1}}=&\text{sin}\frac{\theta_{n-1}^0}{2} \left( \prod_{j_0=0}^{n-2}\text{cos}\frac{\theta_{j_0}^0}{2}\right),
       \\
        \beta_{2^{p_1}+2^{p_2}}=&\left(\prod_{j_2=p_2+1}^{n-1}\text{cos}\frac{\theta_{j_2}^{2^{p_1}+2^{p_2}}}{2} \right)\text{sin}\frac{\theta_{p_2}^{2^{p_1}}}{2}\left(\prod_{j_1=p_1+1}^{p_2-1}\text{cos}\frac{\theta_{j_1}^{2^{p_1}}}{2} \right) \text{sin}\frac{\theta_{p_1}^{0}}{2}\left(\prod_{j_0=0}^{p_1-1}\text{cos}\frac{\theta_{j_0}^{0}}{2}\right)
        \\
        &(0\leq p_1 < p_2 < n-1),\\
        \beta_{2^{p_1}+2^{n-1}}=&\text{sin}\frac{\theta_{n-1}^{2^{p_1}}}{2}\left(\prod_{j_1=p_1+1}^{n-2}\text{cos}\frac{\theta_{j_1}^{2^{p_1}}}{2} \right) \text{sin}\frac{\theta_{p_1}^{0}}{2}\left(\prod_{j_0=0}^{p_1-1}\text{cos}\frac{\theta_{j_0}^{0}}{2}\right)\qquad(0\leq p_1< n-1),\\
        \vdots \\
        \beta_{2^{p_1}+\dots+2^{p_{m-2}}+2^{p_{m-1}}} =&\left(\prod_{j_{m-1}=p_{m-1}+1}^{n-1}\text{cos}\frac{\theta_{j_{m-1}}^{2^{p_1}+\dots+2^{p_{m-1}}}}{2} \right)\text{sin}\frac{\theta_{p_{m-1}}^{{2^{p_1}+\dots+2^{p_{m-2}}}}}{2}\left(\prod_{j_{m-2}={p_{m-2}+1}}^{p_{m-1}-1}\text{cos}\frac{\theta_{j_{m-2}}^{{2^{p_1}+\dots+2^{p_{m-2}}}}}{2} \right)\\
        & \text{sin}\frac{\theta_{p_{m-2}}^{{2^{p_1}+\dots+2^{p_{m-3}}}}}{2}\left(\prod_{j_{m-3}=p_{m-3}+1}^{p_{m-2}-1}\text{cos}\frac{\theta_{j_{n-3}}^{{2^{p_1}+\dots+2^{p_{m-3}}}}}{2} \right) \dots \left(\prod_{j_1=p_1+1}^{p_2-1}\text{cos}\frac{\theta_{j_1}^{2^{p_1}}}{2}\right) \\&\text{sin}\frac{\theta_{p_1}^{0}}{2}\left(\prod_{j_0=0}^{p_1-1}\text{cos}\frac{\theta_{j_0}^{0}}{2}\right)\qquad (0\leq p_1 < p_2<\dots <p_{m-1} < n-1),\\
        \beta_{2^{p_1}+\dots+2^{p_{m-2}}+2^{n-1}} =&\text{sin}\frac{\theta_{p_{n-1}}^{{2^{p_1}+\dots+2^{p_{m-2}}}}}{2}\left(\prod_{j_{m-2}={p_{m-2}+1}}^{n-2}\text{cos}\frac{\theta_{j_{m-2}}^{{2^{p_1}+\dots+2^{p_{m-2}}}}}{2} \right)\\
        & \text{sin}\frac{\theta_{p_{m-2}}^{{2^{p_1}+\dots+2^{p_{m-3}}}}}{2}\left(\prod_{j_{m-3}=p_{m-3}+1}^{p_{m-2}-1}\text{cos}\frac{\theta_{j_{n-3}}^{{2^{p_1}+\dots+2^{p_{m-3}}}}}{2} \right) \dots \left(\prod_{j_1=p_1+1}^{p_2-1}\text{cos}\frac{\theta_{j_1}^{2^{p_1}}}{2}\right) \\&\text{sin}\frac{\theta_{p_1}^{0}}{2}\left(\prod_{j_0=0}^{p_1-1}\text{cos}\frac{\theta_{j_0}^{0}}{2}\right)\qquad
        (0\leq p_1 < p_2<\dots <p_{m-2} < n-1),\\
       \beta_{2^{p_1}+\dots+2^{p_{m-1}}+2^{p_{m}}} =&\text{sin}\frac{\theta_{p_m}^{2^{p_1}+\dots+2^{p_{m-1}}}}{2} \left(\prod_{j_{m-1}=p_{m-1}+1}^{p_m-1}\text{cos}\frac{\theta_{j_{m-1}}^{2^{p_1}+\dots+2^{p_{m-1}}}}{2} \right)\text{sin}\frac{\theta_{p_{m-1}}^{{2^{p_1}+\dots+2^{p_{m-2}}}}}{2}\\&\left(\prod_{j_{m-2}={p_{m-2}+1}}^{p_{m-1}-1}\text{cos}\frac{\theta_{j_{m-2}}^{{2^{p_1}+\dots+2^{p_{m-2}}}}}{2} \right)
        \text{sin}\frac{\theta_{p_{m-2}}^{{2^{p_1}+\dots+2^{p_{m-3}}}}}{2}\left(\prod_{j_{m-3}=p_{m-3}+1}^{p_{m-2}-1}\text{cos}\frac{\theta_{j_{n-3}}^{{2^{p_1}+\dots+2^{p_{m-3}}}}}{2} \right) \\&\dots \left(\prod_{j_1=p_1+1}^{p_2-1}\text{cos}\frac{\theta_{j_1}^{2^{p_1}}}{2}\right) \text{sin}\frac{\theta_{p_1}^{0}}{2}\left(\prod_{j_0=0}^{p_1-1}\text{cos}\frac{\theta_{j_0}^{0}}{2}\right)\\&(0\leq p_1 < p_2<\dots <p_{m-1}<p_m\leq n-1).\\
    \end{aligned}\label{eq:nalltheta}
    \end{equation}

    Subsequently, we will deduce the rotation angles for all $R_y$ gates in the circuit from Equation (\ref{eq:nalltheta}). The process begins with determining $\theta_{n-1}^0$ using the ratio
    \begin{equation*}
        \frac{\beta_{2^{n-1}}}{\beta_{0}}=\text{tan}\frac{\theta_{n-1}^0}{2},
    \end{equation*}
    yielding
    \begin{equation}
        \theta_{n-1}^0=2\text{arctan}\left(\frac{\beta_{2^{n-1}}}{\beta_{0}}\right).
        \label{eq:theta(n-1)(0)}
    \end{equation}
    Further, we find angles whose lower index are all $n-1$
    \begin{equation}\label{eq:theta(n-1)}
        \theta_{n-1}^{2^{p_1}+\dots+2^{p_r}}=2\text{arctan}\left(\frac{\beta_{2^{p_1}+\dots+2^{p_r}+2^{n-1}}}{\beta_{2^{p_1}+\dots+2^{p_r}}}\right), (0\leq p_1 < \dots <p_r < n-1)
    \end{equation}
    where $r=1,2,...,m-1$.
    
    To compute the other $\theta$ values, for $0\leq k< n-1$ we derive
    \begin{equation*}
        \frac{\beta_{2^{k}}}{\beta_{0}}=\frac{\text{cos}\frac{\theta_{n-1}^{2^{k}}}{2}\dots \text{cos}\frac{\theta_{k+1}^{2^{k}}}{2}}{\text{cos}\frac{\theta_{n-1}^{0}}{2}\dots \text{cos}\frac{\theta_{k+1}^{0}}{2}}\text{tan}\frac{\theta_{k}^0}{2},
    \end{equation*}
    from which we conclude
    \begin{equation}\label{theta(p1)(0)}
        \theta_{k}^0=2 \text{arctan}\left(\frac{\text{cos}\frac{\theta_{n-1}^{0}}{2}\dots \text{cos}\frac{\theta_{k+1}^{0}}{2}}{\text{cos}\frac{\theta_{n-1}^{2^{k}}}{2}\dots \text{cos}\frac{\theta_{k+1}^{2^{k}}}{2}} \frac{\beta_{2^{k}}}{\beta_{0}}\right).
    \end{equation}
    Assume that $1\leq k < n-1$. For $r=1,\dots,\min\{k,m-2\}$ and $0\leq p_1<\dots<p_r < k$, we find 
    \begin{equation*}
        \frac{\beta_{2^{p_1}+\dots +2^{p_{r}}+2^{{k}}}}{\beta_{2^{p_1}+\dots +2^{p_{r}}}}=\frac{\text{cos}\frac{\theta_{n-1}^{2^{p_1}+\dots+2^{p_{r}}+2^{k}}}{2}\dots\text{cos}\frac{\theta_{k+1}^{2^{p_1}+\dots+2^{p_{r}}+2^{k}}}{2} }{\text{cos}\frac{\theta_{n-1}^{2^{p_1}+\dots+2^{p_{r}}}}{2}\dots\text{cos}\frac{\theta_{k+1}^{2^{p_1}+\dots+2^{p_{r}}}}{2}}\text{tan} \frac{\theta_{k}^{2^{p_1}+\dots+2^{p_{r}}}}{2},
    \end{equation*}
    leading to
    \begin{equation}\label{thetaothers}
        \theta_{k}^{2^{p_1}+\dots+2^{p_{r-1}}}=2\text{arctan}\left( \frac{\text{cos}\frac{\theta_{n-1}^{2^{p_1}+\dots+2^{p_{r}}}}{2}\dots\text{cos}\frac{\theta_{k+1}^{2^{p_1}+\dots+2^{p_{r}}}}{2}}{\text{cos}\frac{\theta_{n-1}^{2^{p_1}+\dots+2^{p_{r}}+2^{k}}}{2}\dots\text{cos}\frac{\theta_{k+1}^{2^{p_1}+\dots+2^{p_{r}}+2^{k}}}{2} }\frac{\beta_{2^{p_1}+\dots +2^{p_{r}}+2^{k}}}{\beta_{2^{p_1}+\dots +2^{p_{r}}}} \right).
    \end{equation}
    Additionally, if $k \geq m-1$, we find 
    \begin{equation*}
        \frac{\beta_{2^{p_1}+\dots +2^{p_{m-1}}+2^{k}}}{\beta_{2^{p_1}+\dots +2^{p_{m-1}}}}= \frac{1}{\text{cos}\frac{\theta_{n-1}^{2^{p_1}+\dots+2^{p_{m-1}}}}{2}\dots\text{cos}\frac{\theta_{k+1}^{2^{p_1}+\dots+2^{p_{m-1}}}}{2}} \text{tan} \frac{\theta_{k}^{2^{p_1}+\dots+2^{p_{m-1}}}}{2},
    \end{equation*}
    where $0\leq p_1<\dots<p_{m-1} < k$, which gives 
    \begin{equation}\label{thetapm}
        \theta_{k}^{2^{p_1}+\dots+2^{p_{m-1}}}=2\text{arctan} \left( \text{cos}\frac{\theta_{n-1}^{2^{p_1}+\dots+2^{p_{m-1}}}}{2}\dots\text{cos}\frac{\theta_{k+1}^{2^{p_1}+\dots+2^{p_{m-1}}}}{2} \frac{\beta_{2^{p_1}+\dots +2^{p_{m-1}}+2^{k}}}{\beta_{2^{p_1}+\dots +2^{p_{m-1}}}} \right).
    \end{equation}
    From Eq. (\ref{eq:theta(n-1)(0)}) and (\ref{eq:theta(n-1)}) we can get the angles whose lower index are $n-1$. And then From Eq. (\ref{theta(p1)(0)}), (\ref{thetaothers}) and (\ref{thetapm}) we get all the angles we needed in the circuit.
    
    \begin{figure}[h!]
    \includegraphics[width=0.8\textwidth]{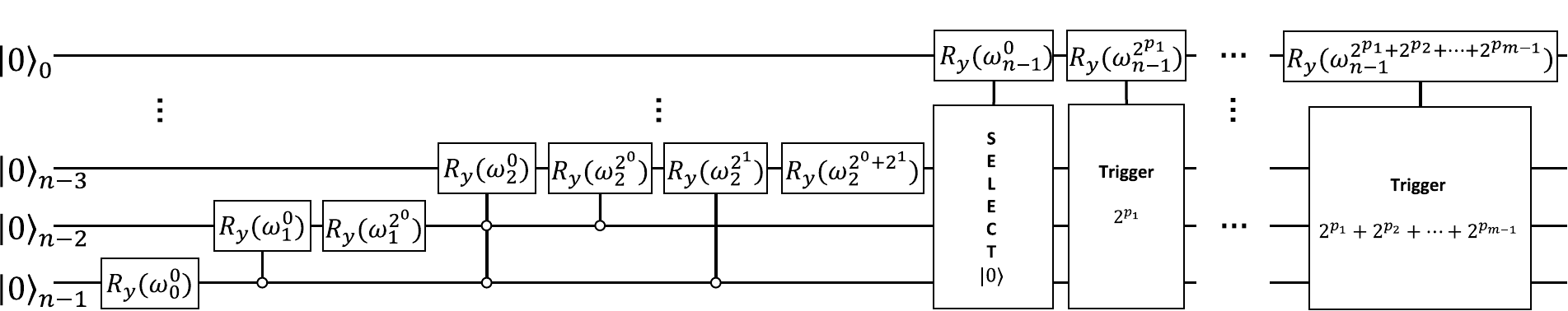}
    \caption{A more efficient design of the quantum circuit that implements $A_\psi$. Angles of $R_y$ gates are given by Eq.(\ref{eqn:omega1})-(\ref{eqn:omega4}). Here, the term "trigger $2^{p_1}+\dots+2^{p_r}$" refers to controlling only those qubits whose positions are represented by '0's in the binary representation of $2^{p_1}+\dots+2^{p_r}$}
    \label{fig:omega_circuit_app}
    \end{figure}

    Based on this circuit, we designed a more efficient circuit, see figure \ref{fig:omega_circuit_app}. It should only remove those "hollow dots" in the original circuit. Here we need to use a property of $R_y$ gate
    \begin{equation*}
    R_y(\varphi_1+\varphi_2)=R_y(\varphi_1)+R_y(\varphi_2).
    \end{equation*}
    Assuming $1\leq k < n$ and $1\leq r \leq \min\{k,m-1\}$ with $0\leq p_1<\dots<p_r<k$. When $r=\min\{k,m-1\}$ 
    \begin{equation}
    \theta_k^{2^{p_1}+\dots+2^{p_{\min\{k,m-1\}}}}=\omega_k^{2^{p_1}+\dots+2^{p_{\min\{k,m-1\}}}},
        \label{eq:ngamma1}
    \end{equation}
    and when $r< \min\{k,m-1\}$
    \begin{equation}\label{eq:ngamma2}
        \begin{aligned}
        \theta_k^{2^{p_1}+...+2^{p_r}}=\omega_k^{2^{p_1}+...+2^{p_r}}+\sum_{1\leq j \leq \min\{k,m-1\}-r}\sum\limits_{\substack{0\leq l_1<\dots<l_j < k,\\\{l_1,\dots,l_j\}\cap\{p_1,\dots,p_r\}=\emptyset}}
        \omega_k^{(2^{p_1}+...+2^{p_r})+(2^{l_1}+\dots+2^{l_j})},
        \end{aligned}
    \end{equation}
    and
    \begin{equation}\label{eq:ngamma3}
        \theta_0^0=\omega_0^0,
    \end{equation}
    \begin{equation}\label{eq:ngamma4}
        \theta_k^0=\sum\limits_{\substack{0\leq l \leq 2^k-1,\\ |l|_b\leq \min\{k,m-1\}}}
        \omega_k^l \quad (1\leq k< n).
    \end{equation}
    The equation (\ref{eq:ngamma4}) means $\theta_k^0$ is equal to the summation of all $\omega$s whose lower index are $k$. Here those $\theta$s are angles of corresponding $R_y$ gates in circuit \ref{fig:theta_circuit_app}.
    
    The $\theta$s are known from (\ref{eq:theta(n-1)(0)})-(\ref{thetapm}). The next step involves solving Equations (\ref{eq:ngamma1}) through (\ref{eq:ngamma3}) to extract $\omega$ values from the $\theta$ values. Specifically, from Equation (\ref{eq:ngamma1}), we have
    
    \begin{equation}
        \label{eqn:omega1}\omega_k^{2^{p_1}+\dots+2^{p_{\min\{k,m-1\}}}}=\theta_k^{2^{p_1}+\dots+2^{p_{\min\{k,m-1\}}}}\quad (1\leq k<n, 0\leq p_1<\dots<p_{\min\{k,m-1\}}<k).
    \end{equation} 
    Further, from Equation (\ref{eq:ngamma2}), we obtain
    \begin{equation}
       \omega_k^{2^{p_1}+...+2^{p_r}}=\theta_k^{2^{p_1}+...+2^{p_r}}-\sum_{1\leq j \leq \min\{k,m-1\}-r}\sum\limits_{\substack{0\leq l_1<\dots<l_j<k,\\\{l_1,\dots,l_j\}\cap\{p_1,\dots,p_r\}=\emptyset}}
       \omega_k^{(2^{p_1}+...+2^{p_r})+(2^{l_1}+\dots+2^{l_j})},
       \label{eqn:omega2}
    \end{equation}
    where $1\leq k < n$, $1\leq r \leq \min\{k,m-1\}$ and $0\leq p_1<\dots<p_r<k$.
    Additionally, from Equation (\ref{eq:ngamma3})
    \begin{equation}
        \omega_0^0=\theta_0^0.
        \label{eqn:omega3}
    \end{equation}
    And from Equation (\ref{eq:ngamma4}), we determine
    \begin{equation}
        \omega_k^0=\theta_k^0- \sum\limits_{\substack{1\leq l \leq 2^k-1,\\|l|_b\leq \min\{k,m-1\}}}
        \omega_k^l \quad (1\leq k< n).
        \label{eqn:omega4}
    \end{equation}
    Thus, we have successfully calculated all necessary $\omega$ values. In summary, we initially designed Circuit \ref{fig:theta_circuit_app} and subsequently simplified it to derive the more efficient Circuit \ref{fig:omega_circuit_app}. The rotation angles for the $R_y$ gates in Circuit \ref{fig:omega_circuit_app} are determined by Equations (\ref{eqn:omega1}) through (\ref{eqn:omega4}).

\end{document}